\documentclass[11pt]{article}

\usepackage[utf8]{inputenc}
\usepackage{lipsum}

\usepackage{amsmath, amssymb,amsfonts,amsthm,mathtools}
\usepackage{xspace,graphicx,relsize,bm,bbm,xcolor}
\usepackage{soul} % provides  s p a c i n g o u t, underlining and some derivatives such as overstriking and highlighting: https://ctan.mc1.root.project-creative.net/macros/generic/soul/soul.pdf

\usepackage{parskip}  % no indentation but some space between paragraphs

\usepackage{libertine}
\usepackage{libertinust1math}
\usepackage{dsfont}
\usepackage[T1]{fontenc}
\usepackage{enumitem}
\usepackage{nicefrac}

\usepackage[
    backend=biber,
    style=alphabetic,
    sorting=anyt,
    minalphanames=3,
    maxalphanames=3,
    maxnames=99,
    backref=true
    ]{biblatex}
\DefineBibliographyStrings{english}{%
  backrefpage = {page},% originally "cited on page"
  backrefpages = {pages},% originally "cited on pages"
}

\usepackage{sepfootnotes}
\newendnotes{x}
\renewcommand\xnotesize\normalsize

\usepackage{hyperref}
\usepackage[margin=1.75cm]{geometry}
\definecolor{linkcol}{rgb}{0.0,0.55,0.7}
\definecolor{citecol}{rgb}{0.0, 0.6, 0.45}
\definecolor{urlcol}{rgb}{0.7, 0.0, 0.55}
\hypersetup{
	colorlinks,
	linkcolor={linkcol},
	citecolor={citecol},
	urlcolor={urlcol}
}

\usepackage{url}
\usepackage{subcaption}
\usepackage{mleftright}
\usepackage{hyperref}
\usepackage{multirow}
\usepackage{physics}  % all braket notation at once

\usepackage{algorithm}
\usepackage{algpseudocodex}[indLines = true,italicComments = false]

\usepackage{zref-clever}
\zcsetup{cap = true}
\newcommand{\Cref}{\zcref}

\usepackage{authblk}  % for Authors

\def\01{\{0,1\}}
\newcommand{\mc}[1]{\mathcal{#1}}

\newcommand{\defeq}{:=}%{\coloneqq}

\let\Pr\relax
\DeclareMathOperator*{\Pr}{\mathbf{Pr}}

\newcommand{\eps}{\epsilon}

\newcommand{\swap}{\mathtt{SWAP}}

\newcommand{\epr}{\mathtt{EPR}}

\newcommand{\bfX}{\boldsymbol{X}}

\newcommand{\bigo}{\mathcal{O}}
\newcommand{\wg}{\mathrm{Wg}}
\newcommand{\qdchi}{\mathrm{D}_{\chi^2}}
\newcommand{\bftheta}{\boldsymbol{\theta}}
\newcommand{\mmstate}{\frac{I}{d}}
\newcommand{\bfz}{\boldsymbol{z}}
\newcommand{\bfU}{\boldsymbol{U}}
\newcommand{\bfrho}{\boldsymbol{\rho}}

\newcommand{\QChi}[2]{\mathrm{D}_{\chi^2}\left(#1 \| #2\right)}
\newcommand{\QChiAlpha}[3]{\mathrm{D}_{\chi^2_{#1}}\left(#2 \| #3\right)}
\newcommand{\Inner}[2]{\langle #1, #2 \rangle}
\newcommand{\InnerNamed}[3]{\langle #1, #2 \rangle_{#3}}
\newcommand{\supp}{\operatorname{supp}}
\newcommand{\vecmap}{\mathrm{vec}}

\newcommand{\etal}{\emph{et al.\@}}
\newcommand{\accept}{\textsf{accept}}
\newcommand{\reject}{\textsf{reject}}
\newcommand{\Cov}{\mathrm{Cov}}
\newcommand{\cyc}{\mathrm{cyc}}

\newtheoremstyle{mydefinitionsty}% 〈name〉
{10pt}% 〈Space above〉
{10pt}% 〈Space below〉
{}% 〈Body font〉
{}% 〈Indent amount〉1
{}% 〈Theorem head font〉
{}% 〈Punctuation after theorem head〉
{.5em}% 〈Space after theorem head〉2
{\textbf{\thmname{#1}~\thmnumber{#2}:  }\thmnote{(#3)}}% 〈Theorem head spec (can be left empty, meaning ‘normal’)〉

\newtheoremstyle{myproblemsty}% 〈name〉
{10pt}% 〈Space above〉
{10pt}% 〈Space below〉
{}% 〈Body font〉
{}% 〈Indent amount〉1
{}% 〈Theorem head font〉
{}% 〈Punctuation after theorem head〉
{.5em}% 〈Space after theorem head〉2
{\textbf{\thmname{#1}~\thmnumber{#2}:  }\thmnote{(#3)}\newline}% 〈Theorem head spec (can be left empty, meaning ‘normal’)〉

\newtheoremstyle{mythmsty}% 〈name〉
{10pt}% 〈Space above〉
{10pt}% 〈Space below〉
{\itshape}% 〈Body font〉
{}% 〈Indent amount〉1
{}% 〈Theorem head font〉
{}% 〈Punctuation after theorem head〉
{.5em}% 〈Space after theorem head〉2
{\textbf{\thmname{#1}~\thmnumber{#2}:  }\thmnote{(#3)}}% 〈Theorem head spec (can be left empty, meaning ‘normal’)〉
\theoremstyle{mythmsty}
\newtheorem{theorem}{Theorem}[section]
\newtheorem{proposition}[theorem]{Proposition}
\newtheorem{lemma}[theorem]{Lemma}
\newtheorem{corollary}[theorem]{Corollary}

\newtheorem{fact}[theorem]{Fact}

\AddToHook{env/theorem/begin}{%
\zcsetup{countertype={theorem=theorem}}}
\AddToHook{env/proposition/begin}{%
\zcsetup{countertype={theorem=proposition}}}
\AddToHook{env/lemma/begin}{%
\zcsetup{countertype={theorem=lemma}}}
\AddToHook{env/corollary/begin}{%
\zcsetup{countertype={theorem=corollary}}}
\AddToHook{env/claim/begin}{%
\zcsetup{countertype={theorem=claim}}}
\AddToHook{env/conjecture/begin}{%
\zcsetup{countertype={theorem=conjecture}}}
\AddToHook{env/fact/begin}{%
\zcsetup{countertype={theorem=fact}}}
\zcRefTypeSetup{fact}{
Name-sg = Fact ,
name-sg = fact ,
Name-pl = Facts ,
name-pl = facts ,
}
\AddToHook{env/question/begin}{%
\zcsetup{countertype={theorem=question}}}

\theoremstyle{mydefinitionsty}
\newtheorem{definition}[theorem]{Definition}

\newtheorem{example}[theorem]{Example}

\newtheorem{remark}[theorem]{Remark}
\AddToHook{env/definition/begin}{%
\zcsetup{countertype={theorem=definition}}}
\AddToHook{env/notation/begin}{%
\zcsetup{countertype={theorem=notation}}}
\AddToHook{env/example/begin}{%
\zcsetup{countertype={theorem=example}}}
\AddToHook{env/assumption/begin}{%
\zcsetup{countertype={theorem=assumption}}}
\AddToHook{env/remark/begin}{%
\zcsetup{countertype={theorem=remark}}}

\theoremstyle{myproblemsty}

\AddToHook{env/problem/begin}{%
\zcsetup{countertype={theorem=problem}}}

\numberwithin{equation}{section}
\definecolor{alexcolor}{rgb}{0.0, 0.47, 0.75}   % {0.27, 0.51, 0.71}  %

\definecolor{questioncolor}{rgb}{0.36, 0.54, 0.66}

\title{Distributed Quantum Property Testing with Quantum Carrier Pigeons}

\author[1]{Kenny Chen}
\author[2]{Mina Doosti}
\author[3,4,5]{Ryan Sweke}
\author[2]{Chirag Wadhwa}
\affil[1]{The University of Sydney}
\affil[2]{School of Informatics, University of Edinburgh}
\affil[3]{African Institute for Mathematical Sciences (AIMS), South Africa}
\affil[4]{Department of Mathematical Sciences, Stellenbosch University, Stellenbosch 7600, South Africa}
\affil[5]{National Institute for Theoretical and Computational Sciences (NITheCS), South Africa}
\date{}

\begin{document}
\maketitle

\begin{abstract}
    We introduce a framework for distributed quantum inference under communication constraints. In our model, $m$ distributed nodes each receive one copy of an unknown $d$-dimensional quantum state $\rho$, before communicating via a constrained one-way communication channel with a central node, which aims to infer some property of $\rho$. This framework generalizes the classical distributed inference framework introduced by Acharya, Canonne, and Tyagi [COLT 2019], by allowing quantum resources such as quantum communication and shared entanglement. 
    
    Within this setting, we focus on the fundamental problem of quantum state certification: Given a complete description of some state $\sigma$, decide whether $\rho=\sigma$ or $\|\rho-\sigma\|_1\geq \epsilon$. Additionally, we focus on the case of limited communication between distributed nodes and the central node: we assume each communication channel is limited to only $n_c$ bits and  $n_q$ qubits with $n_c + n_q \leq \log d$. When all nodes can make use of a shared source of randomness, we show that the copy complexity of distributed state certification is $\Theta(\nicefrac{d^2}{2^{n_q} 2^{n_c/2}\epsilon^2})$. We further demonstrate that shared randomness is necessary to achieve the above complexity, by proving an $\Omega(\nicefrac{d^3}{4^{n_q} 2^{n_c} \epsilon^2})$ lower bound in the \emph{private-coin} setting. Moreover, we develop a private-coin algorithm that matches this bound up to a $\sqrt{\log d}$ factor, showing this complexity is near-optimal. Together, our work establishes a general framework for distributed quantum inference with communication constraints and characterizes the complexity of distributed state certification with limited communication. 
\end{abstract}

\newpage

\tableofcontents
\newpage

\section{Introduction}
Recent years have seen a massive surge in the study of algorithms for testing and learning quantum states~\cite{gs007,anshu2023surveycomplexitylearningquantum}.  Among other things, this development has been spurred by the necessity of such algorithms for both the characterization of emerging quantum computational devices and foundational scientific applications of quantum computing. This line of work has resulted in tight characterizations of the resource requirements for a wide variety of such inference problems under many different assumptions on the algorithm, such as its ability to perform measurements coherently on multiple copies or to choose them adaptively.  However, the vast majority of prior work in this area has been in the \textit{centralized} setting, where the algorithm itself has access to all copies of the unknown state. 

In this work, we study the problem of learning and testing quantum states in the \textit{distributed} setting, where the copies of the unknown state(s) are distributed among multiple nodes, all of which communicate with a central node running a learning/testing algorithm. This setting has a variety of motivations. Firstly, as quantum communication networks develop, this setting is natural and allows us to characterize the potential and limitations of inference over such networks. Additionally, this problem is the natural quantum analogue of classical statistical inference in the distributed setting, whose theoretical foundations have proven essential for the development of large-scale federated learning protocols, in which data is distributed among multiple nodes or data centers.

As per the classical setting, when trying to define a concrete framework for the analysis of \textit{distributed} quantum property testing, one is immediately faced with a variety of choices, such as the following:
\begin{enumerate}
\item What is the nature of the communication channels between the distributed and central nodes? How many bits or qubits can be sent down each channel? Are the channels required to ensure some notion of privacy?
\item What shared resources are available to the distributed nodes? Do they have access to public randomness, or shared entangled states?
\item What communication is allowed between the distributed nodes?
\end{enumerate}
Most of these choices are not merely of theoretical interest; they capture real limitations in networks, especially when quantum resources are involved. Among them, communication constraints are particularly fundamental, since communication is often one of the main bottlenecks in quantum networks and distributed quantum protocols such as quantum secure multiparty computation and delegated quantum computing.

A variety of concrete models, each defined by different answers to the above questions, have been illustrated in Figure~\ref{fig:framework}. However, we can immediately make the following observations. Firstly, when \textit{unlimited quantum communication} is allowed between the distributed and central nodes, one immediately recovers the unconstrained centralized setting. Indeed, in this case, each distributed node can simply send their quantum state to the central node, which can then run any testing algorithm -- even one requiring adaptive or multi-copy coherent measurements. Secondly, in the case where the communication channels between distributed nodes and the central node are purely classical, but \textit{unlimited classical communication} is involved, then one recovers the centralized setting in which only \textit{single-copy} measurements are possible. Taken together, we have the following observation:
\begin{center}
\textit{The problem of distributed quantum inference is primarily interesting under communication constraints!}
\end{center} 
In light of this, our attention here is focused precisely on inference in this setting of limited communication. For concreteness, we focus primarily on the fundamental testing problem of \textit{quantum state certification} (see \Cref{def:quantum-state-certification}), the quantum analogue of \textit{distribution identity testing}~\cite{gs009}. The copy complexity of this problem is well-understood in the centralized setting \cite{o2015quantum,buadescu2019quantum}. With this in mind, we are concerned with the following concrete question:

\begin{center}
\textit{How do communication constraints impact the copy complexity of distributed state certification?}
\end{center}

\begin{figure}
    \centering
     \includegraphics{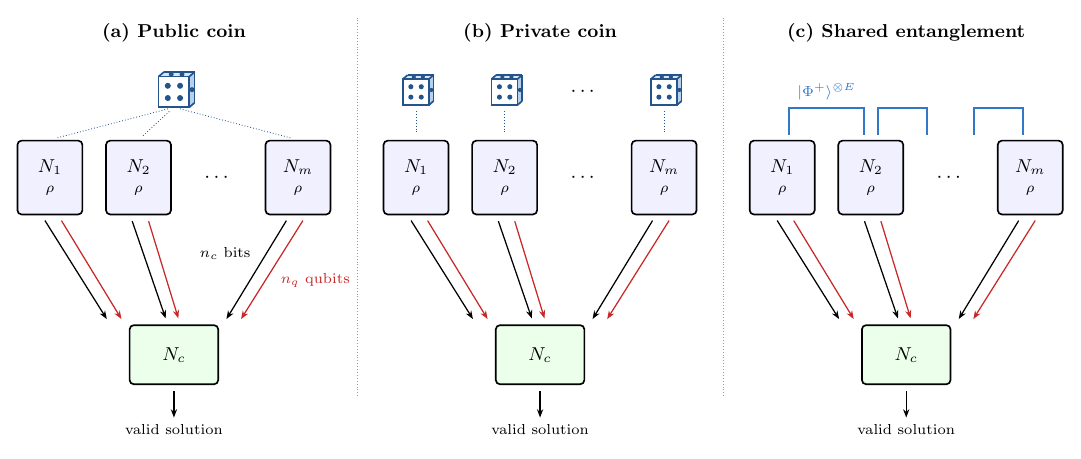} 
    \caption{An illustration of the $(n_c,n_q,R,E)$ model for distributed quantum inference, as per Definition~\ref{def:distributed-model}. Each distributed node $\{N_i\}_{i\in [m]}$ holds a single copy of $\rho$ and communicates with the central node $N_c$ via a communication channel limited to $n_c$ bits and $n_q$ qubits. The central node should output a valid solution to the inference problem (eg ``Accept'' or ``Reject'' in the case of property testing, or a valid hypothesis in the case of learning). We use $R\in \{\mathsf{public},\mathsf{private}\}$ to distinguish between the public and private coin setting (illustrated in panels (a) and (b) respectively), and $E$ to denote the number of Bell pairs shared between each neighbouring distributed node. In this work we focus on the setting with $E=0$.}  
    \label{fig:framework} 
\end{figure}

\subsection{Framework}\label{ss:framework}

We consider quantum inference problems defined by access to multiple copies of some unknown quantum state $\rho$. As illustrated in Figure~\ref{fig:framework}, we consider the distributed setting in which $m$ distributed nodes $\{N_i\,|\,i\in[m]\}$ each hold a single copy of the unknown state $\rho$. Each such node communicates with a central node $N_c$, which should output a candidate solution to the problem. 

As discussed previously, there are a variety of different choices one can make regarding the communication channels between the distributed and central nodes, the communication channels between the distributed nodes, and the resources shared by the distributed nodes. In the present work, we do not allow any communication between distributed nodes $N_i$ and allow only one-way communication between distributed nodes and the central node. With this in mind, we define the $(n_c, n_q, R, E)$ model for distributed quantum inference as follows:

\begin{definition}[$(n_c, n_q, R, E)$-model for distributed quantum inference]\label{def:distributed-model} We define the $(n_c, n_q, R, E)$ model as the model in which no communication is allowed between distributed nodes, only one-way communication is allowed from distributed nodes to the central node, and:
\begin{enumerate}
\item At most $n_c$ classical bits and $n_q$ qubits can be sent from any distributed node $N_i$ to the central node $N_c$.
\item $R\in \{\mathsf{public},\mathsf{private}\}$ indicates whether all nodes have access to a shared source of randomness or if only private randomness is available.
\item $E$ denotes the number of Bell pairs shared between each pair of neighboring nodes (assume the distributed nodes sit on the vertices of a graph). 
\end{enumerate}
\end{definition}

The models defined above should be viewed as a quantum generalization of the classical framework of distributed inference under information constraints introduced in~\cite{acharya2020inferenceinformationconstraints}. Their setting considers multiple players, each sending a single message to a central referee via a constrained communication channel, who then solves an inference task. In the absence of communication between players, and only one-way communication from players to referee, such as we consider here, this setting is known as the \emph{Simultaneous Message Passing (SMP)} model. When generalizing to the quantum setting though, where each distributed node may hold an unknown quantum state, there are several levels of generalization that can be made. First, there is the \emph{quantum communication} generalization, 
where the parties share a quantum communication channel which allows them to send quantum states, in addition to classical bits.

This setting will be the main focus of our work, but for completeness, we note that there is also the \emph{entanglement} generalization, where neighboring nodes are permitted to share entangled states with each other. Since our focus is on the former, for the remainder of this paper, we will omit the $E$ parameter, and for shorthand consider the $(n_c,n_q,R)$ model.

We say that a quantum inference problem for an unknown state $\rho$ can be solved with $m$ distributed nodes in the $(n_c, n_q, R)$ distributed model if there exist algorithms $\mathcal{A}_i$ for each distributed node $N_i$ and an algorithm $\mathcal{A}_c$ for the central node $N_c$, using the randomness model specified by parameter $R$, such that:

\begin{enumerate}
\item On input $\rho$, algorithm $\mathcal{A}_i$ outputs a quantum-classical message $m_i = (\tilde{m}_i,\phi_i)$, where $\tilde{m}_i\in\{0,1\}^{n_c}$ and $\phi_i$ is an $n_q$-qubit quantum state,
\item and on input $\{m_i\}_i$, the central algorithm $\mathcal{A}_c$ outputs a valid solution to the problem with sufficiently high probability.
\end{enumerate}
As our distributed model assumes that each distributed node holds one copy of the unknown quantum state, we define the sample complexity of solving a quantum inference problem in the $(n_c, n_q, R)$ distributed model as the minimum number of nodes for which the problem can be solved. We note that formalizing the above definition for a specific inference problem requires only specifying what constitutes a ``valid solution''. For instance, ``Accept'' or ``Reject'' in the case of property testing, an $\epsilon$-accurate hypothesis in the case of learning, or a sufficiently accurate estimate in the case of property estimation.

While any quantum inference problem can be studied in this framework, we focus specifically on the problem of quantum state certification~\cite{buadescu2019quantum}, defined as follows:

\begin{definition}[Quantum State Certification]\label{def:quantum-state-certification} Given $\epsilon \in (0,1]$, a complete classical description of a quantum state $\sigma$, and multiple copies of an unknown quantum state $\rho$, an algorithm is said to succeed at $\epsilon$-certification of $\sigma$ if it behaves as follows:
\begin{enumerate} 
\item If $\rho=\sigma$, output ``Accept" with probability at least $\frac23$.\footnote{Note that the choice of $\frac23$ for the success probability threshold is without loss of generality; one can use standard arguments to boost this to $1-\delta$ for arbitrary $\delta > 0$ with $\bigo(\log(1/\delta))$ repetitions.}
\item If $\|\rho-\sigma\|_1\geq \epsilon $, output ``Reject" with probability at least $\frac23$.
\end{enumerate}
\end{definition}

 State certification is the natural quantum analogue of distribution identity testing, sometimes referred to as ``testing goodness-of-fit''~\cite{gs009}. We note that above we have used the \textit{trace distance} to determine the reject instances, but one can also consider other distance measures.

With this in mind, the specific question we seek to answer is the following:
\begin{center}
\textit{What is the sample complexity of solving quantum state certification in the $(n_c,n_q,R)$ distributed setting?}
\end{center}
However, as noted previously, for $n$-qubit input states $\rho$ there are only certain parameter regimes where this problem is interesting. Notably:
\begin{enumerate}
\item When $n_q\geq n$, the $(n_c,n_q,R)$ model is equivalent to the centralized setting in which the certification algorithm is allowed to make multi-copy coherent measurements. In particular, each distributed node can just send their state $\rho$ to the central node, which can then execute a multi-copy coherent measurement. In this setting, the complexity of state certification is fully characterized~\cite{o2015quantum,buadescu2019quantum,odonnell2025instanceoptimalquantumstatecertification}.
\item When $n_q = 0$ but $n_c\geq n$, the $(n_c,n_q,R)$ model is equivalent to the centralized setting in which the certification algorithm is allowed to make non-adaptive single-copy measurements. In particular, in this case each distributed node can just perform the desired single-copy measurement, and send the outcome to the central node for classical post-processing. In this setting, the complexity of state certification is again fully characterized~\cite{bubeck2020entanglement,chen2022toward,chen2022tightStateCertification}.
\item When $n_q = 0$ but $n_c< n$, the $(n_c,n_q,R)$ model is equivalent to the centralized one in which the algorithm can only make non-adaptive single-copy measurements with a further restriction: each POVM can have at most $2^{n_c}$ outcomes. State certification under this restriction has also recently been characterized~\cite{liu2024quantum}.
\end{enumerate}
As a result of the above observations, the interesting regime is the one in which $0 \leq n_c < n$ and $0<n_q< n$. More specifically, we focus on the regime where $n_c + n_q \leq n$, i.e., the total communication from each node is limited.

\subsection{Our Results}

We summarize our results and those of relevant prior work in \Cref{tab:results}. We state our results for certification of $d$-dimensional states in the $(n_c,n_q, R)$ setting, where $R \in \{\mathsf{public},\mathsf{private}\}$. Our results hold generally for the case where each node can transmit classical-quantum messages of arbitrary dimensions $d_c \cdot d_q$; however, here, we only state them for the case of $n_c$ bits and $n_q$ qubits, i.e., $d_c = 2^{n_c}$ and $d_q = 2^{n_q}$. In \Cref{tab:results}, we additionally highlight the case where $n_c = 0$, i.e., only quantum communication is available, to isolate the effect of quantum communication.

We start by stating our bounds in the public-coin setting, where we provide a tight characterization of this problem.

\begin{theorem}[Public-Coin Complexity; see \Cref{thm:public_coin_ub,thm:public_coin_lb}]
    \label{thm:intro-public}
    Let $d \geq 6$, and $n_c,n_q \geq 0$ with $1 \leq n_c + n_q \leq \lfloor \log_2(d) \rfloor$, and $0 < \eps \leq \eps_0$ for a sufficiently small constant $\eps_0 < 1$. In the distributed $(n_c,n_q,\mathsf{public})$-setting, the complexity of $\eps$-certification of $d$-dimensional states is $\Theta\Bigl(\frac{d^2}{2^{n_q} 2^{n_c/2} \eps^2}\Bigr)$. 
\end{theorem}

Our result above shows that quantum communication provides a quadratic advantage over classical communication for distributed certification. Note that for $n_c = 0$ and $2^{n_q} \geq \Omega(d)$, i.e., when distributed nodes can transmit their full copies, we recover the centralized $\Theta(d/\epsilon^2)$ complexity~\cite{o2015quantum,buadescu2019quantum}. Moreover, setting $n_q = 0$, we also recover the $\Theta\Bigl(\frac{d^2}{2^{n_c/2} \eps^2}\Bigr)$ bounds for the classical communication setting implied by the work of Liu and Acharya~\cite{liu2024quantum}.

Next, we show that public coins are necessary for achieving this complexity, by proving a stronger lower bound in the private-coin setting; we also develop a private-coin algorithm that nearly matches this rate.

\begin{theorem}[Private-Coin Complexity; see \Cref{thm:private-coin-upper-general,thm:private_coin_lb}]
\label{thm:intro-private}
    Let $d \geq 6$, and $n_c,n_q \geq 0$ with $1 \leq n_c + n_q \leq \lfloor \log_2(d) \rfloor$, and $0 < \eps \leq \eps_0$ for a sufficiently small constant $\eps_0 < 1$.
    In the distributed $(n_c,n_q,\mathsf{private})$-setting, the complexity of $\eps$-certification of $d$-dimensional states is at least $\Omega\Bigl(\frac{d^3}{4^{n_q} 2^{n_c} \eps^2}\Bigr)$ and at most $\bigo\Bigl(\frac{d^3}{4^{n_q} 2^{n_c} \eps^2} \cdot \sqrt{\log d}\Bigr)$.
\end{theorem}

Our bounds imply a strong dimension-dependent gap between the public- and private-coin settings (for constant $n_c,n_q$, this is of order $d$), demonstrating that shared randomness is a powerful \emph{resource} in distributed learning and testing. In the resource-driven mindset of our framework, this highlights that the \emph{classical} resource of \emph{shared randomness} can be leveraged to reduce the use of another, often more expensive, resource: quantum communication. Notably, although the shared random seed is independent of the hypothesis state, it still leads to a substantial reduction in the quantum communication required for distributed certification. Note that for $n_q = 0$, we again recover the private-coin complexity implied by~\cite{liu2024quantum} (up to the $\sqrt{\log d}$ factor in the upper bound).

\begin{table}
    \centering
    \begin{tabular}{|c|c|c|c|c|}
    \hline
         & Distribution identity testing & \multicolumn{3}{|c|}{Quantum state certification} \\
    \hline
         & $2^{n_c} \leq d, n_q = 0$ & $2^{n_c} \leq d, n_q = 0$ & $n_c = 0, 2^{n_q} \leq d$ & $2^{n_c + n_q} \leq d$ \\
    \hline
    Public-coin & $\Theta\left(\frac{d}{\sqrt{2^{n_c}} \eps^2}\right)$  \cite{acharya2020inferenceinformationconstraints, ACT20b} & $\Theta\left(\frac{d^2}{\sqrt{2^{n_c}} \eps^2}\right)$ \cite{liu2024quantum} & $\Theta\left(\frac{d^2}{2^{n_q}\eps^2}\right)$ (Thm. \ref{thm:intro-public}) & $\Theta\left(\frac{d^2}{2^{n_q} 2^{n_c/2}\eps^2}\right)$ (Thm. \ref{thm:intro-public})\\
    \hline
    Private-coin & $\Theta\left(\frac{d^{3/2}}{2^{n_c} \eps^2}\right)$  \cite{acharya2020inferenceinformationconstraints, ACT20b} & $\Theta\left(\frac{d^3}{2^{n_c} \eps^2}\right)$ \cite{liu2024quantum} &  $\widetilde{\Theta}\left(\frac{d^3}{4^{n_q}\eps^2}\right)$ (Thm. \ref{thm:intro-private}) & $\widetilde{\Theta}\left(\frac{d^3}{4^{n_q} 2^{n_c}\eps^2}\right)$ (Thm. \ref{thm:intro-private}) \\
    \hline
    \end{tabular}
    \caption{Summary of our results and prior work in the distributed setting.}
    \label{tab:results}
\end{table}

\subsection{Technical Overview}\label{ss:technical-overview}

In what follows, unless specified otherwise, the public-coin setting refers to the $(n_c, n_q, \mathsf{public})$-setting, and similarly the private-coin setting refers to $(n_c, n_q, \mathsf{private})$ setting, for some $n_c, n_q \geq 0$ such that the total message size is bounded, i.e., $1 \leq n_c + n_q \leq \lfloor \log_2(d)\rfloor$. More generally, instead of an $n_c$-bit message and an $n_q$-qubit message, we assume that each distributed node sends a classical-quantum state (see \Cref{def:cq-states}) with a classical register of dimension $d_c$ and a quantum one of dimension $d_q$, for some $2 \leq d_cd_q \leq d$. Henceforth, we refer to such states as $(d_c,d_q)$-CQ states.

Now, in both settings and for $i \in [m]$, each node $N_i$ receives one copy of an unknown $d$-dimensional state $\rho$ and sends a $(d_c,d_q)$-CQ state to the central node $N_c$. Thus, without loss of generality, we can model the action of $N_i$ as a \emph{quantum instrument} (see \Cref{def:quantum_instrument}) $\mc{I}_i : \mathbb{C}^{d \times d} \mapsto \mathbb{C}^{d_q \times d_q} \otimes \mathbb{C}^{d_c \times d_c}$. In the private-coin setting, each node draws its instrument independently of the others. However, in the public-coin setting, there is some random string $\boldsymbol{r} \in \{0,1\}^*$ that is known to all distributed nodes $N_i$ and the central node, and each instrument $\mc{I}_i$ is parameterized by~$\boldsymbol{r}$.

Thus, we model the setting as follows: the central node $N_c$ receives states $\mathcal{I}_1(\rho), \dots, \mathcal{I}_m(\rho)$, and then chooses an appropriate algorithm to perform a test on these post-processed states. Note that we place no memory restrictions on the central node, so this central algorithm can perform fully coherent measurements on the post-processed states. This perspective will be crucial both in the design of our algorithms as well as in the proofs of our lower bounds.

The rest of this overview is organized as follows. We start by presenting the ideas behind our upper bounds. As a warmup, we first discuss the case when only quantum communication is available, i.e., $n_c = 0$ or $d_c = 1$; this is done in \Cref{sss:overview-warmup}. Next, in \Cref{sss:overview-upper}, we extend these ideas to the general case of both classical and quantum communication. Lastly, we present our lower bound ideas in \Cref{sss:overview-lower}.

\subsubsection{Testing without classical communication}\label{sss:overview-warmup}

In this section, we present the ideas behind our upper bounds with only quantum communication. So, we imagine that each distributed node sends a $d_q$-dimensional state to $N_c$, and this is equivalent to them applying a $d$-to-$d_q$-dimensional quantum channel, say, $\Phi_i$. Note that by the data processing inequality, no matter the choice of channels $\{\Phi_i\}_{i \in [m]}$, for all $\rho \neq \sigma$, we will have
\begin{equation}
    \|\Phi_i(\rho) - \Phi_i(\sigma)\|_1 \leq \|\rho - \sigma\|_1.
\end{equation}
In other words, distinguishing $\rho$ from $\sigma$ only becomes harder in this setting. However, if we could show quantitative lower bounds on the distance between the states $\Phi_i(\rho)$ and $\Phi_i(\sigma)$ for some channels $\Phi_i$, i.e., bounds of the form
\begin{equation}
    \|\Phi_i(\rho) - \Phi_i(\sigma)\|_1 \geq C \cdot \|\rho - \sigma\|_1, 
\end{equation}
for an appropriate factor $C < 1$, then the central node could apply the tight state certification algorithm of \cite{buadescu2019quantum} to test with precision $C \cdot \eps$.  Thus, we wish to find channels $\{\Phi_i\}_i$ that maximize this distance preservation factor $C$. Instead of the trace distance, it will be easier to prove this for the Hilbert--Schmidt distance. 

Such randomized dimensionality reduction maps have been studied previously in various settings. For the case of vectors in an $\ell_2$-space, this is the celebrated Johnson--Lindenstrauss lemma (see e.g., \cite[Section 5.3]{vershynin2018high}), and \cite{ACT20b} also showed such a result for compressing probability distributions while preserving their $\ell_2$-distance for distributed identity testing. In the quantum setting, such a result is also known for \emph{pure} states \cite{sen2018quantum}, and was used for distributed inner product estimation with limited quantum communication \cite{arunachalam2025generalized}. Note that all of the bounds mentioned here are achieved by conceptually simple random operations. For instance, the pure state compression bound of \cite{sen2018quantum} is achieved by projecting states into random subspaces.

Let us also note that for mixed states, Harrow, Montanaro, and Short \cite{harrow2015limitations} have already shown \emph{lower bounds} for distance-preservation with respect to the Hilbert--Schmidt and trace distances, as well as an upper bound for the trace distance. In particular, their $\|\cdot\|_2$-distance lower bound suggests that no algorithm can achieve a better factor than $C \leq \bigo\Bigl(\sqrt{\frac{d_q}{d}}\Bigr)$ with constant probability. We note that a generalization of the dimension-reduction map of \cite{arunachalam2025generalized} only achieves a distance-preservation factor of $C \geq \Omega(d_q/d)$ for mixed states with respect to the $\|\cdot\|_2$-distance. However, we will show that a remarkably simple channel achieves quadratically better distance-preservation factors, matching the lower bound of \cite{harrow2015limitations}. We will use the channel used by \cite{harrow2015limitations} for their trace-distance upper bound: simply apply a random unitary $\bfU$ to the state $\rho$, and trace out all but a $d_q$-dimensional subspace. In particular, letting $A,B$ be a fixed bipartition of $\mathbb{C}^{d \times d}$ with $A \cong \mathbb{C}^{d_q \times d_q}$, we define
\begin{equation}
    \Phi_{\bfU}(\rho) \triangleq \Tr_{B}(\bfU \rho \bfU^\dag).
\end{equation}
Then, we show that the above channel satisfies
\begin{equation}
    \|\Phi_{\bfU}(\rho) - \Phi_{\bfU}(\sigma)\|_2 \geq \frac12 \sqrt{\frac{d_q}{d}} \cdot \|\rho - \sigma\|_2, 
    \label{eq:overview-1}
\end{equation}
with at least constant probability. We prove this inequality by using Weingarten calculus to appropriately bound the second and fourth moments of $\|\Phi_{\bfU}(\rho) - \Phi_{\bfU}(\sigma)\|_2$, and then applying the Paley-Zygmund inequality. We state this result in \Cref{lem:mixed-state-compression-quantum} and expect it to be of independent interest.

Given \Cref{eq:overview-1}, our public-coin algorithm is now straightforward: the distributed nodes use their shared randomness to sample a random unitary $\bfU$, and each node sends the state $\Phi_{\bfU}(\rho)$ to the central node $N_c$. $N_c$ then applies the optimal Hilbert--Schmidt state certification algorithm (see \Cref{lem:hscertify}) to these copies for an appropriate precision parameter $\eps^\prime$. This allows one to succeed at certification with constant probability, which can be boosted to an arbitrary $1-\delta$ probability by repeating this entire operation across $\bigo(\log(1/\delta))$ batches of distributed nodes. The public-coin upper bound using this algorithm is stated in \Cref{thm:public-upper-bound-quantum-only}.

Now, for the private-coin algorithm, we will derandomize our dimension-reduction channels. In particular, we show the existence of a fixed set of $L = \bigo\left(\frac{d^2  \log d}{d_q^2}\right)$ unitaries $U_1, \dots, U_L$, such that for all pairs of states $\rho,\sigma$,
\begin{equation}
\label{eq:overview-derandomization}
    \frac{1}{L}\sum_{i = 1}^L \|\Phi_{U_i}(\rho-\sigma)\|_2^2 \geq \frac14 \cdot \frac{d_q}{d} \|\rho-\sigma\|_2^2.
\end{equation}
We prove the existence of these unitaries using the probabilistic method, i.e., we show that for $L$ Haar-random unitaries, the above equation holds with non-zero probability; our proof uses a matrix Chernoff bound to control extremal eigenvalues of sums of iid random matrices (see \Cref{thm:matrix_chernoff_inequality}). Then, we use the Hilbert--Schmidt estimator of \cite{buadescu2019quantum} (see \Cref{lem:hscertify}) in $L$ batches to estimate the LHS of Eq.~\eqref{eq:overview-derandomization}, and compare it to a suitable threshold. Finally, picking the right size for each batch yields our private-coin upper bound when $n_c = 0$ in \Cref{thm:private_coin_no_classical_ub}. 

\subsubsection{Testing with classical and quantum communication}\label{sss:overview-upper}
We now outline our extensions to the above ideas to obtain upper bounds with both classical and quantum communication. Recall that our nodes will now implement instruments $\mathcal{I}_1, \dots, \mathcal{I}_m$ that produce $(d_c,d_q)$-CQ states. We will again construct instruments with the best possible compression factors, and then have $N_c$ test the resulting states. We naturally generalize our quantum compression channels above to quantum instruments. First, split the $d$-dimensional input state across three registers $A,Q,C$ with dimensions $d_A,d_q,d_C$ respectively, such that $d = d_A d_q d_c$. For a Haar-random unitary $\bfU$, we define the instrument $\mathcal{I}_{\bfU}$ by

\begin{equation}
    \mathcal{I}_{\bfU}(\rho) = \sum_{c = 1}^{d_c} I_{\bfU,c}(\rho) \otimes \ketbra{c}{c}, \quad \textnormal{where} \quad I_{\bfU,c}(\rho) = \Tr_{A,C} [\bfU \rho \bfU^\dag \cdot I_{QA} \otimes \ketbra{c}{c}];
\end{equation}
in other words, our instrument applies a random unitary, measures the classical register (without observing the outcome), and then traces out the extra register $A$. In \Cref{thm:Haar_compression}, we show that this instrument satisfies a similar compression property: with constant probability, we have
\begin{equation}
    \|I_{\bfU}(\rho-\sigma)\|_2^2 \geq \frac{d_q}{4d} \|\rho - \sigma\|_2^2.
\end{equation}
However, this compression factor does not depend on $d_c$ and is thus not sufficient to obtain the $1/\sqrt{d_c}$ factor in our public-coin upper bound. To obtain this saving, we improve on the next step, i.e., the Hilbert--Schmidt test. Rather than directly passing to the generic \cite{buadescu2019quantum} estimator, we tailor their estimator and its analysis to classical-quantum states. In \Cref{lemma:cq_hs_estimator_variance}, we show that the variance of this new estimator can be much lower for certain kinds of classical-quantum states. In particular, we require that the marginal classical distributions are well-spread, i.e., that the probability of observing any classical label when measuring the CQ state is at most $\bigo(d_c^{-1})$. We show that the output states from the instruments defined above satisfy this with high probability, implying that the variance is indeed reduced by $\mathrm{poly}(d_c)$ factors; this finally allows us to obtain our main public-coin upper bound in \Cref{thm:intro-public}.

Next, to extend this to the private-coin case, we will derandomize the instruments constructed above. Note that the probabilistic properties of our instruments were the compression guarantee and the fact that the classical label properties were well spread out. We use the probabilistic method to show that a fixed set of $\bigo\left(\frac{d^2 \log d}{d_q^2 d_c}\right)$ unitaries satisfies both properties simultaneously in \Cref{thm:derandomization-classical-quantum}; our proof again makes use of the matrix Chernoff bound once. Note that the number of unitaries here is smaller than that in \Cref{sss:overview-warmup} by a factor of $d_c$; this, and the reduced classical-quantum variance, are the savings necessary to obtain our final private-coin upper bound in \Cref{thm:intro-private}.

\subsubsection{Lower bound ideas}\label{sss:overview-lower}

We will now provide an overview of our lower bound techniques, and refer to \Cref{s:lower_bounds} for further details. Note that lower bounds for state certification are typically shown by proving the hardness of distinguishing between the maximally mixed state and a specifically designed distribution over states that are $\eps$-far from it. Clearly, any algorithm for state certification must be able to solve this point-versus-mixture distinguishing task, and so a lower bound for the latter task implies one for certification. Similar ideas are also used to prove lower bounds for distribution identity testing, where one considers distinguishing between the uniform distribution and a mixture over distributions that are $\eps$-far from it. Before moving on to our techniques, let us describe prior lower bound techniques used by \cite{acharya2020inferenceinformationconstraints,liu2024quantum} to prove such lower bounds in the setting of classical communication.

\paragraph{Prior lower bounds for classical communication:} 
For both distribution testing and quantum state testing with limited classical communication, the central node $N_c$ receives a series of classical messages from the distributed nodes. To prove lower bounds for the distinguishing tasks above, it then suffices to show that the induced distributions over received messages are statistically indistinguishable across the two cases unless the number of nodes $m$ is sufficiently large. For any fixed input, as the distributed nodes cannot communicate with each other, the distribution over messages is easily shown to be a product distribution. Thus, to prove lower bounds for the point-versus-mixture task outlined above, one aims to show that a product distribution and a specific mixture of product distributions are statistically close. For such distributions, the Ingster--Suslina method (see e.g., \cite[Lemma 3.1]{CanonneTopicsDT2022}) allows one to easily upper bound the $\chi^2$-divergence between them. Indeed, \cite{acharya2020inferenceinformationconstraints,liu2024quantum} both use this method to upper bound the $\chi^2$-divergence between the induced distributions as a function of the number of nodes $m$ and the communication channels implemented by each node. Naturally, this divergence must be large for these distributions to be statistically distinguishable. Together, these upper and lower bounds on the $\chi^2$-divergence yield a lower bound on $m$.

While the above arguments allow \cite{acharya2020inferenceinformationconstraints,liu2024quantum} to obtain tight lower bounds in the public-coin setting, they do not immediately yield stronger lower bounds in the private-coin setting. To extend to the private-coin setting, Acharya, Canonne, and Tyagi \cite{acharya2020inferenceinformationconstraints} essentially showed that, here, one can imagine that the mixture of alternatives was chosen adversarially to minimize the $\chi^2$-divergence \emph{conditioned} on a fixed choice of channels implemented by the distributed nodes. This allows them to prove tighter upper bounds on the $\chi^2$-divergence, in turn leading to stronger lower bounds in the private-coin setting.

\paragraph{Lifting to quantum and classical communication:}

Recall that in our setting, on input state $\rho$, the central node $N_c$
receives classical-quantum messages $\mathcal{I}_1(\rho), \dots, \mathcal{I}_m(\rho)$, for some instruments $\mathcal{I}_i$ implemented by the distributed nodes. Thus, to solve the distinguishing task, the central node must be able to distinguish between the two states $\bigotimes_{i = 1}^m \mathcal{I}_i\Bigl(I/d\Bigr)$ and $\mathbb{E}_{\bfrho \sim D}  [ \bigotimes_{i = 1}^m \mathcal{I}_i(\bfrho)]$, for all distributions $D$ over states that are $\epsilon$-far from $I/d$. In analogy with the setting of classical communication, we aim to show that the former product state and the latter mixture over product states are statistically indistinguishable, for at least one specifically designed distribution $D$, unless $m$ is large. To prove this, we will make use of the recently developed quantum analogue of the Ingster--Suslina method \cite{odonnell2025instanceoptimalquantumstatecertification}, allowing us to upper bound the \emph{quantum} $\chi^2$-divergence between such states. While \cite{odonnell2025instanceoptimalquantumstatecertification} instantiated this method for a specific quantum $\chi^2$-divergence, we generalize this to a whole class of divergences (see \Cref{lem:quantum_ingster-suslina_gen}). Although the proof is near-identical to their original lemma, this extension is necessary for our lower bounds, as we discuss below. Specifically, we will use the divergence $\mathrm{D}_{\chi^2}(\rho \| \sigma) = \Tr(\sigma^{-1/2} (\rho-\sigma) \sigma^{-1/2} (\rho-\sigma))$. It is also easy to show that for the two states to be statistically distinguishable, this quantum $\chi^2$ divergence must be larger than a constant. Again, these upper and lower bounds together imply a lower bound on the number of distributed nodes.

Further, in the private-coin setting, we lift the arguments of \cite{acharya2020inferenceinformationconstraints} to show that any successful protocol satisfies
\begin{equation}
    \max_{\mathcal{I}_1, \dots, \mathcal{I}_m} \min_{D} \qdchi\left(\mathbb{E}_{\rho \sim D}  \left[ \bigotimes_{i = 1}^m \mathcal{I}_i(\rho)\right] \Bigg\| \bigotimes_{i = 1}^m \mathcal{I}_i\Bigl(\mmstate\Bigr)\right) \geq \frac{1}{16},
\end{equation}
where we minimize over all distributions $D$ over states that are $\eps$-far from the maximally mixed state. In other words, in correspondence with the classical setting, we can imagine the mixture of alternatives is picked adversarially to minimize the quantum $\chi^2$-divergence \emph{after} the channels implemented by the nodes have been chosen. This will later allow us to prove our improved lower bound in the private-coin setting.

\paragraph{Hard mixture of alternatives:} We will use the mixture of alternatives previously introduced by Liu and Acharya \cite{liu2024role} and later used by them in the setting of limited classical communication~\cite{liu2024quantum}. In particular, they pick $\ell$ orthonormal traceless matrices $V_1, \dots, V_\ell$ in $\mathbb{C}^{d \times d}$, and randomly perturb the maximally mixed state along these directions. In particular, for $\bfz \sim \{-1,+1\}^\ell$, a vector consisting of $\ell$ random Rademacher variables, they define
\begin{equation}
    \rho_{\bfz} \approx \mmstate + \frac{c\eps}{\sqrt{d\ell}} \sum_{i = 1}^\ell V_i z_i.
\end{equation}
Applying our generalized quantum Ingster--Suslina method (\Cref{lem:quantum_ingster-suslina_gen}) to this hard instance, we wish to bound the moment generating function of the quantity $Z_{\mathcal{I}}(\bfz,\bfz^\prime)$ with respect to $\bfz,\bfz^\prime \sim \{-1,+1\}^\ell$, where we have
\begin{equation}
    Z_{\mathcal{I}}(\bfz,\bfz^\prime) = \Tr\left(\mathcal{I}(I/d)^{-1/2} \mathcal{I}(\Delta_{\bfz}) \mathcal{I}(I/d)^{-1/2}  \mathcal{I}(\Delta_{\bfz^\prime})\right),
\end{equation}
and $\Delta_z = \rho_z - \mmstate$. Defining the completely positive map $\Lambda : X \mapsto \mathcal{I}(I/d)^{-1/4} \mathcal{I}(X) \mathcal{I}(I/d)^{-1/4}$, we can rewrite the above as $Z_{\mathcal{I}}(\bfz,\bfz^\prime) = \langle \Lambda(\Delta_{\bfz}) , \Lambda(\Delta_{\bfz^\prime}) \rangle$. It is crucial to our analysis that this object has such a ``symmetric'' form. If we were to instead use the divergence in \cite{odonnell2025instanceoptimalquantumstatecertification}, we would have to handle an unmanageable $\mathcal{I}(I)^{-1}$ term, which is why we need the divergence mentioned above. 

Now, via vectorization, $Z_{\mathcal{I}}(\bfz,\bfz^\prime)$ can be rewritten as a quadratic polynomial in the entries of the random Rademacher vectors $\bfz,\bfz^\prime$. To bound its moment generating function, we use a standard lemma for mgfs of such quadratic forms (see \Cref{lem:exponential_expectation_bound}), allowing us to obtain a lower bound on $m$ that depends on the basis $V_1, \dots, V_\ell$ and the instruments $\mathcal{I}_1, \dots, \mathcal{I}_m$. Carefully controlling the quantities arising here yields our main lower bounds.

We note that for $d_c = 1$ and $d_q = d$, our lower bounds recover the centralized $\Omega(d/\eps^2)$ lower bound for state certification~\cite{o2015quantum}. Moreover, in \Cref{s:centralized-mixedness-testing-bound}, we also apply the quantum Ingster--Suslina method of \cite{odonnell2025instanceoptimalquantumstatecertification} to the hard mixture of \cite{acharya2020inferenceinformationconstraints}, obtaining an alternate self-contained proof of this $\Omega(d/\eps^2)$ lower bound. 

\subsection{Prior Work}

There is a wide array of prior work, both classical and quantum, which is relevant to this work. 

\textbf{Classical distributed statistical inference with communication constraints:} There is a rich history of prior work on classical distributed statistical inference with communication constraints, in a wide variety of communication models, and with differing objectives. As a comprehensive overview is not possible, we provide here a representative selection of such works. Initial work on this topic~\cite{AhlswedeCsiszar1986, Han1987, HanAmari1998} studied the setting in which multiple nodes receive a stream of incomplete samples from a single source, and can communicate at some constrained rate with a decision node (potentially interactively~\cite{XiangKim2013}) who should solve an inference problem with respect to the sample source, with the goal of characterizing asymptotic error rates as a function of communication rate.

Branching off from this initial direction, a variety of works studied the setting in which some fixed number of nodes $m$ receive a fixed number of samples $n$ from a source, and should collectively solve an inference problem under some communication constraint~\cite{ZhangDuchiJordanWainwright2013, GargMaNguyen2014, BravermanGargMaNguyenWoodruff2016,DiakonikolasGrigorescuLiNatarajanOnakSchmidt2017,DiakonikolasGouleakisKaneRao2019,HanMukherjeeOzgurWeissman2018}. However, many of these works (a) allowed for a blackboard model of communication in which all messages are public and/or (b) were concerned with constraints on the \textit{total} amount of communication. The classical version of the model we consider -- in which each node holds only a single sample, only one-way communication is allowed between distributed nodes and the decision node, and the limitation is on the capacity of the channels from distributed nodes -- was initially studied in~\cite{acharya2020inferenceinformationconstraints, han2018geometric}. As mentioned, the most relevant to our work is Ref.~\cite{acharya2020inferenceinformationconstraints}, which introduced both an abstract SMP model for distributed inference under communication constraints, with each sample held by a different node, and provided a methodology for proving lower bounds in such a model. Indeed, as mentioned in Section~\ref{ss:framework}, our $(n_c,n_q, R,E)$ model is a natural quantum generalization of the model introduced in~\cite{acharya2020inferenceinformationconstraints}, and as discussed in Section~\ref{ss:technical-overview}, our lower bounds are obtained by techniques directly inspired by the ones introduced in this work. 

Following the introduction of the multi-node constrained SMP framework in Ref.~\cite{acharya2020inferenceinformationconstraints}, multiple works then provided specific techniques and algorithms for learning and testing distributions in this model~\cite{ACT19,ACT20b,ACHST20, ACT20c}, in the process characterizing the interplay between public/private randomness and information constraints. Simultaneously, Ref.~\cite{acharya2021inferenceinformationconstraintsiii} considered the setting of \textit{privacy-preserving} communication channels (building on~\cite{duchi2013local,sheffet2018locally}), and gave algorithms for learning and testing in this setting. Interactivity was then introduced into the model and studied for learning and testing discrete distributions~\cite{ACLST22} and high-dimensional continuous parameteric distributions~\cite{ACST23}, and for non-parameteric density estimation~\cite{ACST24}. Additionally, more recent work has again considered the setting of multiple samples per node~\cite{ACLST21,Vuursteen2024}, which is a natural abstract model for federated learning. Finally, we note that recent work has also studied \textit{hypothesis testing} in the distributed model with communication constraints~\cite{PJL24,PAJL23,PAJL25,KPJ25}.

\textbf{Quantum learning and testing:} Recent years have witnessed a large amount of work on the problem of learning and testing quantum states, processes and systems, in a wide variety of learning models, under the broad umbrella of \textit{quantum learning theory}. Providing a complete survey of such work is not possible here, but we refer to reviews on the complexity of learning quantum states~\cite{anshu2023surveycomplexitylearningquantum} and quantum property testing~\cite{gs007}, as well as the recently established \textit{quantum learning theory zoo}~\cite{QuantumLearningTheoryZoo}, which provides an up-to date repository of work on quantum learning theory. We note that the sub-field of quantum learning theory most relevant to our work is that of learning and testing \textit{quantum states} from multiple copies. Indeed, while we focus primarily on quantum state certification, any problem of this type could be immediately studied in the distributed model we introduce here. 

\textbf{Centralized quantum state certification:} In the centralized setting, initial work on quantum state certification focused on characterizing the worst-case complexity of the task both with~\cite{o2015quantum,buadescu2019quantum} and without~\cite{bubeck2020entanglement} the ability to perform coherent measurements on multiple copies of the unknown quantum state. More recent work has provided \textit{instance-optimal} bounds in both settings~\cite{chen2022toward,chen2022tightStateCertification, odonnell2025instanceoptimalquantumstatecertification} and shown that without coherent measurements, adaptivity provides no advantage for this problem. Additionally, the problem of state certification with incoherent measurements restricted to a certain number of outcomes was recently studied in Ref.~\cite{liu2024quantum}. As mentioned earlier, this setting is equivalent to the $(n_c\leq n,n_q=0,R)$ version of the model we study here, and recovers the worst-case results of~\cite{bubeck2020entanglement,chen2022toward} when $n_c=n$.  Finally, recent work has studied the role of \textit{shared randomness} in quantum state certification with incoherent measurements~\cite{liu2024role}, contrasting the complexity of this task in the public and private coin settings, as well as the case of \textit{non-iid} samples~\cite{depalma2025noniidhypothesistestingclassical}.

\textbf{Distributed inner product estimation (DIPE):} Given copies of two quantum states $\rho$ and $\sigma$, the problem of estimating $\Tr(\rho\sigma)$ is known as inner product estimation, and provides a fundamental quantum algorithmic primitive and one quantum generalization of distribution closeness testing. In the centralized setting (with multi-copy coherent measurements), the SWAP test provides an algorithm with only constant sample complexity. Motivated by similar reasoning as us, recent work has characterized the sample complexity of  \emph{distributed} inner product estimation -- in this setting, one party holds copies of $\rho$ and the other party holds copies of $\sigma$ -- with both only classical communication~\cite{anshu2022distributed} and limited quantum communication~\cite{gong2024samplecomplexitypurityinner,arunachalam2025generalized} allowed between parties. In the distributed setting, this problem has also been called \textit{cross device verification} for its applications in the verification of quantum devices~\cite{EfficientDIPE2}. More recent works have also focused on characterizing the class of quantum states for which \textit{computationally} efficient DIPE is possible~\cite{EfficientDIPE1,EfficientDIPE2}. This line of work is very similar in spirit to the work on distributed state certification we initiate here, with the primary differences being (a) the choice of testing problem, (b) the fact that the input to DIPE is multiple copies of \textit{two} unknown states, and (c) the fact that the unknown states are distributed among only two nodes as opposed to $m$ distributed nodes. Additionally,~\cite{arunachalam2025generalized} only consider pure states. We note that these prior results do not immediately imply bounds for these problems in our model.

\subsection{Discussion}\label{ss:discussion}

In this work, we have introduced a new model for distributed quantum inference by generalizing the classical setup of \cite{acharya2020inferenceinformationconstraints} and have studied the problem of state certification in the public- and private-coin settings, with both limited classical and quantum communication. We have given algorithms for both the public coin and private coin setting, the former of which is tight, and the latter of which is tight up to a $\sqrt{\log d}$ factor. This factor arises from the matrix Chernoff bound, and it is unclear if removing this requires a refined analysis or an entirely new algorithm.

Apart from these immediate open questions, there is a vast array of future directions one can consider within our distributed inference framework. We have only considered the problem of state certification in this work. However, in the analogous classical setting, the complexity of distribution learning has also been characterized. Similarly, in our setting, could we fully characterize the complexity of state tomography? Next, though we do not consider it, the framework we introduce also allows the distributed nodes to share entanglement. In \Cref{s:entanglement}, we demonstrate that a class of problems, including purity testing~\cite{chen2022exponential} and unsigned Pauli shadow tomography~\cite{chen2024optimalA}, require exponentially many copies in our setting even with unbounded classical communication, but which can be solved with a constant number of copies if shared entanglement is permitted. This positions entanglement as a powerful resource in our distributed framework, and it would be interesting to
understand its utility for other inference tasks.

Moreover, motivated by classical follow-ups to the work of \cite{acharya2020inferenceinformationconstraints}, one can consider many interesting extensions to our framework. We have only considered the simultaneous message-passing model, where distributed nodes are only allowed one-way communication to the central node. It would be interesting to model more general settings by understanding the effects of two-way communication and inter-node communication. Already, we note that allowing $n$-qubits of quantum communication among  $m'$ distributed nodes is essentially equivalent to grouping those nodes into one and allowing an $m'$-copy measurement. This allows one to recover the setting of $m'$-copy memory in the distributed nodes, for any $m'$. Similarly, instead of allowing only one copy per distributed node, we could imagine that each node obtains multiple copies of the unknown state. For our public-coin results, we have assumed that the nodes share an unbounded amount of randomness. From a practical perspective, an important direction of study would thus be to quantify the effects of limited shared randomness on the complexity of state certification.
Along these lines,~\cite{canonne2026distributedgaussianmeantesting} recently considered the problem of distributed, communication-constrained Gaussian mean estimation in the classical setting where the distributed nodes share a bounded number of random bits, each distributed node holds a different number of samples, and each user can send a different number of bits to the central node. As we move towards future deployments of such distributed quantum settings, it would also be important to understand the tradeoffs imposed by requiring the use of privacy-preserving algorithms or in the presence of \textit{untrusted} parties.

Lastly, we note that we have only discussed the study of \emph{one-state} inference problems. More generally, one can imagine that the nodes receive copies of distinct states. For instance, if $m/2$ distributed nodes receive copies of a state $\rho$ and the other $m/2$ nodes receive copies of another state $\sigma$, one can now consider testing for properties of both states. This setting would allow modeling the problem of DIPE discussed above, as well as the generalization of state certification to closeness testing, i.e., the problem of testing whether two \emph{unknown} states $\rho,\sigma$ are identical or $\epsilon$-far. We remark that our algorithm for state certification in \Cref{thm:intro-public} can be made to work for closeness testing with just a small amount of additional work. For DIPE, it is unclear whether our channels/instruments or those considered in prior work~\cite{arunachalam2025generalized,gong2024samplecomplexitypurityinner} can achieve desirable rates here. Now, generalizing this two-state setting further, one can imagine that \emph{each} node receives a distinct state, i.e., we have states $\rho_1, \dots, \rho_m$. Here, we may consider testing properties of the average of these $m$ states, in the spirit of the non-iid setting of ~\cite{garg2023testing,depalma2025noniidhypothesistestingclassical}.

\subsection*{Organization}
We start by presenting preliminary background and necessary notation in \Cref{s:preliminaries}. We then give our algorithms for both the public coin and private coin case in \Cref{s:public_ub} and \Cref{s:private_ub} respectively. In each section we start with a warmup algorithm that considers the restricted setting where only quantum communication is permitted. Though these can be recovered from our general results by setting $n_c=0$, they provide intuition for the main ideas used in the general case and are thus worth isolating out. These are presented in \Cref{ss:quantum_public_ub} and \Cref{ss:quantum_private_ub} for the public and private coin model respectively. Our main upper bounds are then proven in \Cref{ss:public_ub} for the public-coin setting, and \Cref{ss:quantum_classical_private_ub} for the private-coin one. We finally prove our complementary lower bounds in \Cref{s:lower_bounds}, which are tight in the public-coin model and tight up to log factors in the private-coin one.

% \paragraph{AI Use:}
% % Anonymized version:
% Preliminary versions of \Cref{lem:instrument_infty_bound,lem:instrument_2_bound} were proven with the assistance of Claude Opus 4.6. 

\subsection*{Acknowledgments}
This work subsumes the original paper by MD, RS and CW in \cite{doosti2026distributed}, and the followup paper building on their work by KC in \cite{chen2026distributed}. MD, RS and CW thank Matthias Caro for helpful discussions on quadratic forms of Rademacher random variables. KC would like to thank Cl\'{e}ment Canonne for helpful discussions as well as their comments on their prior preprint. MD, RS and CW acknowledge the use of Claude Opus 4.6 to aid with proofs of preliminary lemmas, which were then generalized by KC into \Cref{lem:instrument_infty_bound} and \Cref{lem:instrument_2_bound}.

MD and CW acknowledge the support of the Quantum Advantage Pathfinder (QAP), with grant reference EP/X026167/1, and the UK Engineering and Physical Sciences Research Council. MD also acknowledges the support of the Quantum Advantage TurboCHarger (QATCH) Programme. RS is grateful to the Alexander von Humboldt Foundation for support under the German Research Chair program at the African Institutes for Mathematical Sciences. Part of this work was carried out while MD and CW visited the African Institute for Mathematical Sciences, Cape Town for the 1st AIMS Workshop on the Theory of Quantum Learning Algorithms (2025). 

\section{Preliminaries}\label{s:preliminaries}
    Throughout, we use standard asymptotic notation, as well as $\tilde{\bigo}(\cdot)$ and $\tilde{\Omega}(\cdot)$ to hide polylogarithmic factors in the argument. We also let $d_q=2^{n_q}$ and $d_c=2^{n_c}$. We start with some key definitions, used throughout the paper.
\begin{definition}[Schatten $p$-norms]
    For matrices/linear operators $T$ and $p\in[1,\infty]$, the \emph{Schatten} $p$-norm of $T$ is given by
    \begin{equation}
        \|T\|_p= \Tr(|T|^p)^{1/p},
    \end{equation}
    where $|T|=\sqrt{T^\dagger T}$, for $p<\infty$. For $p=\infty$, this is defined as $\|T\|_\infty=\max_{\|u\|=1}\|Tu\|$, where the vector norm is the standard Euclidean norm.
\end{definition}
\begin{definition}[Quantum $\chi^2$ Divergence, \cite{temme2010chi}]\label{def:quantum_chi_squared_divergence}
    Let $\rho,\sigma\in\mathbb{C}^{d\times d}$ and $\alpha \in [0,1]$. When $\supp(\rho)\subseteq\supp(\sigma)$, the \emph{$\alpha$-$\chi^2$ divergence between $\rho$ and $\sigma$} is defined as 
    \begin{equation}
        \QChiAlpha{\alpha}{\rho}{\sigma}= \Tr(\rho-\sigma)\sigma^{-\alpha}(\rho-\sigma)\sigma^{\alpha-1})= \Tr(\rho\sigma^{-\alpha}\rho\sigma^{\alpha-1})-1.
    \end{equation}
\end{definition}
\noindent In particular, we will use the ``symmetrized'' version of this $\chi^2$ divergence, which is when $\alpha=\frac{1}{2}$. From this point on, unless stated otherwise, we will use this symmetrized version and write $\QChi{\rho}{\sigma}$ for $\QChiAlpha{1/2}{\rho}{\sigma}$. We then have
\begin{align}
    \QChi{\rho}{\sigma}&=\Tr((\rho-\sigma)\sigma^{-1/2}(\rho-\sigma)\sigma^{-1/2})\\
    &=\Tr(\sigma^{-1/2}(\rho-\sigma)\sigma^{-1/2}(\rho-\sigma))\\
    &=\Tr((\sigma^{-1/4}(\rho-\sigma)\sigma^{-1/4})^2)\label{eq:useful:symmetrization}\\
    &=\Tr(\rho\sigma^{-1/2}\rho\sigma^{-1/2})-1.
\end{align}
Notice the third equality, \eqref{eq:useful:symmetrization}, is what makes this ``symmetric'', since one can collapse the trace into a square. We note that all such divergences upper bound the trace distance:
\begin{lemma}[{\cite[Lemma 5]{temme2010chi}}]\label{lem:1_norm_qchi_inequality}
    Let $\rho$ and $\sigma$ be two quantum states. For all $\alpha \in [0,1]$, $\|\rho-\sigma\|_1^2 \leq \QChiAlpha{\alpha}{\rho}{\sigma}$.
\end{lemma}

We will apply the above lemma specifically for the case of $\alpha = \frac12$.

We also briefly recall what a \emph{classical-quantum  (CQ)  state} is.

\begin{definition}[Classical-Quantum States]
\label{def:cq-states}
    We denote the ``classical register'' by $C$ and the ``quantum register'' by $Q$, with dimension $d_c$ and $d_q$ respectively; we may omit these labels when the registers are clear from context. Then, a classical-quantum state over these registers is a special $d_cd_q$ dimensional quantum state of the following form:
\begin{equation}
\omega_{CQ}=\sum_{c=1}^{d_c}\omega_c\otimes|c\rangle\langle c|_C,
\end{equation}
where each $\omega_c$ is a $d_q$-dimensional positive semidefinite operator with trace at most $1$, and $\Tr(\omega_{CQ}) =1$. That is, each $\omega_c$ is not itself a valid quantum state, but the overall state $\omega_{CQ}$ has trace 1, and hence is a valid state.
\end{definition}

We will use the register labels $C$ and $Q$ defined above throughout this paper.  

\subsection{Quantum Channels and Quantum Instruments}\label{sec:quantum_channels_instruments}
In this section, we establish the main tools used by the distributed nodes to communicate with the central node. Firstly, when only quantum communication is allowed, each node can design a \emph{quantum channel} $\Phi$, which they can use to send a $d_q$-dimensional qudit to the central referee node.
\begin{definition}[Quantum Channel]\label{def:quantum_channel}
    A map $\Phi\colon\mathbb{C}^{d\times d}\to\mathbb{C}^{d'\times d'}$ is a \emph{quantum channel} if it is a completely positive and trace-preserving map (CPTP map).
\end{definition}
Analogously to how bits can be transmitted across classical channels, these are the quantum generalization that allows transmissions of quantum states. To transmit both quantum and classical bits, the nodes can instead implement a channel that only produces classical-quantum states, i.e., a \emph{quantum instrument}.

\begin{definition}[Quantum Instrument, \cite{hashim2026understandingquantuminstruments}]\label{def:quantum_instrument}
    A \emph{quantum instrument} (QI) $\mathcal{I}\colon\mathbb{C}^{d\times d}\to \mathbb{C}^{d_q\times d_q}\otimes \mathbb{C}^{d_c\times d_c}$ is a completely positive trace preserving map:
    \begin{equation}
        \mathcal{I}(\rho)=\sum_{i\in [d_c]}\mathcal{E}_i(\rho)\otimes|i\rangle\langle i|,
    \end{equation}
    where each $\mathcal{E}_i : \mathbb{C}^{d \times d} \to \mathbb{C}^{d_q \times d_q}$ is a completely positive map.
    When measuring the state above, the probability that an outcome $i$ is observed is given by
    $
    \Pr[i\mid \rho]=\Tr(\mathcal{E}_i(\rho)).
    $   
    Then, the post-measurement state, conditioned on the classical outcome $i$, is 
    \begin{equation}
        \rho_i\coloneqq\frac{\mathcal{E}_i(\rho)}{\Tr(\mathcal{E}_i(\rho))}.
    \end{equation}
\end{definition}

We make a few comments on the above definition before continuing. First, notice that although each individual $\mathcal{E}_i$ is not necessarily a quantum channel (since it may not preserve traces), we require $\mathcal{I}$ overall to preserve traces. Then, we have that 
\begin{equation}
    \Tr(\rho)=\sum_i \Tr(\mathcal{E}_i(\rho)) = 1.
\end{equation}
Furthermore, the post-measurement state $\rho_i$, conditioned on seeing the classical outcome $i$, is a valid quantum state. Second, when we use ``measurement'' in the above definition, it is better interpreted in the colloquial sense of ``extracting information from a quantum system,'' and not the mathematical formalism of applying a measurement on the qubit system (though one can think of it as the latter). We give some examples to help illustrate quantum instruments, and hopefully make the definition and comments clear.
\begin{example}
    As a sanity check, let us assume the quantum instrument has $d_c=1$ (corresponding to $n_c=0$), so no classical communication is allowed. Then, we have that $\mathcal{I}(\rho)=\mathcal{E}(\rho)$. In this case, we have that $\mathcal{E}$ is in fact a CPTP map, and this reduces to a quantum channel, as is expected.
\end{example}
\begin{example}
    To illustrate the second point, let us assume a distributed node receives the state $\rho$, and does the following. For the classical string, they sample an $n_c$-bit string $s$ uniformly at random, and send that. Then, they use some quantum channel $\Phi$, independent of the random string $s$, and send the quantum qudit $\Phi(\rho)$, meaning they send the pair $\{s, \Phi(\rho)\}$. Then, denoting $d_c=2^{n_c}$, each string $s$ appears with probability $1/d_c$. Correspondingly, for each of the strings, they have $\mathcal{E}_s(\rho)=\frac{1}{d_c}\Phi(\rho)$. Hence, the full quantum instrument is
    \begin{equation}
        \mathcal{I}(\rho)=\sum_{i\in\{0,1\}^{n_c}}\mathcal{E}_i(\rho)\otimes|i\rangle\langle i|.
    \end{equation}
    Then the central node, who receives the output $\mathcal{I}(\rho)$, can measure the classical register and see some string $i = s$ with probability $1/d_c$. Then, conditioned on seeing that string $s$, what the referee node has remaining is the quantum state
    \begin{equation}
        \frac{\frac{1}{d_c}\Phi(\rho)}{\Tr(\frac{1}{d_c}\Phi(\rho))}=\Phi(\rho).
    \end{equation}
    Hence, the referee node sees a uniformly random classical string and the output of $\rho$ fed into the quantum channel $\Phi$, which is exactly the process we described above.
\end{example}
\begin{example}
    Finally, lets imagine that the distributed node always wants to send a fixed classical string $s$, and use a quantum channel $\Phi$ for quantum communication. Then, we have that $\mathcal{E}_s(\rho)=\Phi(\rho)$, and $\mathcal{E}_i(\rho)=0$, for all $i\neq s$. Then, we have that the corresponding quantum instrument is
    \begin{equation}
        \mathcal{I}(\rho)=\Phi(\rho)\otimes|s\rangle\langle s|.
    \end{equation}
    Then, the referee receives the output of the quantum instrument and can do the same: they measure the classical register to see the classical string $s$, and what remains is the output of the channel $\Phi(\rho)$.
\end{example}

We stated the simple examples above to help the reader develop some intuition about quantum instruments. More generally, the maps $\mc{E}_i$ will be heavily correlated with the classical string $i$, and this will be important for the design of our algorithms.

It will also be helpful to outline some alternative ways of expressing these channels and instruments. One of these is the \emph{Liouville representation}.

\begin{definition}[Vectorization, Liouville Matrix Representation]\label{def:vectorization}
    The \emph{vectorization} map $\vecmap :\mathbb{C}^{d\times d}\to\mathbb{C}^{d^2}$ flattens a matrix into a vector. This naturally preserves inner product, where if $X$ and $Y$ are matrices, then $\InnerNamed{X}{Y}{HS}=\Inner{\vecmap(X)}{\vecmap(Y)}$, where the first inner product is the Hilbert--Schmidt inner product ($\Tr(X^\dagger Y)$) and the latter is the inner product on vectors. The \emph{Liouville matrix} of a quantum channel $\Phi:\mathbb{C}^{d\times d}\to\mathbb{C}^{d'\times d'}$ is the unique matrix $M_\Phi\in\mathbb{C}^{d'^2\times d^2}$ such that $\vecmap(\Phi(X))=M_\Phi\vecmap(X)$.
\end{definition}
The other is the \emph{Kraus decomposition}:
\begin{definition}[Kraus Operator Decomposition]\label{def:kraus_operator_decomposition}
    Any completely positive operator $\Phi:\mathbb{C}^{d\times d}\to\mathbb{C}^{d'\times d'}$ can be described by a set of operators $\{A_k\}\subseteq\mathbb{C}^{d'\times d}$, where $1 \leq k \leq dd'$, such that the action of the channel is given by $\Phi(X)=\sum_k A_kXA_k^\dagger$, and $\sum_k A_k^\dagger A_k=I$.
\end{definition}
The existence of the above decomposition is known as the Choi-Kraus theorem, a standard tool; see, for instance, \cite[Chapter 4]{Wilde_2016}.

\subsection{Haar-measure moments}

We state here some miscellaneous lemmas that will be required later. We will often need to perform averages over the $d$-dimensional unitary Haar measure, which we denote with $U(d)$. The simplest instance of these is when we wish to compute second moments of Haar-random unitaries. We will use the following standard lemma (see e.g.,~\cite[Corollary 13]{mele2024introduction}).

\begin{lemma}
    \label{lem:haar-second-moment}
    Let $M \in (\mathbb{C}^{d \times d})^{\otimes 2}$. Then,
    \begin{equation}
        \mathbb{E}_{\bfU \sim U(d)} [\bfU^{\otimes 2} M \bfU^{\dag \otimes 2}] = \frac{I}{d^2-1} \left(\Tr(M) - \frac1d\Tr(M \cdot \swap)\right) + \frac{\swap}{d^2-1} \left(\Tr(M \cdot \swap) -\frac1d \Tr(M)\right),
    \end{equation}
    where $\swap$ is the unitary representation of the $2$-element swap permutation on $(\mathbb{C}^{d \times d})^{\otimes 2}$.
\end{lemma}

Specifically, we will use the following corollary implied by the above lemma:

\begin{corollary}[Second-order integral]
\label{cor:second-order}
For any $A_1,A_2,B_1,B_2 \in \mathbb{C}^{d \times d}$, we have
    \begin{align}
        &\mathbb{E}_{\bfU \sim U(d)}[\Tr(\bfU^{\otimes 2} (A_1 \otimes A_2)\bfU^{\dag \otimes 2} (B_1 \otimes B_2))] \nonumber\\&= \frac{\Tr(A_1A_2)\Tr(B_1B_2) + \Tr(A_1 \otimes A_2) \Tr(B_1 \otimes B_2)}{(d^2-1)} - \frac{\Tr(A_1A_2)\Tr(B_1 \otimes B_2) + \Tr(A_1 \otimes A_2)\Tr(B_1B_2)}{d(d^2-1)}.
    \end{align}
\end{corollary}

More generally, we may wish to compute $k$th moments, for some $k > 2$; such higher-order moments can be computed using the \emph{Weingarten calculus}. We only state the most general form necessary in the section where these computations are carried out, i.e., \Cref{app:compression-bounds}.
\section{Public-Coin Upper Bound}\label{s:public_ub}

We will prove our public-coin upper bound from \Cref{thm:intro-public} in this section. The key idea is for the distributed nodes to compress their states down to lower-dimensional spaces, so that they can be sent over the restricted communication channels. However, this compression must be carried out while preserving as much of the distance between $\rho$ and $\sigma$ as possible, to ensure certification remains tractable.

As a warmup, we first consider the case where only restricted quantum communication is available (i.e., $n_q < n$, $n_c = 0$) in \Cref{ss:quantum_public_ub}. Here, we will consider an appropriate randomized dimensionality-reduction channel that preserves an optimal amount of distance between the states (up to constant factors). After applying this channel, we will appeal to an algorithm of B\u{a}descu, O'Donnell, and Wright~\cite{buadescu2019quantum} for estimating and testing the Hilbert--Schmidt distance between two states.

\begin{theorem}[Hilbert--Schmidt estimation and certification,~{\cite[Theorem 1.4 \& Lemma 5.6]{buadescu2019quantum}}]
    \label{lem:hscertify}
    Given $t$ copies of $\rho$ and the description of a state~$\sigma$\footnote{The estimator and tester here also work just as well when $\sigma$ is unknown and one only has access to $t$ copies of it.}, there is an unbiased estimator for $\|\rho - \sigma\|_2^2$ with variance at most
    \begin{equation}
        \bigo\left(\frac{1}{t^2}+\frac{\|\rho - \sigma\|_2^2}{t}\right).
    \end{equation}
    By applying this estimator and thresholding, one can obtain a tester $\mathsf{HSCertify}(\sigma,\epsilon,\delta)$ that distinguishes between $\rho = \sigma$ and $\|\rho - \sigma\|_2 \geq \eps$ with probability at least $1-\delta$ from $t = \bigo(\log(1/\delta)/\epsilon^2)$ copies. 
\end{theorem}

This yields the tight upper bound for testing with only quantum communication.

For the mixed-communication model, we will generalize the compression argument above to construct a randomized quantum \emph{instrument} achieving tight distance-preservation bounds. We will then again use the \cite{buadescu2019quantum} estimator (henceforth the ``BOW estimator''), but its analysis is too loose for our purposes here; as we only ever apply it to classical-quantum states, we can prove tighter bounds on its variance, which finally lead to our main public-coin upper bound. We provide an overview of these arguments in \Cref{ss:outline} before moving on to the actual proofs. As a final technical note, our proofs will also use the Paley-Zygmund inequality.
\begin{lemma}[Paley-Zygmund inequality; see, e.g.,\cite{steele2004paley}]\label{lemma:Paley_Zygmund}
    Let $Y\geq 0$ be a random variable with finite first moment. Then, for any $\theta\in [0,1]$, we have
    \begin{equation}
        \Pr[Y\geq \theta\mathbb{E}[Y]]\geq (1-\theta)^2\frac{(\mathbb{E}[Y])^2}{\mathbb{E}[Y^2]}.
    \end{equation}
\end{lemma}

\subsection{Warmup 1: Public-Coin Testing with Only Quantum Communication}\label{ss:quantum_public_ub}

In this section, we prove our upper bound for distributed state certification using only quantum communication and public coins. More generally, instead of only allowing $n_q$-qubit messages, we will allow the algorithm to send a single qudit of dimension $d_q$, for any $2 \leq d_q \leq d$. We restate this upper bound in the following theorem: 
\begin{theorem}
\label{thm:public-upper-bound-quantum-only}
    Assume the distributed nodes are allowed to communicate a $d_q$-dimensional qudit, no classical bits, and all parties share public coins. Then, to $\epsilon$-certify a $d$-dimensional state with probability at least $2/3$, it suffices to take $m= \bigo\left(\frac{d^2}{d_q\epsilon^2} \right)$, i.e., in the $(0,n_q,\mathsf{public})$ setting,
    \begin{equation}
        m= \bigo\left(\frac{d^2}{2^{n_q}\epsilon^2}\right)
    \end{equation}
    copies are sufficient.
\end{theorem}

The probability of success can be amplified using standard arguments, which we omit here. Our arguments to prove the above theorem assume that $d_q$ divides $d$, but these can easily be adjusted to hold for general $d_q$\footnote{When $d_q$ does not divide $d$, one can embed the unknown state and the hypothesis state into a space of dimension $d^\prime = d_q \times \lceil d/d_q \rceil$ while leaving their distance unchanged. This increases the dimension by at most a constant factor, as $d^\prime < d +d_q \leq 2d$, allowing our subsequent arguments to follow with only minor adjustments.}. Now, let $Q,A$ be a bipartition of $\mathbb{C}^{d \times d}$ where the dimension of $Q$ is $d_q$. For any unitary $U$, let $\Phi_U(\rho) \triangleq \Tr_A(U\rho U^\dag)$. Before stating our algorithm formally, let us provide a brief overview of its steps. The central node $N_c$ and a batch of distributed nodes will sample a random unitary $\bfU$, and the distributed nodes will each send the $d_q$-dimensional state $\Phi_{\bfU}(\rho)$ to $N_c$. Then, $N_c$ will use \Cref{lem:hscertify} to test whether $\Phi_{\bfU}(\rho)$ and $\Phi_{\bfU}(\sigma)$ are identical or appropriately far in $\|\cdot\|_2$-distance with constant success probability. We then repeat this across several batches to get the probability to at least $\frac23$. The complete description of the constant-probability algorithm is stated in \Cref{alg:public-coin-quantum}.

\begin{algorithm}[H]
    \begin{algorithmic}[1]
        \caption{Distributed state certification with $d_q$-dimensional quantum communication and public coins}
        \label{alg:public-coin-quantum}
        \State Set $R = \bigo(1), m^\prime = \bigo(d^2/d_q\epsilon^2)$, such that $m = m^\prime \cdot R$.
        \State Set $\epsilon^\prime, \delta^\prime,\tau$ according to \Cref{eq:parameters-2}.
        \State All distributed nodes and the central node use their shared randomness to sample $R$ random unitaries $\bfU_1,\dots,\bfU_R \sim U(d)$.
        \For{$r = 1$ to $R$}
            \State Nodes $N_{(r-1)\cdot m^\prime + 1},\dots,N_{r\cdot m^\prime}$ apply $\Phi_{\bfU_r}$ to their individual copies of $\rho$ and send the output state to $N_c$.
            \State The central node $N_c$ applies $\mathsf{HSCertify}(\Phi_{\bfU_r}(\sigma),\epsilon^\prime,\delta^\prime)$ to these states and records the outcome ``Close'' or ``Far''. 
        \EndFor
        \If{``Far'' occurs more than $\tau R$ times} reject.
        \Else{ } accept.
        \EndIf
    \end{algorithmic}
\end{algorithm}

To prove \Cref{thm:public-upper-bound-quantum-only}, we will leverage the fact that the channel $\Phi_{\bfU}$ defined above only compresses the distance between two states by a factor proportional to $\sqrt{\frac{d_q}{d}}$:

\begin{lemma}
\label{lem:mixed-state-compression-quantum}
    Let $\rho,\sigma \in \mathbb{C}^{d \times d}$ be two quantum states. Then, there exists an absolute constant $C_2 > 1$ such that
    \begin{equation}
        \Pr_{\bfU \sim U(d)}\left[\|\Phi_{\bfU}(\rho) - \Phi_{\bfU}(\sigma)\|_2 \geq \frac12\sqrt{\frac{d_q}{d}} \cdot \|\rho-\sigma\|_2\right] \geq \frac{1}{C_2}.
    \end{equation}
\end{lemma}

We defer the proof of the above lemma and first use it to prove \Cref{thm:public-upper-bound-quantum-only}:

\begin{proof}[Proof of \Cref{thm:public-upper-bound-quantum-only}]
    Consider a fixed iteration $r \in [R]$ of \Cref{alg:public-coin-quantum}. In the case of $\rho \neq \sigma$, \Cref{lem:mixed-state-compression-quantum} says 
    \begin{align}
     \Pr\left[\|\Phi_{\bfU_r}(\rho) - \Phi_{\bfU_r}(\sigma)\|_2 \geq \frac12\sqrt{\frac{d_q}{d}} \cdot \|\rho-\sigma\|_2\right] \geq \frac{1}{C_2}.
    \end{align}
    When $\|\rho-\sigma\|_1 \geq \epsilon$, a standard norm inequality implies $\|\rho - \sigma\|_2 \geq \epsilon/\sqrt{d}$, and therefore with probability at least $1/C_2$, we have 
    \begin{equation}
        \|\Phi_{\bfU_r}(\rho) - \Phi_{\bfU_r}(\sigma)\|_2 \geq \frac{\sqrt{d_q} \epsilon}{2d}.
    \end{equation}
    We now define the following parameters:
    \begin{equation}
        \label{eq:parameters-2}
        \epsilon^\prime = \frac{\sqrt{d_q} \epsilon}{2d}, \delta^\prime = \frac{1}{4C_2}, \tau = \frac{1}{2C_2}.
    \end{equation}
    In iteration $r$, $\mc{A}$ calls $\mathsf{HSCertify}(\Phi_{\bfU_r}(\sigma),\epsilon^\prime,\delta^\prime)$ on its copies of $\Phi_{\bfU_r}(\rho)$; using \Cref{lem:hscertify}, this can be done when $m^\prime = \bigo(\log(1/\delta^\prime)/\epsilon^{\prime 2}) = \bigo(d^2/d_q \epsilon^2)$. When $\|\rho-\sigma\|_1 \geq \epsilon$, we have
    \begin{equation}
        \Pr[\text{``Far''}] \geq (1-\delta^\prime) \cdot \Pr[\|\Phi_{\bfU_r}(\rho) - \Phi_{\bfU_r}(\sigma)\|_2 \geq \epsilon^\prime] \geq \left(1 - \frac{1}{4C_2}\right) \frac{1}{C_2} \geq \frac{3}{4C_2}.
    \end{equation}
    On the other hand, when $\rho = \sigma$, we have $\Phi_{\bfU_r}(\rho) = \Phi_{\bfU_r}(\sigma)$, and thus in this case, 
    \begin{equation}
        \Pr[\text{``Far''}] \leq \delta^\prime = \frac{1}{4C_2}.
    \end{equation}
    Thus, one can distinguish between the two cases by repeating this procedure and comparing the frequency of ``Far'' to $1/2C_2$. Using Hoeffding's inequality, this test succeeds with probability at least $2/3$ using $R = \bigo(C_2^2) = \bigo(1)$ repetitions. The overall number of servers used is 
    \begin{equation}
        m = m^\prime \cdot R = \bigo\left(\frac{d^2}{d_q\epsilon^2}\right),
    \end{equation}
    as desired.
\end{proof} 

Our proof of \Cref{lem:mixed-state-compression-quantum} will use the following bounds on the moments of $\|\Phi_U(\rho) - \Phi_U(\sigma)\|_2$:

\begin{lemma}
\label{lem:expected-distance-improved}
    For quantum states $\rho,\sigma \in \mathbb{C}^{d \times d}$ with $d \geq d_q \geq 2$,
    \begin{equation}
        \mathbb{E}_{\bfU \sim U(d)} [\|\Phi_{\bfU}(\rho) - \Phi_{\bfU}(\sigma)\|_2^2] \geq \frac{d_q}{2d} \|\rho-\sigma\|_2^2.
    \end{equation}
\end{lemma}

We will also use the following bound on the fourth moment of the distance.

\begin{lemma}
    \label{lem:expected-distance-fourth-improved}
    For quantum states $\rho,\sigma \in \mathbb{C}^{d \times d}$ with $d$ sufficiently large, 
    \begin{equation}
        \mathbb{E}_{\bfU \sim U(d)} [\|\Phi_{\bfU}(\rho) - \Phi_{\bfU}(\sigma)\|_2^4] \leq C_1 \cdot \frac{d_q^2}{d^2} \|\rho-\sigma\|_2^4,
    \end{equation}
    where $C_1 > 1$ is an absolute constant.
\end{lemma}

Applying the Paley-Zygmund inequality (\Cref{lemma:Paley_Zygmund}) to the random variable $\|\Phi_{\bfU}(\rho) - \Phi_{\bfU}(\sigma)\|_2^2$ and using the above \Cref{lem:expected-distance-improved,lem:expected-distance-fourth-improved}, we immediately obtain \Cref{lem:mixed-state-compression-quantum}. As the proofs of \Cref{lem:expected-distance-improved,lem:expected-distance-fourth-improved} involve tedious bookkeeping and aren't particularly instructive, we defer them to \Cref{app:compression-bounds}.

\subsection{Main Proof Outline}\label{ss:outline}
In this section, we overview our arguments for extending the results in \Cref{ss:quantum_public_ub} to when both classical and quantum communication are permitted. Throughout, we assume the algorithm can communicate $n_c$ bits, meaning there are $d_c=2^{n_c}$ possible strings that may be communicated. Before giving the algorithm and the proofs, we make a few observations to provide intuition as to why the setting with both classical and quantum communication allowed is significantly more challenging. For example, one might argue as follows: we know how to test when communicating classical bits only, and we now know how to test when communicating qudits only, so we could just run both tests and accept a state if both pass. More specifically, the central node knows the quantum instruments that the distributed nodes are using to communicate. Hence, they can run those quantum instruments on the known state $\sigma$, and compare its output to the classical bits and the qudit obtained from $\rho$. The issue is that this approach ignores the relations between the joint states, and thus it is possible to accept on both the classical and the quantum registers, but still have $\rho$ be far from $\sigma$. We give an example here:
\begin{example}
    Assume we have the unknown state $\rho$, and after passing it through the quantum instrument, we transmit the state
    \begin{equation}
        \rho_{CQ}=\frac{1}{2}|0\rangle\langle0|_Q\otimes|0\rangle\langle0|_C+\frac{1}{2}|1\rangle\langle1|_Q\otimes|1\rangle\langle1|_C,
    \end{equation}
    in other words, it is a state where the quantum qubit equals the classical bit. Assume we have a known state $\sigma$, and under the same quantum instrument, we produce the state
    \begin{equation}
        \sigma_{CQ}=\frac{1}{2}|1\rangle\langle1|_Q\otimes|0\rangle\langle0|_C+\frac{1}{2}|0\rangle\langle0|_Q\otimes|1\rangle\langle1|_C.
    \end{equation}
    Now, one can try the above, analyzing the marginals obtained by classical communication and the quantum communication separately. In the classical communication, we have that
    \begin{align}
        \rho_C&=\frac{1}{2}|0\rangle\langle0|+\frac{1}{2}|1\rangle\langle1|\\
        \sigma_C&=\frac{1}{2}|0\rangle\langle0|+\frac{1}{2}|1\rangle\langle1|.
    \end{align}
    In other words, both see each bit with probability half. A classical only test would see no difference in these two marginals, and accept. Computing the quantum marginals, one would also see that
    \begin{align}
        \rho_Q&=\frac{1}{2}|0\rangle\langle0|+\frac{1}{2}|1\rangle\langle1|\\
        \sigma_Q&=\frac{1}{2}|1\rangle\langle1|+\frac{1}{2}|0\rangle\langle0|.
    \end{align}
    These two marginals are also equal, so the quantum only test would accept this part. So both have accepted, yet the corresponding outputs of the quantum instrument are actually far in trace distance.
\end{example}
It is clear from this example that if one wants to succeed in the mixed-communication setting, one must have an algorithm that looks at the classical bits along with the quantum bits together, and not one that isolates each part. Since the remainder of the proof is dense, we give a brief overview of the steps we take to prove this upper bound.

Overall, our upper bound approach can be broken into four main steps:
\begin{enumerate}
    \item Specific analysis of the BOW estimator for ``balanced'' classical-quantum states. At a high level, we aim to show that the BOW estimator actually has variance that depends on $d_c$, assuming the CQ states have probability mass that is roughly $\bigo(1/d_c)$ over each of the classical labels.
    \item Specific design of quantum instruments which, for a fixed state, have a high probability to spread probability mass over all the classical labels, such that each has probability mass $\bigo(1/d_c)$.
    \item Proof that these quantum instruments satisfy a compression property similar to \Cref{lem:mixed-state-compression-quantum}. 
    \item Explicit design of an algorithm using the above.
\end{enumerate}
We now discuss each of these four points in detail, giving the rationale, and the technical details behind each.

\paragraph{Specific analysis of the BOW estimator.} Since our goal is to detect large trace distance, a natural idea is to use the same approach as in the quantum-only communication algorithm: that is, use the BOW tester (which relies on the BOW estimator as a subroutine) on the CQ states to test trace distance. After all, the BOW tester works on all quantum states, and CQ states are just specific quantum states. The issue with this is that the variance in the BOW estimator does not have any dependence on $d_c$, and thus such an approach will not, at least na\"ively, benefit from having classical bits. Moreover, improving the variance bound in the BOW estimator to obtain a dependence on $d_c$ also seems like a difficult task. An insight, however, is that the very generality of the BOW estimator is where we can try to make improvements, as we only need a good estimator for CQ states, not arbitrary ones. The idea then is to go over the analysis for the BOW estimator, and get a variance that depends, at a high level, on how much the probability spreads over each of the classical labels. Importantly, this does \emph{not} lead to a generic improvement to the BOW tester: indeed, if the CQ state concentrates on one (or, more broadly, few) of the classical strings, with no mass on many, we recover the variance of the BOW tester. Yet, whenever the CQ states spread mass ``roughly uniformly'' across all the classical labels (each label appears with probability $\bigo(1/d_c)$), we obtain a sharper bound on the variance; crucially, this bound depends on $d_c$, allowing us to get the $d_c$-dependence in our complexity when we apply Chebyshev's inequality. We carry out this step in \Cref{sss:bow_estimator}.

\paragraph{Quantum Instrument Design.} The goal, given the above analysis of the BOW estimator, is then to design quantum instruments that, given the states $\rho$ and $\sigma$, sufficiently spread the communication over all (or most) of the classical labels. Roughly speaking, we want the distributed node to communicate using, on average, almost all of the classical strings available to it. Since the BOW estimator needs, to benefit from this improved variance, each classical label to appear with probability $\bigo(1/d_c)$, we aim for quantum instruments which use each classical string with probability $\bigo(1/d_c)$. To this end, we show that the instrument analogue of the quantum channel used above suffices: that is, a quantum instrument defined by (1) applying a Haar-random unitary to a state, then (2) measuring the classical register to create the classical message, (3) tracing out an unused register, and (4)~sending the remaining state as the quantum message. We define the instrument and show the desired spreading property in \Cref{sss:quantum-instrument-design}.

\paragraph{Instrument Compression.} The next thing we require is that our instruments satisfy a similar compression result as the channels do above. At a high level, we require that the instruments do not ``contract'' the Hilbert--Schmidt distance too much, so that we can still distinguish when $\rho=\sigma$ and when $\|\rho-\sigma\|_1\geq\epsilon$. The approach to this structurally is the same as in \Cref{ss:quantum_public_ub}, and one can think of this step as just a generalization of the arguments therein. We show the compression property for our instruments in \Cref{ss:quantum-classical-compression}.

\paragraph{Algorithm Design} With these three former setup steps handled, we then generalize~\Cref{alg:public-coin-quantum} to obtain our mixed-communication~\Cref{alg:public_coin} and the proof of our main upper bound \Cref{thm:intro-public} in \Cref{ss:public-coin-algorithm}.

We note that the first two steps,i.e., the analysis of the BOW estimator and the design of quantum instruments, will also be key ingredients in the analysis of our private coin algorithms. We will now go through these steps, introducing necessary preliminaries and results in each section.
\subsection{Main Proof}\label{ss:public_ub}
\subsubsection{Analysis of the BOW Estimator}\label{sss:bow_estimator}
We start by defining some notation. First, we assume $\rho$ and $\sigma$ are CQ states, meaning that they are of the form
\begin{equation}
    \rho=\sum_{c=1}^{d_c}\rho_c\otimes|c\rangle\langle c|_C, \quad\sigma=\sum_{c=1}^{d_c}\sigma_c\otimes|c\rangle\langle c|_C.
\end{equation}
When it is clear, we also drop the subscript $C$ for the classical register. We also let
\begin{equation}
    p_c\triangleq \Tr[\rho_c], \quad q_c\triangleq\Tr[\sigma_c],
\end{equation}
which are the classical marginal probabilities of seeing the label $c$ under $\rho$ and $\sigma$. Finally, we also note
\begin{equation}
    \|p\|_\infty= \max_{c}p_c, \quad
    \|p\|_2^2= \sum_cp_c^2, \quad
    \|q\|_2^2=\sum_c q_c^2.
\end{equation}
These are important quantities which capture how ``spread out'' $\rho$ and $\sigma$ are over the classical labels. With this in hand, the main result we prove is the following:
\begin{theorem}\label{thm:cq_hs_estimator}
    Given $m$ copies of a CQ state $\rho$, and full knowledge of a CQ state $\sigma$, there is an estimator $\hat{\Gamma}$ such that
    \begin{equation}
        \mathbb{E}[\hat{\Gamma}]=\Gamma= \|\rho - \sigma\|_2^2,
    \end{equation}
    and 
    \begin{equation}
        \mathbb{V}(\hat\Gamma)\leq \bigo\left(\frac{\|p\|_\infty\Gamma}{m}+\frac{\|p\|_2^2+\|q\|_2^2}{m^2}\right).
    \end{equation}
\end{theorem}

Our estimator is a slight modification of the BOW estimator tailored to classical-quantum states. We also assume a classical description of the state $\sigma$ instead of its copies, as this mildly simplifies our analysis; we leave it to the diligent reader to verify that a similar treatment leads to the same variance bound when both states are unknown. We also emphasize again that the theorem above is \emph{not} a general improvement over the results of~\cite{buadescu2019quantum}: for instance, if the CQ states concentrate on one classical label, then  $\|p\|_\infty=1$, and $\|p\|_2^2=\|q\|_2^2=1$, and we recover the result in \Cref{lem:hscertify}. However, if all $\|p\|_2^2=\|q\|_2^2=\bigo(1/d_c)$, which implies $\|p\|_\infty=\bigo(1/\sqrt{d_c})$, then we do get a genuine improvement on the variance. This latter regime is the one relevant to our analysis.

To prove this result, we follow the analysis of~\cite{buadescu2019quantum} to design the estimator and to bound the variance. Though we adapt the arguments to work for CQ states, much of the analysis is exactly the same as theirs, with careful modifications to obtain the more fine-grained variance bound. Due to the large similarity, we defer the proof to \Cref{app:BOW_analysis}.

\subsubsection{Quantum Instrument Design}
\label{sss:quantum-instrument-design}

Recall, in light of the above variance results, we need a quantum instrument which sufficiently spreads out over the classical labels. In this section, we give such a design. Without loss of generality, let us assume $d=d_cd_qd_A$. In other words, we define our input Hilbert space $\mathbb{C}^d\cong C\otimes Q\otimes A$, where $\dim C=d_c,~\dim Q=d_q, \dim A=d_A$. (As before, if $d$ is not divisible by $d_cd_q$, one can pad the input space, with only a constant factor loss.) At a high level, $C$ will be the classical register we send, $Q$ will be the quantum register we send, and $A$ will be the remaining register we trace out. To do this, we first design a quantum instrument based on sampling Haar-random unitaries, which we call a \emph{Haar instrument}. This instrument, to the best of our knowledge, is novel, and may find interest in more general quantum communication.

\begin{definition}[Haar instrument]\label{def:Haar_instrument}
    A \emph{Haar instrument} is a quantum instrument constructed as follows. First, sample a $d$-dimensional Haar-random unitary $\bfU$. Then for each classical label $c\in[d_c]$, define the completely positive branch
    \begin{equation}
        \mathcal{I}_{\bfU,c}(X) \triangleq \Tr_A[(\langle c|_C\otimes I_{QA})\bfU X\bfU^\dagger(|c\rangle_C\otimes I_{QA})].
    \end{equation}
    Then, the full quantum instrument is the completely positive trace preserving map
    \begin{equation}
        \mathcal{I}_{\bfU} \triangleq \sum_{c=1}^{d_c}\mathcal{I}_{\bfU,c}(X)\otimes|c\rangle\langle c|_C.
    \end{equation}
\end{definition}
Roughly speaking, the user of this instrument applies $\bfU$ to their state, traces out the $A$ register, then measures $C$ in the computational basis to obtain a string $c$, which they send, without observing, along with the $Q$ register. We now argue that this is a valid instrument.
\begin{lemma}\label{lem:haar_instrument_is_instrument}
    For every positive semidefinite $X\succeq0$ and every unitary $U$, each branch $\mathcal{I}_{U,c}(X)$ is also positive semidefinite. Moreover, $\mathcal{I}_U$ preserves traces, and if $\rho$ is a density matrix, then $\mathcal{I}_U(\rho)$ is a normalized CQ state.
\end{lemma}
\begin{proof}
    To see the first claim, notice each branch consists of a unitary conjugation, projection to a subspace corresponding to a classical label $c$, and then a partial trace over $A$. Each of these operations is completely positive, thus so is their composition. Let us define the projector
    \begin{equation}
        \Pi_c\triangleq |c\rangle\langle c|_C\otimes I_Q\otimes I_A.
    \end{equation}
    These $d_c$ projectors are pairwise-orthogonal, and sum to the identity. Then, we have that $\Tr[\mathcal{I}_{U,c}(X)]=\Tr[\Pi_cUXU^\dagger]$. Thus,
    \begin{align}
        \Tr[\mathcal{I}_U(X)]&=\sum_{c=1}^{d_c}\Tr[\mathcal{I}_{U,c}(X)]
        =\Tr\left[\left(\sum_{c=1}^{d_c}\Pi_c\right)UXU^\dagger\right]
        =\Tr[UXU^\dagger]
        =\Tr[X].
    \end{align}
    Hence, the instrument is trace-preserving, so if $\rho$ is a density matrix, the output $\mathcal{I}_U(\rho)$ is a normalized CQ state.
\end{proof}

Recalling the variance of the estimator in \Cref{thm:cq_hs_estimator}, we wish to show that the marginal distributions under this instrument are sufficiently spread out; i.e., we wish to bound the $\ell_\infty$ and $\ell_2$ norms of these marginal distributions. First, we compute the second moment of the probability of a particular classical label $c$.

\begin{lemma}\label{lem:single_string_second_moment}
    Let $\rho$ be a density matrix, and let $\eta\triangleq\Tr[\rho^2]$. Further, let $c \in [d_c]$, and let $\rho_c(U)=\mathcal{I}_{U,c}(\rho)$. Recall that the probability to observe any specific classical outcome $c$ is $p^\rho_c(U)\triangleq \Tr[\rho_c(U)]$. Then, for all $c$, we have
    \begin{equation}
        \mathbb{E}_{\bfU}\left[(p_c^\rho(\bfU))^2\right]=a_1r^2+a_2r,
    \end{equation}
    where 
    \begin{equation}
        a_1=\frac{d-\eta}{d(d^2-1)},~a_2=\frac{d\eta-1}{d(d^2-1)},
    \end{equation}
    and $r$ is the rank of the projector $\Pi_c=d/d_c$.
\end{lemma}
\begin{proof}
    By definition, we have that $p_c^\rho(U)=\Tr[\Pi_cU\rho U^\dagger]$, the trace of $\rho$ after being conjugated by $U$ and projected onto the relevant string $c$. Hence, we have that
    \begin{equation}
    \label{eq:single-string-second}
        (p_c^\rho(U))^2=\Tr[(\Pi_c\otimes\Pi_c)(U\rho U^\dagger)^{\otimes 2}].
    \end{equation}
    Using \Cref{cor:second-order}, we have
    \begin{equation}
        \mathbb{E}_{\bfU} [(\bfU\rho \bfU^\dag)^{\otimes 2}] = a_1 I + a_2 \swap.
    \end{equation}
    Then, \Cref{eq:single-string-second} implies
    \begin{equation}
        \mathbb{E}_{\bfU}\left[(p_c^\rho(\bfU))^2\right]=\Tr[(\Pi_c\otimes\Pi_c)(aI+bF)].
    \end{equation}
    We can compute each of these terms. Notice $\Tr[(\Pi_c\otimes\Pi_c)I]=\Tr[\Pi_c]^2=r^2$. Next, $\Tr[(\Pi_c\otimes\Pi_c)F]=\Tr[\Pi_c^2]=\Tr[\Pi_c]=r$. Putting these two together with the coefficients proves the result.
\end{proof}
With these two lemmas in hand, we can now prove our Haar instrument sufficiently spreads out the mass among the classical labels. For convenience, let us define:
\begin{equation}
   \|p_{\rho}\|_2^2 \triangleq \sum_c(p_c^\rho(\bfU))^2.
\end{equation}
Showing that we have mass sufficiently spread across all possible labels is then equivalent to showing $\|p_{\rho}\|_2^2$ is small. We prove the following result:
\begin{lemma}\label{lem:Haar_spread}
    For every fixed density matrix $\rho$, we have that
    \begin{equation}
        \mathbb{E}_{\bfU}[|p_{\rho}\|_2^2]\leq \frac{2}{d_c}.
    \end{equation}
\end{lemma}
\begin{proof}
    Using \Cref{lem:single_string_second_moment}, we have that
    \begin{equation}
         \mathbb{E}_{\bfU}[|p_{\rho}\|_2^2]=\sum_{c=1}^{d_c}\mathbb{E}_{\bfU}\left[(p_c^\rho(\bfU))^2\right]=d_c(a_1r^2+a_2r).
    \end{equation}
    Plugging in the constants, and simplifying, one gets that this quantity simplifies to
    \begin{equation}
        \mathbb{E}_{\bfU}[|p_{\rho}\|_2^2] = \frac{d^2 + d(d_c-1)\eta - d_c}{d_c (d^2-1)} \leq \frac{d+d_c}{d_c (d^2-1)} \leq \frac{2}{d_c},
    \end{equation}
    where we simply use $\Tr(\rho^2) \leq 1$ and $d \geq d_c \geq 1$.
\end{proof}
By Markov's inequality, we get the simple corollary:
\begin{corollary}\label{cor:Haar_spread}
    Let $\rho$ be any fixed density matrix. Then, for every $A>0$, we have that
    \begin{equation}
        \Pr_{\bfU}\left[\|p_{\rho}\|_2^2>\frac{A}{d_c}\right]\leq \frac{2}{A}.
    \end{equation}
\end{corollary}
Then, taking $A\geq16$, we can use a union bound to show that with good probability, the chosen Haar instrument will spread out both $\rho$ and $\sigma$.
\begin{corollary}\label{cor:Haar_double_good_spread}
    Let $\rho$ and $\sigma$ be fixed density matrices. Then, with probability at least $3/4$ over the Haar-random $\bfU$ sampled with public-coin randomness, we have $\|p_{\rho}\|_2^2, \|p_{\sigma}\|_2^2\leq 16/d_c$. Consequently, with probability at least $3/4$, $\|p_{\rho}\|_2^2+\|p_{\sigma}\|_2^2\leq 32/d_c$.
\end{corollary}
In particular, with high constant probability our Haar instrument does not concentrate too much on any one classical label, and spreads its probability mass around~--~achieving the second item of our overall blueprint for a public-coin algorithm. We can also straightforwardly bound the maximum mass $\|p\|_\infty$ as follows:
\begin{corollary}\label{cor:maximum_mass_bound}
    Let $\rho$ be a fixed unitary and $\|p\|_\infty\triangleq \max_cp_c^\rho(U)$. Under the events in \Cref{cor:Haar_spread}, $\|p\|_\infty\leq 4/\sqrt{d_c}$.
\end{corollary}
\begin{proof}
    By \Cref{cor:Haar_spread}, we have that $\|p_{\rho}\|_2^2\leq 16/d_c$. Then, since $\|p\|_\infty^2\leq \|p_{\rho}\|_2^2$, we get $\|p\|_\infty\leq 4/\sqrt{d_c}$.
\end{proof}

\subsubsection{Quantum Instrument Compression}
\label{ss:quantum-classical-compression}
The aim of this section is to generalize \Cref{lem:mixed-state-compression-quantum} to quantum instruments. Namely, we show the following.
\begin{lemma}\label{thm:Haar_compression}
    For every traceless Hermitian $X$, there is an absolute constant $\gamma>0$ such that
    \begin{equation}
        \Pr_{\bfU}\left[\|\mathcal{I}_{\bfU}(X)\|_2^2\geq \frac{d_q}{4d}\|X\|_2^2\right]\geq \gamma.
    \end{equation}
\end{lemma}
The proof here relies on the generalization of \Cref{lem:expected-distance-improved} and \Cref{lem:expected-distance-fourth-improved}, which for brevity, we will also defer to \Cref{app:compression-bounds}. Though we will give an alternate proof of \Cref{lemma:first_moment_upsilon} which relies on simple linear algebra for interest's sake, the proof of \Cref{lemma:upsilon_second_moment} is mostly a generalization of \Cref{lem:expected-distance-fourth-improved}.
\begin{lemma}\label{lemma:first_moment_upsilon}
    Let $X$ be a traceless Hermitian. Then,
    \begin{equation}
        \mathbb{E}_{\bfU}[\|\mathcal{I}_{\bfU}(X)\|_2^2]=\frac{d_qd-d_A}{d^2-1}\|X\|_2^2.
    \end{equation}
    In particular, as $d_cd_q^2 \geq 2$, we have
    \begin{equation}
        \mathbb{E}_{\bfU}[\|\mathcal{I}_{\bfU}(X)\|_2^2]\geq \frac{1}{2}\cdot\frac{d_q}{d}\|X\|_2^2.
    \end{equation}
\end{lemma}
\begin{lemma}\label{lemma:upsilon_second_moment}
    There is an absolute constant $C > 0$, such that for every traceless Hermitian $X$, 
    \begin{equation}
        \mathbb{E}_{\bfU}[\|\mathcal{I}_{\bfU}(X)\|_2^4]\leq C\left(\frac{d_q}{d}\right)^2\|X\|_2^4.
    \end{equation}
\end{lemma}
Assuming these two results, we prove \Cref{thm:Haar_compression}.
\begin{proof}[Proof of \Cref{thm:Haar_compression}]
    By \Cref{lemma:first_moment_upsilon}, $\mathbb{E}[\|\mathcal{I}_{\bfU}(X)\|_2^2]\geq \frac{d_q}{2d}\|X\|_2^2$. Then, by \Cref{lemma:upsilon_second_moment}, \smash{$\mathbb{E}[\|\mathcal{I}_{\bfU}(X)\|_2^4]\leq C\cdot \frac{d_q^2}{d^2}\|X\|_2^4$}. Applying the Paley-Zygmund inequality (\Cref{lemma:Paley_Zygmund}) with $\theta=1/2$, we get
    \begin{equation}
        \Pr[\|\mathcal{I}_{\bfU}(X)\|_2^2\geq \frac{1}{2}\mathbb{E}[\|\mathcal{I}_{\bfU}(X)\|_2^2]]\geq \frac{1}{4}\cdot\frac{1}{4C}=\frac{1}{16C}.
    \end{equation}
    Then, on this event, we have that
    \begin{equation}
        \|\mathcal{I}_{\bfU}(X)\|_2^2\geq\frac{1}{2}\mathbb{E}[\|\mathcal{I}_{\bfU}(X)\|_2^2]\geq \frac{d_q}{4d}\|X\|_2^2, 
    \end{equation}
    as required.
\end{proof}
\subsubsection{A Public Coin Algorithm}
\label{ss:public-coin-algorithm}
With all that setup complete, we can finally give our algorithm in the public-coin setting, with both classical and quantum communication, shown in \Cref{alg:public_coin}.
\begin{algorithm}[H]
    \caption{Public-coin distributed algorithm with both classical and quantum communication}\label{alg:public_coin}
    \begin{algorithmic}[1]
        \State Set $L=\bigo(1)$, $\tau= \frac{1}{4}\frac{d_q\epsilon^2}{d^2}$, and $t= \bigo\left(\frac{1}{\sqrt{d_c}\tau}\right)$. Set $m=Lt$; group nodes by indices $i\in[L]$.
        \State Using public randomness, all nodes jointly sample $L$ different Haar random unitary matrices, $\bfU_1,...,\bfU_L$, and create the corresponding quantum instruments $\mathcal{I}_{\bfU_i}$.
        \State Each distributed node sends $\mathcal{I}_{\bfU_i}(\rho)$ to the central node $N_c$.
        \State $N_c$, with full knowledge of $\sigma$, computes descriptions $\sigma_{i}=\mathcal{I}_{\bfU_i}(\sigma)$ for each sampled $\bfU_i$. 
        \State\label{algo:publiccoin:checkspreadout} $N_c$ then checks, with full knowledge of $\sigma$ and the quantum instruments, that the classical labels are sufficiently spread out over all strings, meaning each string has mass at most $\bigo(1/d_c)$. If not, record \reject~for that group $i\in[L]$. Assume $L'$ groups pass this phase. 
        \For {$i$ in the remaining $L'$ indices}
            \State Take all messages corresponding to the $i$-th quantum instrument, and compute the estimate $\hat{\Gamma}$ of $\|\rho - \sigma\|_2^2$ via \Cref{lemma:cq_hs_estimator_variance}.
            \If {$\hat\Gamma> \tau/2$}
                \State record \reject,
            \Else \State record \accept.
            \EndIf
        \EndFor
        \State\Return \accept{} if the majority of the groups accept, else \reject.
    \end{algorithmic}
\end{algorithm}

\begin{remark}
    Note that \Cref{alg:public-coin-quantum,alg:public_coin} involve unitaries drawn from the Haar measure over $U(d)$. However, our analysis only involves up to fourth-order moments of these random unitaries. Consequently, the procedure can be made more efficient by instead drawing these unitaries from any (approximate) unitary $4$-design.
\end{remark}

In Line~\ref{algo:publiccoin:checkspreadout}, the central node has to verify whether the quantum instruments sufficiently spread out $\sigma$. For the sake of analysis, we assume that all sampled unitaries produce good instruments, meaning it sufficiently spreads out over all classical labels for both $\rho$ and $\sigma$. By \Cref{cor:Haar_double_good_spread}, we know this occurs with good probability, and hence we can always just adjust the number of repetitions $L$, costing only a constant factor in the overall complexity. We then have our tight upper bound for testing with public coins:

\begin{theorem}\label{thm:public_coin_ub}
    Assume the distributed nodes are allowed to communicate a $d_q$-dimensional qudit, and an $n_c$-length classical string, and all parties share public coins. Then, to $\epsilon$-certify a $d$-dimensional state with probability at least $2/3$, it suffices to take $m= \bigo\left(\frac{d^2}{d_q\sqrt{d_c}\epsilon^2} \right)$, i.e., in the $(n_c,n_q,\mathsf{public})$ setting,
    \begin{equation}
        m= \bigo\left(\frac{d^2}{2^{n_q+n_c/2}\epsilon^2}\right)
    \end{equation}
    copies are sufficient.
\end{theorem}
\begin{proof}
    Let $X=\rho-\sigma$. Then, by \Cref{thm:Haar_compression}, with constant probability over a public-coin choice of unitary $U$, we have that
    \begin{equation}
        \|\mathcal{I}_U(X)\|_2^2\geq \frac{d_q}{4d}\|X\|_2^2.
    \end{equation}
    If $\rho=\sigma$, this quantity equals zero, while if $\|\rho-\sigma\|_1 \geq \epsilon$ this is at least $\frac{1}{4}\cdot\frac{d_q\epsilon^2}{d^2}$. Recall that the algorithm set a threshold at $\tau/2$, where $\tau=\frac{1}{4}\cdot\frac{d_q\epsilon^2}{d^2}$. Next, by \Cref{cor:Haar_double_good_spread} we know that with good probability, the quantum instrument $I_{\bfU_i}$ outputs CQ states $\rho_i$ and $\tau_i$ that are spread over all classical strings. In other words, we have that 
    \begin{equation}
        \|p_i\|_2^2 , \|q_i\|_2^2 \leq \bigo(1/d_c).
    \end{equation}
    Going forward, let us condition on these good events for the $i$th test.
    Then, by \Cref{thm:cq_hs_estimator}, the CQ Hilbert--Schmidt estimator has variance
    \begin{equation}
        \mathbb{V}[\hat\Gamma_i]\leq O\left(\frac{\Gamma_i}{\sqrt{d_c}t}+\frac{1}{d_ct^2}\right).
    \end{equation}
    Now we consider two separate cases. First, assume $\rho=\sigma$, so $\mathbb{V}[\hat\Gamma_i]\leq \bigo\left(\frac{1}{d_ct^2}\right)$. Then, by Chebyshev's inequality, we have that
    \begin{equation}
        \Pr[\hat\Gamma_i\geq \tau/2]\leq \bigo\left(\frac{\mathbb{V}[\hat\Gamma_i]}{\tau^2}\right)=\bigo\left(\frac{1}{d_ct^2\tau^2}\right),
    \end{equation}
    which is at most $1/3$ by setting the constant in $t$ large enough. Considering the other case, if $\|\rho-\sigma\|_1\geq\epsilon$, then $\Gamma_i\geq\tau$, so $\Gamma_i-\tau/2\geq \Gamma_i/2$. Then applying Chebyshev's in this case, we get that
    \begin{align}
        \Pr[\hat\Gamma_i<\tau/2]&\leq\Pr[|\hat\Gamma_i-\Gamma_i|\geq \Gamma_i/2]\\
        &\leq \bigo\left(\frac{\mathbb{V}[\hat\Gamma_i]}{\Gamma_i^2}\right)\\
        &=\bigo\left(\frac{1}{\sqrt{d_c}t\Gamma_i}+\frac{1}{d_ct^2\Gamma_i^2}\right)\\
        &\leq \bigo\left(\frac{1}{\sqrt{d_c}t\tau}+\frac{1}{d_ct^2\tau^2}\right),
    \end{align}
    which again is at most $1/3$ by increasing the constant in $t$. Overall, the $i$th test has at least constant success probability. Using a boosting argument one can amplify to arbitrary constant probability of success. Finally, the total number of nodes used here is
    \begin{equation}
        m=Lt=\bigo\left(\frac{d^2}{d_q^2\sqrt{d_c}\epsilon^2} \cdot\right),
    \end{equation}
    establishing the theorem. 
\end{proof}

\section{Private-Coin Upper Bound}\label{s:private_ub}

In this section, we prove our upper bounds for the private-coin setting. The key difficulty in translating arguments from the public-coin setting here is that the distributed nodes can no longer sample shared random unitaries. Nevertheless, we will show that one can essentially \emph{derandomize} this step; in particular, we show that all nodes can agree on a pre-determined set of ``good'' unitaries. The tradeoff, of course, will be that we will require a much larger number of copies to succeed at certification in this private-coin setting. Formally, we obtain the following upper bound:

\begin{theorem}[Private-Coin Mixed Communication Upper Bound]
\label{thm:private-coin-upper-general}
    Fix any $ n_c,n_q \geq 0$ with $1 \leq n_c + n_q \leq \lceil\log_2 d\rceil$. Let $d_c \triangleq 2^{n_c}, d_q \triangleq 2^{n_q}$. Then, there exists a protocol for $\epsilon$-certification of $d$-dimensional states in the $(n_c,n_q,\textnormal{private})$ model with copy complexity
    \begin{equation}
        m = \bigo\left(
            \frac{d^3 \sqrt{\log d}}{d_q^2 d_c \epsilon^2}
        \right).
    \end{equation}
\end{theorem}

Before we prove this theorem, we will first consider the case where $n_c = 0$, i.e., one can only communicate quantum messages. In this case, we obtain an upper bound that is tight up to a $\sqrt{\log d}$ factor. Though this result can be seen as a special case of \Cref{thm:private-coin-upper-general}, its proof serves as a warmup providing intuition necessary for the more general proof. We present the warmup result and its proof in the subsequent section, with the proof of the more general \Cref{thm:private-coin-upper-general} in \Cref{ss:quantum_classical_private_ub}.

\subsection{Warmup 2: Private-Coin Testing with Only Quantum Communication}\label{ss:quantum_private_ub}
In this section, we first provide a near-optimal algorithm for the setting when only quantum communication is permitted. We  will split the $d$ dimensional register into two registers, $A$ and $Q$, with dimensions $d_A$ and $d_q$, such that $d=d_Ad_q$. The main result we will prove in this section is as follows:

\begin{theorem}[Private-Coin No Classical Communication Upper Bound]\label{thm:private_coin_no_classical_ub}
    Fix any $1\leq n_q\leq \lceil \log_2 d\rceil$. There is a protocol for $\epsilon$-certification of $d$-dimensional states in the $(0,n_q,\text{private})$ model with copy complexity
    \begin{equation}
        m= \bigo\left(\frac{d^3}{2^{2n_q}\epsilon^2}\cdot \sqrt{\log d}\right).
    \end{equation}
\end{theorem}

As alluded to previously, the main difficulty in transitioning to the private coin model is that we can no longer sample random unitaries when the distributed nodes are restricted their own private randomness. The key result we rely on is a surprising claim that the distributed nodes actually do not need to sample any unitaries, and instead can agree beforehand on a set of ``good'' unitaries. We state this derandomization result first, and defer its proof to Section~\ref{ss:proof-derandomization}.

\begin{theorem}\label{thm:good_unitaries_private_only_ub}
     For $L=\bigo\left(\frac{d^2}{d_q^2}\log d\right)$, there exist unitaries $U_1,...,U_L$ such that, for every traceless Hermitian $X\in\mathbb{C}^{d\times d}$,
    \begin{equation}
        D\triangleq \frac{1}{L}\sum_{i=1}^L\|\Tr_A(U_iXU_i^\dagger)\|_2^2\geq \frac{1}{4}\cdot \frac{d_q}{d}\|X\|_2^2.
    \end{equation}
\end{theorem}

We also use, as a black-box subroutine, the BOW estimator introduced in \Cref{lem:hscertify}. Finally, for any unitary $U$, we define a quantum channel $\Phi_U(X)$ as the channel
\begin{equation}
    \Phi_U(X)=\Tr_A(UXU^\dagger).
\end{equation}
Then, the algorithm for this setting is as follows in \Cref{alg:private_coin_no_classical}:
\begin{algorithm}[H]
    \caption{Private-Coin Distributed Algorithm for Certification}\label{alg:private_coin_no_classical}
    \begin{algorithmic}[1]
        \State All nodes agree on a good set of unitaries $U_1,...,U_L$ satisfying \Cref{thm:good_unitaries_private_only_ub}.
        \State Set $L=\bigo(\frac{d^2}{d_q^2}\log d)$, $t=\bigo(\frac{d}{\sqrt{\log d}\epsilon^2})$, $m=Lt$, and $\tau=\frac{d_q}{8d^2}{\epsilon^2}$; group nodes by indices $i\in[L]$, so each node has a label $N_{i,j}$ for $i\in[L]$ and $j\in[t]$.
        \State For $j\in[t]$, $N_{i,j}$ sends $\Phi_{U_i}(\rho)$ to the central node $N_c$.
        \State The central node receives $t$ copies of the state $\Phi_{U_i}$ for each $i\in [L]$.
        \State $N_c$ applies the BOW estimator in \Cref{lem:hscertify} to each group of $t$ copies, obtaining estimates of $\|\Phi_{U_i}(\Delta)\|_2^2$, where $\Delta=\rho-\sigma$.
        \State $N_c$ averages these estimates and compares the resultant quantity $\hat{D}$ to the threshold $\tau$; accept if $\hat{D}<\tau$, reject otherwise.
    \end{algorithmic}
\end{algorithm}
We make a few comments about \Cref{alg:private_coin_no_classical}. First, due to pre-agreeing on the set of unitaries used, this algorithm is actually \emph{deterministic} in the choice of channels, although it is not completely deterministic, due to the use of the randomness inherent to the BOW estimator. Next, the main observation which we exploit in the algorithm is as follows:
\begin{itemize}
    \item If $\rho=\sigma$, this implies for every unitary $U_i$, $\|\Phi_{U_i}(\Delta)\|_2^2=0$, and hence $D=0$. This we refer to as the \emph{completeness} case.
    \item Otherwise, if $\|\rho-\sigma\|_1\geq \epsilon$, by Cauchy-Schwarz, $\|\Delta\|_2\geq\frac{\|\Delta\|_1}{\sqrt{d}}\geq \frac{\epsilon}{\sqrt{d}}$. Hence, applying \Cref{thm:good_unitaries_private_only_ub}, we have that 
    \begin{equation}
        D\geq \frac{1}{4}\cdot\frac{d_q\epsilon^2}{d^2}=2\tau.
    \end{equation}
    This we refer to as the \emph{soundness} case.
\end{itemize}
Hence, to perform $\epsilon$-certification, it is sufficient to distinguish between the above two cases. In essence, this is what \Cref{alg:private_coin_no_classical} is doing. We now prove that it satisfies \Cref{thm:private_coin_no_classical_ub}, assuming for now the truth of \Cref{thm:good_unitaries_private_only_ub}.
\begin{proof}[Proof of \Cref{thm:private_coin_no_classical_ub}]
Let $\hat{D}_i$ denote the output of the BOW estimator on group $i$ of the distributed nodes. By \Cref{lem:hscertify}, we have that
\begin{equation}
    \mathbb{V}[\hat{D}_i] \leq \bigo\left(\frac{1}{t^2}+\frac{\mathbb{E}[\hat{D}_i]}{t}\right).
\end{equation}
Due to the estimator being unbiased, we have that $\mathbb{E}[\hat{D}_i]=\|\Phi_{U_i}(\Delta)\|_2^2$. Then, the overall variance of the estimate $\hat{D}$ is given by
\begin{equation}
    \mathbb{V}[\hat{D}]\leq \frac{1}{L^2}\sum_{i=1}^L\mathbb{V}[\hat{D}_i]\leq \bigo\left(\frac{1}{Lt^2}+\frac{D}{Lt}\right).
\end{equation}
We next consider the two cases outlined above. In the completeness case, if $\rho=\sigma$, we have that 
\begin{equation}
    \mathbb{V}[\hat{D}]\leq \bigo\left(\frac{1}{Lt^2}\right).
\end{equation}
Then, by Chebyshev's inequality,
\begin{equation}
    \Pr[\hat{D}\geq\tau]\leq\bigo\left(\frac{\mathbb{V}[\hat{D}]}{\tau^2}\right)=\bigo\left(1\right),
\end{equation}
which is at most $1/3$ by adjusting the constants in front of $L$ and $t$. Else in the soundness case, if $\|\rho-\sigma\|_1\geq\epsilon$, we have that $D\geq 2\tau$, and thus $D-\tau\geq D/2$. Then, we can again apply Chebysev's inequality to conclude that 
\begin{align}
    \Pr[\hat{D} <\tau]&\leq \Pr[|\hat{D}-D|\geq D/2]\\
    &\leq \bigo\left(\frac{\mathbb{V}[\hat{D}]}{D^2}\right)\\
    &\leq \bigo\left(\frac{1}{LtD}+\frac{1}{Lt^2D^2}\right)\\
    &\leq \bigo\left(\frac{1}{Lt\tau}+\frac{1}{Lt^2\tau^2}\right)\\
    &\leq\bigo\left(\frac{d_q}{d\sqrt{\log d}}+1\right)\\
    &\leq \bigo\left(\frac{1}{\sqrt{\log d}}\right),
\end{align}
which is at most $1/3$ for large enough $d$. Hence, \Cref{alg:private_coin_no_classical} succeeds at distinguishing $D=0$ and $D\geq \tau$ with probability at least $2/3$, and hence succeeds in $\epsilon$-certifying with probability at least $2/3$. The number of nodes used is
\begin{equation}
    m=Lt=\bigo\left(\frac{d^2}{d_q^2}\sqrt{\log d}\frac{d}{\epsilon^2}\right)=\bigo\left(\frac{d^3}{d_q^2\epsilon^2}\sqrt{\log d}\right),
\end{equation}
as claimed.
\end{proof}

This proof is now complete, conditional on \Cref{thm:good_unitaries_private_only_ub}, which we now prove.
\subsubsection{Proof of \Cref{thm:good_unitaries_private_only_ub}}\label{ss:proof-derandomization}
We now aim to show that there exists a (small) set of good unitaries which satisfy our requirements. To do so, we will use the probabilistic method, and show a set of Haar random unitaries satisfies this theorem with positive probability. Then, this shows the existence of such a set of unitaries, which we can then find, and use as the set of pre-agreed unitaries for \Cref{alg:private_coin_no_classical}. Before going into the proof, it will be helpful to lay out some notation and some intermediate results.

\begin{definition}
    We will let
    \begin{itemize}
    \item $H=\{X=X^\dagger, \Tr(X)=0, X\in\mathbb{C}^{d\times d}$ be the space of $d$-dimensional traceless Hermitian operators. Note since $X$ has dimension $d^2$, $H$ has dimension $d^2-1$.
    \item $W=\{Y\otimes I_A: Y\in\mathbb{C}^{d_q\times d_q}, Y=Y=Y^\dagger, \Tr(Y)=0\}\subseteq{H}$ be the subset of traceless Hermitian operators that are identity on the junk register $A$. For the same reason as above, this space has dimension $d_q^2-1$.
    \item Let $\Pi_W$ be the orthogonal projector onto $W$.
    \item For any unitary $U$, let $W_U=\{U^\dagger M U, M\in W\}$, and let $\Pi_W$ be the orthogonal projector onto $W_U$.
\end{itemize}
\end{definition}

We first prove an easy lemma about how the projector impacts norms:
\begin{lemma}\label{lem:W_projector_private_ub_no_classical}
    For any $X\in H$, we have that $\|\Tr_A[X]\|_2^2=\frac{d}{d_q}\|\Pi_WX\|_2^2$.
\end{lemma}
\begin{proof}
    Since $X$ is traceless, we have that $\Tr_A(\Tr_Q(X))=\Tr_{AQ}(X)=0$. Next, we can consider the projection, which is $\Pi_W(X)=\frac{d}{d_q}\Tr_A(X)\otimes I_A\in W$, noting that because of the above, this is still traceless and the coefficient coming from the fact that $d_A=d/d_q$. Taking norms, we now get
    \begin{align}
        \|\Pi_W(X)\|_2^2&=\||\frac{d}{d_q}\Tr_A(X)\otimes I_A\|_2^2
        =\frac{d}{d_q}\|\Tr_A(X)\|_2^2\,,
    \end{align}
    as claimed.
\end{proof}

Next, we analyze how the action of a channel $\Phi_U$ impacts the norm. Specifically, we show:
\begin{lemma}\label{lem:channel_norm_private_ub_no_classical}
    For any unitary $U$,
    \begin{equation}
        \|\Phi_U(X)\|_2^2=\frac{d}{d_q}\|\Pi_UX\|_2^2.
    \end{equation}
\end{lemma}
\begin{proof}
    For a unitary $U$, if $R_U$ denotes the conjugation map $R_U(X)=UXU^\dagger$, then the map $\Pi_U=R^{-1}_U\Pi_WR_U$, which can be verified just by calculation. This then means that $\Pi_UX=U^\dagger\Pi_W(UXU^\dagger)U$. Then, taking norms, and using that $U$ and $U^\dagger$ do not change norms, we get that $\|\Pi_UX\|_2=\|\Pi_W(UXU^\dagger)\|_2$. We denote $Z=UXU^\dagger$. Then, recalling that $\Phi_U(X)=\Tr_A(UXU^\dagger)=\Tr_B(Z)$, we get from \Cref{lem:W_projector_private_ub_no_classical} that $\|\Phi_U(X)\|_2^2=\|\Tr_A[Z]\|_2^2=\frac{d}{d_q}\|\Pi_WZ\|_2^2$. This then gives that $\|\Phi_U(X)\|_2^2=\frac{d}{d_q}\|\Pi_UX\|_2^2$, as required.
\end{proof}
It is also helpful at this point to recall some properties of the projector $\Pi_U$. Namely, since it is an orthogonal projector, it has only eigenvalues $0$ and $1$, which also implies $0\preceq \Pi_U\preceq I$. Finally, since it projects onto a subspace of dimension $d_q^2-1$, we have that $\Tr_H[\Pi_U]=d_q^2-1$. Here, the trace is of $\Pi_U$ as a linear operator, and not of the quantum state.
The above properties all establish something about a particular $U$. However, we wish to sample a list of random matrices, and then say something about their property. Hence, we now wish to derive properties about the expectation of this projector $\Pi_U$. We will make use of Schur's Lemma, a standard result in representation theory, see for instance \cite[Corollary 1.17]{etingof2011introductionrepresentationtheory}:
\begin{lemma}[Schur's Lemma]\label{lem:Schur}
    Let $V$ be an algebraically closed field. Then, endomorphisms $\phi\colon V\to V$ of finite dimensional irreducible representations of $V$ are scalar multiples of the identity operator.
\end{lemma}
We make a few comments on the above. First, for completeness, it is stated more generally than we needed. Indeed, we are trivially in an algebraically closed field, and all endomorphisms will also trivially be finite dimensional, just by our setting.  Furthermore, it is worth pointing out again that the identity operator is not the identity matrix. $I_H$ is the identity map that sends an element of $H$ to itself. We use one final fact of traceless Hermitian matrices:
\begin{fact}\label{fact:traceless_hermitian_operators}
    The conjugation action of traceless Hermitian matrices is irreducible.
\end{fact}
To see this, first note that the adjoints are ideals of the traceless matrices \cite{AlAssal2014InvitationLieAlgebras}, which is a simple space, and hence irreducible \cite{Morgan2022SimpleLieAlgebrasCompactLieGroups}. We can now state and prove the main lemma:
\begin{lemma}\label{lem:projector_expectation_private_ub_no_classical}
Taking expectation with respect to a Haar-random unitary $\bfU$, we have
    \begin{equation}
        \mathbb{E}_{\bfU}[\Pi_{\bfU}]=\frac{d_q^2-1}{d^2-1}I_{H}.
    \end{equation}
\end{lemma}
\begin{proof}
    For shorthand, let $T=\mathbb{E}_{\bfU}[\Pi_{\bfU}]$. For any fixed Haar unitary matrix $V$, let $R_V(X)=VXV^\dagger$. Notice that
    \begin{align}
        R_V W_U&=V(U^\dagger W U)V^\dagger\\
        &=(UV^\dagger)^\dagger W(UV^\dagger)\\
        &=W_{UV^\dagger}.
    \end{align}
    This implies that $R_V\Pi_UR_V^{-1}=\Pi_{UV^\dagger}$. By Haar-measure invariance, when $\bfU$ is Haar-random, then so is $\bfU V^\dagger$. Hence,
    \begin{align}
        R_VTR_V^{-1}&=R_V\mathbb{E}_{\bfU}[\Pi_{\bfU}]R_V^{-1}\\
        &=\mathbb{E}_{\bfU}[R_V\Pi_{\bfU}R_V^{-1}]\\
        &=\mathbb{E}_{\bfU}[\Pi_{\bfU V^\dagger}]\\
        &=T.
    \end{align}
    Hence, for all unitaries $V$, $R_V T=T R_V$, meaning that $T$ commutes with every conjugation action on $H$. Then, by \Cref{fact:traceless_hermitian_operators}, this is irreducible as a representation. Then, this means we can apply Schur's lemma, and $T$ must be a scalar multiple of the identity, so $T=\lambda I_H$. To find the value of $\lambda$, we take traces over the space $H$. First, notice that 
    \begin{align}
        \Tr_H(T)&= \Tr(\mathbb{E}[\Pi_{\bfU}])=\mathbb{E}_{\bfU}[\Tr_H(\Pi_{\bfU})].
    \end{align}
    Since each $\Pi_U$ has rank $d_q^2-1$, we have that $\Tr_H[T]=d_q^2-1$. On the other hand, we can take
    \begin{equation}
        \Tr_H[T]=\lambda\Tr_H[I_H]=\lambda(d^2-1).
    \end{equation}
    Then, equating the two, we get that 
    \begin{align}
        \lambda&=\frac{d_q^2-1}{d^2-1}\\
        \mathbb{E}_{\bfU}[\Pi_{\bfU}]&=\frac{d_q^2-1}{d^2-1}I_H.
    \end{align}
\end{proof}

We will now show that empirical averages of independently drawn random projectors $\Pi_{\bfU_1}, \dots, \Pi_{\bfU_L}$ are very likely to have concentrated eigenvalues; this will allow us to argue that inner products with such averaged operators are approximated, up to constant factors, by those with the mean in the above lemma. To do this, we require the matrix Chernoff bound \cite[Corollary 5.2 \& Remark 5.3]{Tropp_2011}:

\begin{theorem}[Matrix Chernoff Inequality]\label{thm:matrix_chernoff_inequality}
    Consider a finite sequence $\{X_k\}$ of $d$-dimensional, independent, random, self-adjoint matrices such that each satisfies $X_k\succeq 0$ and $\lambda_{\max}(X_k)\leq R$ almost surely, where $\lambda_{\max}$ denotes the largest eigenvalue. Then, letting $\lambda_{\min}$ denote the smallest eigenvalue, define
    \begin{equation}
        \mu_{\min} \triangleq\lambda_{\min}\left(\sum_k\mathbb{E}[X_k]\right),~\mu_{\max} \triangleq \lambda_{max}\left(\sum_k\mathbb{E}[X_k]\right).
    \end{equation}
    Then, for $0<s<1$,
    \begin{equation}
    \label{eq:matrix-chernoff-operator-lower}
        \Pr\left[\lambda_{\min{}}\left(\sum_k X_k\right)\leq (1-s)\mu_{\min{}}\right]\leq d\cdot\exp(-s^2\mu_{\min{}}/2R),
    \end{equation}
    and, for all $t \geq e$,
    \begin{equation}
       \label{eq:matrix-chernoff-operator-upper}
       \mathbf{Pr}\left[\lambda_{\max}\left(\sum_k \bfX_k\right) \geq t \mu_{\max}\right] \leq d \cdot \left(\frac{e}{t}\right)^{t\mu_{\max}/R}. 
    \end{equation}
\end{theorem}

For this warmup section, we only need \Cref{eq:matrix-chernoff-operator-lower} appearing above. The latter \Cref{eq:matrix-chernoff-operator-upper} will be necessary only for the more general proof in \Cref{ss:quantum_classical_private_ub}.

We can now proceed to apply this Matrix Chernoff inequality to our list of unitaries. Recall that since we will be applying this to projector matrices, we trivially satisfy the required conditions for $R=1$. We will now aim to prove the following result:
\begin{lemma}\label{lem:deterministic_good_choice_private_no_classical}
    For $L=\bigo\left(\frac{d^2}{d_q^2}\log d\right)$,  there exist unitaries $U_1,...,U_L$ such that 
    \begin{equation}
        \frac{1}{L}\sum_{r=1}^L\Pi_{U_r}\succeq\frac{1}{2}\frac{d_q^2-1}{d^2-1}I.
    \end{equation}
\end{lemma}
\begin{proof}
    Recall by \Cref{lem:projector_expectation_private_ub_no_classical}, for each projector we have that $\mathbb{E}[P_{\bfU_r}]=\frac{d_q^2-1}{d^2-1}I$, and hence, 
    \begin{equation}
        \mathbb{E}\left[\sum_{r=1}^L\Pi_{\bfU_r}\right]=L\frac{d_q^2-1}{d^2-1}I.
    \end{equation}
    Then, this also implies that $\mu_{\min{}}=L\frac{d_q^2-1}{d^2-1}$, since $I$ only has eigenvalue $1$. Then, applying \Cref{eq:matrix-chernoff-operator-lower} with $s=0.5$, we obtain
    \begin{equation}
        \Pr\left[\lambda_{\min}\left(\sum_{r=1}^L\Pi_{\bfU_r}\right)\leq \frac{1}{2}L\frac{d_q^2-1}{d^2-1}\right]\leq (d^2-1)\exp\left(-\frac{1}{8}L\frac{d_q^2-1}{d^2-1}\right).
    \end{equation}
    We want to make this probability less than $1$, which implies that we need to choose $L= \bigo\left(\frac{d^2}{d_q^2}\log d\right)$. Then, with positive probability, we obtain a set of $L$ projectors such that the minimum eigenvalue of the sum is large enough, meaning such unitaries $U_1,...,U_L$ with
    \begin{equation}
        \sum_{r=1}^L\Pi_{U_r}\succeq\frac{1}{2}L\frac{d_q^2-1}{d^2-1}I
    \end{equation}
    exist, as we aimed to prove.
\end{proof}

We finally have all the ingredients to prove \Cref{thm:good_unitaries_private_only_ub}.
\begin{proof}[Proof of~\Cref{thm:good_unitaries_private_only_ub}]
    Take any $X\in H$, and fix a good list of unitaries $U_1,..,U_L$, satisfying \Cref{lem:deterministic_good_choice_private_no_classical}. Then, we have that
    \begin{equation}
        \left\langle X, \left(\sum_{r=1}^L\Pi_{U_r}\right)X\right\rangle\geq \frac{1}{2}L\frac{d_q^2-1}{d^2-1}\langle X,X\rangle.
    \end{equation}
    The right hand side equals $\frac{1}{2}L\frac{d_q^2-1}{d^2-1}\|X\|_2^2$. Using the fact that $\Pi_{U_r}$ is an orthogonal projection, we have that 
    \begin{align}
        \left\langle X, \left(\sum_{r=1}^L\Pi_{U_r}\right)X\right\rangle&=\sum_{r=1}^L \langle X, \Pi_{U_r}X\rangle = \sum_{r = 1}^L \langle \Pi_{U_r}X,\Pi_{U_r}X\rangle = \sum_{r = 1}^L \|\Pi_{U_r} X\|_2^2.
    \end{align}
    Thus, putting everything together and dividing by $L$, we get that
    \begin{equation}
        \frac{1}{L}\sum_{r=1}^L\|\Pi_{U_r}X\|_2^2\geq \frac{1}{2}\frac{d_q^2-1}{d^2-1}\|X\|_2^2.
    \end{equation}
    Then, using \Cref{lem:channel_norm_private_ub_no_classical}, we have that
    \begin{equation}
        \frac{1}{L}\sum_{r=1}^L\|\Phi_{U_r}(X)\|_2^2\geq\frac{d}{d_q}\cdot\frac{1}{2}\cdot\frac{d_q^2-1}{d^2-1}\|X\|_2^2.
    \end{equation}
    Notice that for $d_q\geq 2$, we have that $d_q^2-1\geq \frac{1}{2}d_q^2$, and also $d^2-1\leq d^2$. Thus, rearranging, we get that 
    \begin{equation}
         \frac{1}{L}\sum_{r=1}^L\|\Phi_{U_r}(X)\|_2^2\geq\frac{1}{4}\cdot\frac{d_q}{d}\|X\|_2^2,
    \end{equation}
    as required. The final thing to point out is that we required $L$ large enough so that \Cref{lem:deterministic_good_choice_private_no_classical} held, meaning we take $L=\bigo\left(\frac{d^2}{d_q^2}\log d\right)$, completing the proof.
\end{proof}

\subsection{Classical and Quantum Communication}\label{ss:quantum_classical_private_ub}

We will now present our main private-coin algorithm which achieves \Cref{thm:private-coin-upper-general}. Our idea will be to derandomize the public-coin algorithm from \Cref{ss:public_ub} using ideas from the warmup section. We will split the $d$-dimensional register into three registers A, Q, and C, with dimensions $d_A,d_q,d_c$ resp., such that $d = d_Ad_qd_c$\footnote{As before, indivisibility of $d$ by $d_qd_c$ can be handled by a higher-dimensional embedding while incurring only a constant-factor loss}. For a unitary $U \in U(d)$, we define the instrument $\mathcal{I}_{U}$ mapping $d$-dimensional states to classical-quantum states over registers $Q,C$ by its action on a state $\rho$:
\begin{equation}
    \mathcal{I}_U(\rho) = \sum_{c = 1}^{d_c} \mathcal{I}_{U,c}(\rho)_Q \otimes \ketbra{c}{c}_C, \quad \textnormal{ where } \mathcal{I}_{U,c}(\rho) = \Tr_{A,C} [U\rho U^\dagger \cdot I_A \otimes I_Q \otimes \ketbra{c}{c}_C].
\end{equation}

We will consider finitely many unitaries $U_1. \dots, U_L$ satisfying certain desirable properties. To state these properties, let us first define some preliminary notation. For $i \in [L]$, we define $p_i = \Tr_{Q}(\mathcal{I}_{U_i}(\rho))$, and similarly $q_i$ for $\sigma$; these are diagonal matrices representing the probability marginals on the classical register. We will follow this notation throughout this section.

We now prove the existence of $L$ deterministic unitaries that achieve all properties required of the randomly sampled unitaries in the public-coin upper bound. This is captured by the following theorem, whose proof we defer: 
\begin{theorem}
\label{thm:derandomization-classical-quantum}
    For all quantum states $\rho,\sigma$, and  $L = \bigo(\frac{d^2}{d_q^2 d_c} \log d)$, there exist $d$-dimensional unitaries $U_1, \dots, U_L$ satisfying 
    \begin{equation}
    \label{eq:derandomized-guarantee-spread}
        \frac1L \sum_{i = 1}^L \|p_i\|_2^2 \leq \frac{16}{d_c}, \quad \frac1L \sum_{i = 1}^L \|q_i\|_2^2 \leq \frac{16}{d_c},
    \end{equation}
    and,
    \begin{equation}
            \label{eq:derandomized-guarantee-compression}
            \frac1L \sum_{i = 1}^L \|\mathcal{I}_{U_i}(\rho - \sigma)\|_2^2 \geq \frac14 \frac{d_q}{d} \|\rho - \sigma\|_2^2.
        \end{equation}
\end{theorem}

Given the guarantees of~\Cref{thm:derandomization-classical-quantum} we can now construct~\Cref{alg:private-coin-general}, a private-coin algorithm for state certification.

\begin{algorithm}
    \begin{algorithmic}[h]
        \caption{Private-Coin Certification Algorithm with Classical and Quantum Communication}
        \label{alg:private-coin-general}
        \State All nodes agree on $L$ unitaries $U_1, \dots, U_L$ as in \Cref{thm:derandomization-classical-quantum}.
        \State Set $t = \bigo(\frac{d}{\sqrt{\log d}\epsilon^2})$, $m = Lt$ and $\tau = \frac{d_q}{8d^2} \epsilon^2$; group nodes by indices $i \in [L]$.
        \State For $j \in [t]$, $N_{i,j}$ sends $\mathcal{I}_{U_i}(\rho)$ to the central node $N_c$.
        \State $N_c$ now possesses $t$ copies of the state $\mathcal{I}_{U_i}(\rho)$ for each $i \in [L]$. 
        \State $N_c$ applies the classical-quantum BOW estimator in \Cref{thm:cq_hs_estimator} to each group of $t$ copies, obtaining estimates of $\|\mathcal{I}_{U_i}(\Delta)\|_2^2$, where $\Delta = \rho - \sigma$.
        \State $N_c$ averages these estimates and compares the resultant quantity $\widehat{D}$ to the threshold $\tau$; $\accept$ if $\widehat{D} < \tau$, $\reject$ otherwise.
    \end{algorithmic}
\end{algorithm}

We now show that~\Cref{alg:private-coin-general} achieves the complexity in \Cref{thm:private-coin-upper-general}.

\begin{proof}[Proof of \Cref{thm:private-coin-upper-general}]
    For $L = \bigo(\frac{d^2}{d_q^2d_c} \log d)$, let $U_1, \dots, U_L$ be the unitaries in \Cref{thm:derandomization-classical-quantum}. Set $t = \bigo(\frac{d}{\sqrt{\log d}\epsilon^2})$. For $\Delta = \rho - \sigma,$ we will be concerned with the quantity $D \triangleq \frac1L \sum_{i = 1}^L \|\mathcal{I}_{U_i}(\Delta)\|_2^2$. When $\rho = \sigma$, $D = 0$. Otherwise, when $\|\rho - \sigma\|_1 \geq \epsilon$, by \Cref{thm:derandomization-classical-quantum}, we have
    \begin{equation}
        \|\rho - \sigma\|_2^2 \geq \frac{\epsilon^2}{d} \implies D \geq \frac{d_q}{4d^2}\epsilon^2. 
    \end{equation}
    We will thus compare our estimate $\widehat{D}$ to the threshold 
    \begin{equation}
        \tau \triangleq \frac{d_q}{8d^2} \epsilon^2.
    \end{equation}

    By \Cref{thm:cq_hs_estimator}, the variance of the $i$th estimator satisfies
    \begin{equation}
        \mathbb{V} [\widehat{D}_i] \leq \bigo\left(
            \frac{\|p_i\|_{\infty} D_i}{t} + \frac{\|p_i\|_2^2 + \|q_i\|_2^2}{t^2}
        \right).
    \end{equation}

    Averaging the $L$ estimates results in the following variance bound:
    \begin{align}
        \mathbb{V}[\widehat{D}] 
        &\leq \bigo\left(
            \frac{\sum_{i = 1}^L \|p_i\|_\infty D_i}{L^2 t} + \frac{\sum_{i = 1}^L \|p_i\|_2^2 + \|q_i\|_2^2}{L^2t^2}
        \right)
        \\&\leq \bigo\left(
            \frac{\sqrt{\sum_{i = 1}^L \|p_i\|_\infty^2}}{L} \frac{\sqrt{\sum_{i \in [L]}D_i^2}}{L} \frac1t + \frac{1}{Lt^2}\frac{\sum_{i = 1}^L \|p_i\|_2^2 + \|q_i\|_2^2}{L}
        \right)
        \\&\leq \bigo\left(
            \sqrt{\frac{\sum_{i = 1}^L \|p_i\|_2^2}{L}} \frac{\sum_{i = 1}^L D_i}{L^{3/2}t} + \frac{1}{Lt^2}\frac{\sum_{i = 1}^L \|p_i\|_2^2 + \|q_i\|_2^2}{L}
        \right)
        \\&\leq \bigo\left(
            \sqrt{\frac{\sum_{i = 1}^L \|p_i\|_2^2}{L}} \frac{D}{L^{1/2}t} + \frac{1}{Lt^2}\frac{\sum_{i = 1}^L \|p_i\|_2^2 + \|q_i\|_2^2}{L}
        \right)
        \\&\leq \bigo\left(
            \frac{D}{\sqrt{d_c} \sqrt{L} t} + \frac{1}{d_cLt^2}
        \right),
    \end{align}
    where we rearranged the terms and used Cauchy-Schwarz in the second line, $\|\cdot\|_\infty \leq \|\cdot\|_2$ and the elementary identity $\sum_i x_i^2 \leq (\sum_i x_i)^2$ for non-negative $x_i$  next, the definition of $D$ in the penultimate step, and \Cref{thm:derandomization-classical-quantum} in the final step.

    Now, in the completeness case, $D = 0$, and by Chebyshev's inequality, the probability of error is at most
    \begin{equation}
        \Pr[\widehat{D} > \tau] \leq \mathbb{V}[\widehat{D}]/\tau^2 \leq \bigo\left( 
            \frac{1}{d_c L t^2} \cdot \frac{d^4}{d_q^2 \epsilon^4}
        \right) \leq \bigo(1)
    \end{equation}
    for our choice of $L$ and $t$; this can be made smaller than, say, $\frac13$ by increasing $L$ and $t$ by sufficiently large constant factors.

    For soundness, we have $D \geq 2\tau$. Thus, the probability of error here is at most
    \begin{align}
         \Pr[\widehat{D} > \tau] &\leq \Pr[\widehat{D} \geq D/2] \leq \frac{4}{D^2} \cdot \mathbb{V}[\widehat{D}]
         \\&\leq \bigo\left(
            \frac{1}{\sqrt{d_c}\sqrt{L} t \cdot D} + \frac{1}{d_c L t^2 \cdot D^2}
         \right)
         \\&\leq \bigo\left(
            \frac{1}{\sqrt{d_c}\sqrt{L}t \cdot \tau} + \frac{1}{d_c L t^2 \tau^2}
         \right);
    \end{align}
    we already showed that the latter term is $\bigo(1)$, and the former term is its square root, thus also $\bigo(1)$. Again, we can ensure that this is at most $\frac13$ by increasing the constant factors in $L$ and $t$.

    The number of distributed nodes is $m = Lt = \bigo\left(
        \frac{d^3 \sqrt{\log d}}{d_q^2 d_c \epsilon^2},
    \right)$
    as claimed.
\end{proof}

It now remains to prove \Cref{thm:derandomization-classical-quantum}. This proof will combine arguments from \Cref{ss:public_ub} and \Cref{ss:quantum_private_ub}. We will use the probabilistic method and show that for $L$ unitaries drawn independently with the Haar measure, \Cref{eq:derandomized-guarantee-compression,eq:derandomized-guarantee-spread} hold simultaneously with non-zero probability; consequently, there exists a deterministic choice of $L$ unitaries satisfying these inequalities. In particular, we will show that each inequality holds except with probability strictly less than $\frac14$; by a union bound, the overall failure probability will be strictly less than $1$.

Let us first consider \Cref{eq:derandomized-guarantee-spread}.

\begin{lemma}
\label{lem:classical-quantum-spread}
    For $L = \bigo(\frac{d^2}{d_q^2 d_c} \log d)$, let $\bfU_1, \dots, \bfU_L \sim U(d)$ be independent Haar random unitaries. Then, for all quantum states $\rho$ and $i \in [L]$, 
    with $p_i$ the associated distribution over classical labels under the quantum instrument $\mc{I}_{\bfU_i}$, we have
    \begin{equation}
        \frac1L \sum_{i = 1}^L \|p_i\|_2^2 \leq \frac{16}{d_c},
    \end{equation}
    except with probability at most $\frac14$.
\end{lemma}

\begin{proof}
    Let us write $\mc{I}_{U}(\rho) = \sum_{c = 1}^{d_c} \rho_c \otimes \ketbra{c}{c}$. Then, the associated distribution is $p = (\Tr(\rho_1), \dots, \Tr(\rho_{d_c}))$, with
    \begin{equation}
        \|p\|_2^2 = \sum_{c = 1}^{d_c} \Tr(\rho_c)^2 = \sum_c \Tr(\swap \cdot \rho_c^{\otimes 2}).
    \end{equation}
    Recall that $\rho_c = \tr_{A,C} (U\rho U^\dag \cdot I_A \otimes I_Q \otimes \ketbra{c}{c})$.
    Thus,
    \begin{align}
        \|p\|_2^2 &= \sum_{c = 1}^{d_c} \Tr ( (U\rho U^\dag)^{\otimes 2} \cdot I_A^{\otimes 2} \otimes \swap_Q \otimes \ketbra{c}{c}^{\otimes 2})
        \\&= \sum_{c = 1}^{d_c} \Tr(\rho^{\otimes 2} M_c(U)),
    \end{align}
    where we define
    \begin{equation}
        M_c(U) \triangleq U^{\dag \otimes 2} (I_A^{\otimes 2} \otimes \swap_Q \otimes \ketbra{c}{c}^{\otimes 2}) U^{\otimes 2}.
    \end{equation}
    Further, we define
    \begin{equation}
        M(U) \triangleq \sum_{c = 1}^{d_c} M_c(U) = U^{\dag \otimes 2}(I_A^{\otimes 2} \otimes \swap_Q \otimes \sum_c \ketbra{c}{c}^{\otimes 2}) U^{\otimes 2}.
    \end{equation}

    Now, let us compute $\mathbb{E}_{\bfU} M(\bfU)$. For ease of notation, we write $\Pi_C \triangleq \sum_C \ketbra{c}{c}^{\otimes 2}$. Then, by \Cref{lem:haar-second-moment},
    \begin{align}
        \mathbb{E}_{\bfU} M(\bfU) &= \mathbb{E}_{\bfU} [U^{\otimes 2} I_A^{\otimes 2} \otimes \swap_Q \otimes \Pi_C U^{\dag \otimes 2}]
        \\&= \frac{I}{d^2-1} [\Tr(I_A^{\otimes 2} \otimes \swap_Q \otimes \Pi_C) - \frac1d \Tr(\swap_A \otimes I_Q^{\otimes 2} \otimes \swap \Pi_C)] \nonumber
        \\&+ \frac{\swap}{d^2-1} [\Tr(\swap_A \otimes I_Q^{\otimes 2} \otimes \swap \Pi_C) -\frac1d \Tr(I_A^{\otimes 2} \otimes \swap_Q \otimes \Pi_C)]
        \\&= \frac{I}{d^2-1} \cdot (d \cdot d_A - d_q) + \frac{\swap}{d^2-1} \cdot (d \cdot d_q - d_A), \label{eq:classical-quantum-spread-1}
    \end{align}
    where the last step follows via elementary calculations which we omit. Let us call the above matrix $M \triangleq \mathbb{E}_{\bfU} M(\bfU)$. Now, we have that
    \begin{align}
        \mathbb{E}_{\bfU} \|p\|_2^2 &= \Tr(\rho^{\otimes 2} M)
        \\&= \frac{1}{d^2-1} [d\cdot d_A - d_q] + \frac{\Tr(\rho^2)}{d^2-1} [d \cdot d_q - d_A]
        \\&\leq \frac{d_A}{d+1} + \frac{d_q}{d+1} \leq \frac{2}{d_c},
    \end{align}
    where the last line used $\Tr(\rho^2) \leq 1$ and then $d \geq d_A d_c$ and $d \geq d_q d_c$. Consequently, to prove the lemma, it suffices to show that the average of $M(\bfU_1), \dots, M(\bfU_L)$ is concentrated around their mean, for sufficiently large $L$.

    We will apply the matrix Chernoff bound, i.e.,~\Cref{eq:matrix-chernoff-operator-upper} of~\Cref{thm:matrix_chernoff_inequality}, to the random matrices $\{M(\bfU_i)\}_i$. By unitary invariance of the operator norm,
    \begin{equation}
        \|M(U)\|_\infty = \|I_A^{\otimes 2} \otimes \swap_Q \otimes \Pi_C\|_\infty = 1,
    \end{equation}
    i.e., we will take $R = 1$. Next, note that
    \begin{equation}
        \mu_{\max} = \lambda_{\max}\left(\sum_{i = 1}^L \mathbb{E} M(\bfU_i)\right) = L \lambda_{\max}(M).
    \end{equation}
    By \Cref{eq:classical-quantum-spread-1}, 
    \begin{equation}
        \lambda_{\max}(M) = \frac{d_A + d_q}{d+1};
    \end{equation}
    this is at most $\frac{2}{d_c}$ and at least $\frac{d_q}{2d}$. Consequently, 
    \begin{equation}
        \frac{2L}{d_c} \geq \mu_{\max} \geq \frac{L d_q}{2 d}
    \end{equation}
    Let $\overline{M} = \sum_{i = 1}^L M(\bfU_i)$. We will now apply~\Cref{eq:matrix-chernoff-operator-upper} with $t = e^2 < 8$: 
    \begin{align}
        \mathbf{Pr}[\lambda_{\max}(\overline{M}) \geq 16L/d_c] &\leq \mathbf{Pr}[\lambda_{\max}(\overline{M}) \geq e^2 \cdot \mu_{\max}]
        \\&\leq d^2 \cdot \exp\left(\frac{-e^2\mu_{\max}}{R}\right)
        \\&\leq d^2 \cdot \exp\left(\frac{-e^2 L d_q}{2 d}\right);
    \end{align}
    the above is at most $\frac14$ whenever $L \geq \widetilde{L} \triangleq c_o \cdot d d_q^{-1} \log(d)$ for some absolute constant $c_0 > 0$. We take  $L = c_0 \frac{d^2}{d_q^2 d_c} \log(d) = d_A \Tilde{L}$, which is thus sufficiently large. 
\end{proof}

Next, we will show that \Cref{eq:derandomized-guarantee-compression} also holds with sufficiently high probability. The proof will follow the exact same steps as those in the proof of \Cref{thm:good_unitaries_private_only_ub}, with only some parameter adjustments to account for the classical-quantum case; we will brush over the similar steps and primarily highlight these adjustments.

As before, let $H = \{X = X^\dag, (X) = 0, X \in \mathbb{C}^{d \times d}\}$ be the space of $d$-dimensional traceless Hermitian operators, which has dimension $d^2-1$. For each $c \in [d_c]$, we define $W_c = \{I_A \otimes Y_Q \otimes \ketbra{c}{c}: Y \in \mathbb{C}^{d_q \otimes d_q}, Y = Y^\dag, \Tr(Y) = 0\}$ to be the space of Hermitian traceless operators that are identity on the junk register and the classical bit $c$ on $C$. Each space $W_c$ has dimension $d_q^2-1$. We will define $\Pi_{c}$ to be the projector onto $W_c$. Let $W_{U,c} = \{U^\dag M U : M \in W_c\}$, and let $\Pi_{U,c}$ be the projector onto this space. 

We now have the following bound on the norm under the instrument corresponding to the identity unitary, which follows the same proof as \Cref{lem:W_projector_private_ub_no_classical}.

\begin{lemma}
\label{lem:classical-quantum-compression-1}
    For any $X \in H$,
    \begin{equation}
        \|\Tr_{A,C} [X \cdot I_A \otimes I_Q \otimes \ketbra{c}{c}]\|_2^2 \geq \frac{d}{d_q d_c} \|\Pi_c\|_2^2. 
    \end{equation}
\end{lemma}

\begin{proof}
    For ease of notation, define $M_c \triangleq \Tr_{A,C}[X \cdot I_A \otimes I_Q \otimes \ketbra{c}{c}]$. Then, note that

    \begin{equation}
        \Pi_c(X) = \left(M_c - \frac{\Tr(M_c)}{d_q} I_Q\right) \otimes \frac{I_A}{d_A} \otimes \ketbra{c}{c},
    \end{equation}
    ensuring that the resulting operator is traceless. The first operator above satisfies
    \begin{equation}
        \left\|M_c - \frac{\Tr(M_c)}{d_q} I_Q\right\|_2^2 = \|M_c\|_2^2 - \frac{\Tr(M_c)^2}{d_q} \leq \|M_c\|_2^2.
    \end{equation}
    Consequently, 
    \begin{equation}
        \|\Pi_c(X)\|_2^2 \leq \frac{1}{d_A} \|M_c\|_2^2.
    \end{equation}
    The statement then follows by rearrangement and recalling $d_A = d / d_qd_c$.
\end{proof}

Let us now extend this to generic instruments.

\begin{lemma}
\label{lem:classical-quantum-compression-2}
    For any $X \in H$ and $U \in U(d)$,
    \begin{equation}
        \|\mathcal{I}_{U,c}(X)\|_2^2 \geq \frac{d}{d_q d_c} \|\Pi_{U,c}\|_2^2. 
    \end{equation}
\end{lemma}

\begin{proof}
   Following the same arguments as in the proof of \Cref{lem:channel_norm_private_ub_no_classical}, we have
    \begin{equation}
        \Pi_{U,c}(X) = U^\dag \Pi_c (UXU^\dag) U.
    \end{equation}
    By the unitary invariance of the $2$-norm, we obtain
    \begin{equation}
        \|\Pi_{U,c}(X)\|_2^2 = \|\Pi_c(UXU^\dag)\|_2^2 \leq \frac{d_qd_c}{d}\|\Tr_{A,C} [UXU^\dag \cdot I_A \otimes I_Q \otimes \ketbra{c}{c}]\|_2^2 = \frac{d_qd_c}{d} \|\mathcal{I}_{U,c}(X)\|_2^2,
    \end{equation}
    where the inequality is due to \Cref{lem:classical-quantum-compression-1}.
\end{proof}

\begin{lemma}
    \label{lem:classical-quantum-compression-3}
    \begin{equation}
        \mathbb{E}_{\bfU} [\Pi_{\bfU,c}] = \frac{d_q^2-1}{d^2-1} I_H.
    \end{equation}
\end{lemma}

\begin{proof}
    Invoking Schur's lemma and following the same steps as we did in \Cref{lem:projector_expectation_private_ub_no_classical}, the expected projector is proportional to the identity. Equating traces, we get
    \begin{equation}
        \mathbb{E}_{\bfU} [\Pi_{\bfU,c}] = \frac{\mathrm{dim}(W_{U,c})}{\mathrm{dim}(H)} I_H = \frac{d_q^2-1}{d^2-1} I_H,
    \end{equation}
    as claimed.
\end{proof}

We will now apply the matrix Chernoff bound in \Cref{thm:matrix_chernoff_inequality} to show a probabilistic bound on the empirical averaged projector.

\begin{lemma}
    \label{lem:classical-quantum-compression-4}
    For $\bfU_1, \dots, \bfU_L \sim U(d)$, there exists a constant $C > 0$, such that
    \begin{equation}
        \Pr \left[
            \sum_{r = 1}^L \sum_{c = 1}^{d_c} \Pi_{U_r,c} \preceq \frac{L d_c}{2} \frac{d_q^2-1}{d^2-1} I
        \right] \leq \frac14.
    \end{equation}
    whenever $L \geq \frac{C \cdot d^2}{d_q^2 d_c} \log d$.
\end{lemma}

\begin{proof}
    By \Cref{lem:classical-quantum-compression-4} and linearity of expectation,
    \begin{equation}
        \mathbb{E}\left[ \sum_{r = 1}^L \sum_{c = 1}^{d_c} \Pi_{U_r,c}\right] = Ld_c \frac{d_q^2-1}{d^2-1} I_H.
    \end{equation}
   Here we have $\mu_{\mathrm{min}} = Ld_c \frac{d_q^2-1}{d^2-1}$ and $R = 1$. Applying \Cref{eq:matrix-chernoff-operator-lower} of \Cref{thm:matrix_chernoff_inequality} with $s = \frac12$, we have
    \begin{equation}
        \Pr \left[
            \lambda_{\mathrm{min}}\left(\sum_{r = 1}^L \sum_{c = 1}^{d_c} \Pi_{U_r,c}\right) \leq \frac{L d_c}{2} \frac{d_q^2-1}{d^2-1}
        \right] \leq (d^2-1) \exp(\frac{-1}{8} Ld_c \frac{d_q^2-1}{d^2-1}),
    \end{equation}
    which is at most $\frac14$ for our choice of $L$.
\end{proof}

We will finally show that the above \Cref{lem:classical-quantum-compression-4} implies that \Cref{eq:derandomized-guarantee-compression} holds except with probability at most $\frac14$.

\begin{corollary}
\label{cor:classical-quantum-compression}
    For $\bfU_1, \dots, \bfU_L \sim U(d)$ and any $X \in H$, there exists a constant $C > 0$, such that
    \begin{equation}
        \Pr \left[\frac1L \sum_{i = 1}^L \|\mathcal{I}_{U_i}(X)\|_2^2 \leq \frac14 \frac{d_q}{d} \|X\|_2^2 \right] \leq \frac14,
    \end{equation}
    whenever $L \geq \frac{C \cdot d^2}{d_q^2 d_c} \log d$.
\end{corollary}

\begin{proof}
    Condition on the high-probability event in \Cref{lem:classical-quantum-compression-4}; i.e., we have
    \begin{equation}
        \sum_{r = 1}^L \sum_{c = 1}^{d_c} \Pi_{U_r,c} \succeq \frac{L d_c}{2} \frac{d_q^2-1}{d^2-1} I.
    \end{equation}
    Then, we have
    \begin{equation}
        \left\langle 
            X, \sum_{r = 1}^L \sum_{c = 1}^{d_c} \Pi_{\bfU_r,c} X
        \right\rangle \geq \frac{L d_c}{2} \frac{d_q^2-1}{d^2-1} \|X\|_2^2,
    \end{equation}
    i.e.,
    \begin{equation}
        \frac{1}{L} \sum_{r = 1}^L \sum_{c = 1}^{d_c} \|\Pi_{\bfU_r,c} X\|_2^2 \geq \frac{d_c}{2} \frac{d_q^2-1}{d^2-1} \|X\|_2^2.
    \end{equation}
    Now, by \Cref{lem:classical-quantum-compression-2}, 
    \begin{equation}
        \frac{1}{L} \sum_{r = 1}^L \sum_{c = 1}^{d_c} \|\mathcal{I}_{\bfU_r,c}(X)\|_2^2 \geq \frac{d}{d_qd_c} \frac{d_c}{2} \frac{d_q^2-1}{d^2-1} \|X\|_2^2 \geq \frac{1}{4} \frac{d_q}{d} \|X\|_2^2.
    \end{equation}
    Finally, note that for each $r \in [L]$,
    \begin{equation}
        \sum_{c = 1}^{d_c} \|\mathcal{I}_{U_r,c}(X)\|_2^2 = \|\mathcal{I}_{U_r}(X)\|_2^2.
    \end{equation}
    Thus, the event in the statement only holds under the low-probability event in \Cref{lem:classical-quantum-compression-4}, which has probability at most $\frac14$.
\end{proof}

Let us now conclude the proof of \Cref{thm:derandomization-classical-quantum}.

\begin{proof}[Proof of \Cref{thm:derandomization-classical-quantum}]
    By \Cref{lem:classical-quantum-spread} and \Cref{cor:classical-quantum-compression} and a union bound, when sampling $L = \bigo\left(
        \frac{d^2}{d_q^2 d_c} \log d
    \right)$ Haar-random unitaries $\bfU_1, \dots, \bfU_L \sim U(d)$, \Cref{eq:derandomized-guarantee-compression,eq:derandomized-guarantee-spread} hold except with probability at most $\frac12$, i.e., they hold with positive probability. Consequently, there exists a deterministic choice of unitaries $U_1, \dots, U_L$ satisfying these inequalities, as claimed. 
\end{proof}

\section{Lower Bounds}\label{s:lower_bounds}

We start by outlining the main technical framework used in our lower bound proofs. As is standard, to prove lower bounds for state certification in our setting, we will consider its hardest instance, i.e., mixedness testing. Specifically, we will prove lower bounds for distinguishing between the maximally mixed state $\mmstate$ and an ensemble of states that are $\epsilon$-far from $\mmstate$. Note that any algorithm for mixedness testing must also succeed at this latter task with high probability. To formalize this, we start by defining the notion of ``almost-$\epsilon$ perturbation ensembles'', borrowing terminology from \cite{acharya2020inferenceinformationconstraints}.

\begin{definition}[Almost-$\epsilon$ perturbation]
    \label{def:almost-eps-perturbations}
    An ensemble $D$ of quantum states is an almost-$\epsilon$ perturbation of a state $\sigma$ if $\mathrm{Pr}_{\rho \sim D}[\|\rho - \sigma\|_1 \geq \epsilon] \geq \frac12$. Let the set of all almost-$\epsilon$ perturbations of $\sigma$ be $\mc{D}_{\epsilon}(\sigma)$.
\end{definition}

To prove lower bounds for distinguishing between a fixed state $\sigma$ and a state drawn from some almost-$\epsilon$ perturbed ensemble $D$, one typically shows that the states $\sigma^{\otimes m}$ and $\mathbb{E}_{\rho \sim D}[\rho^{\otimes m}]$ are statistically indistinguishable unless the number of copies $m$ is large enough. However, in our distributed setting, we can show stronger lower bounds than permitted by such arguments. In particular, consider the setting where each distributed node receives a copy of some unknown $d$-dimensional state $\rho$, and can send a $d_q$-dimensional qudit and $n_c$ classical bits (i.e., a $d_c \triangleq 2^{n_c}$-dimensional classical message) to the central node. Then, the most general action of the distributed node $N_i$ can be modeled as a quantum instrument $\mathcal{I}\colon\mathbb{C}^{d\times d}\to \mathbb{C}^{d_q\times d_q}\otimes \mathbb{C}^{d_c\times d_c}$. The central node then receives messages $\mathcal{I}_1(\rho), \dots, \mathcal{I}_m(\rho)$, and must determine whether $\rho = \sigma$ or $\rho$ was drawn from $D$. Thus, we can now argue that $m$ must be large enough for $\bigotimes_{i = 1}^m \mathcal{I}_i(\sigma)$ and $\mathbb{E}_{\rho \sim D} \left[\bigotimes_{i = 1}^m \mathcal{I}_i(\rho)\right]$ to be statistically distinguishable. Moreover, when we consider the private-coin setting, we can show even stronger lower bounds. In this weaker setting, one can imagine that the ensemble $D$ is chosen adversarially depending on the choice of each $\mc{I}_i$. We formalize these arguments in the following lemma: 

\begin{lemma}[Lower bounds via minimax or maximin $\chi^2$ divergences]\label{lem:instrument_minimax_maximin_bounds}
    Let $d\geq d_cd_q\geq 2$. Then, for any protocol in the $(n_c,n_q, \mathsf{public})$ setting -- i.e.  for any public-coin protocol using only $n_c$ length classical messages and $d_q$-dimensional quantum messages -- to succeed at $\epsilon$-certifying $d$-dimensional states the number of distributed nodes $m$ must be large enough such that 
    \begin{equation}
        \min_{D\in\mathcal{D}_\epsilon\left(\mmstate\right)}\max_{\mathcal{I}_1,...,\mathcal{I}_m}\QChi{\mathbb{E}_{\rho\sim D}\left[\bigotimes_{i=1}^m\mathcal{I}_i(\rho)\right]}{\bigotimes_{i=1}^m\mathcal{I}_i\left(\mmstate\right)}\geq\frac{1}{16},
    \end{equation}
    where each $\mathcal{I}_i$ is a quantum instrument outputting a length $n_c$ classical string and a $d_q$-dimensional quantum state. On the other hand, for any protocol in the $(n_c,n_q, \mathsf{private})$ setting -- i.e. for any private-coin protocol with the same communication constraints -- to succeed we must have $m$ large enough such that
    \begin{equation}
        \max_{\mathcal{I}_1,...,\mathcal{I}_m}\min_{D\in\mathcal{D}_\epsilon\left(\mmstate\right)}\QChi{\mathbb{E}_{\rho\sim D}\left[\bigotimes_{i=1}^m\mathcal{I}_i(\rho)\right]}{\bigotimes_{i=1}^m\mathcal{I}_i\left(\mmstate\right)}\geq\frac{1}{16}.
    \end{equation}
\end{lemma}
The key point is that \Cref{lem:instrument_minimax_maximin_bounds} allows us to turn the task of lower bounding the complexity of distributed mixedness testing into an equivalent one of upper bounding the quantum $\chi^2$-divergence. We note that this is an adaptation and generalization of the methods used in Acharya \etal~\cite{acharya2020inferenceinformationconstraints} and  Liu \etal~\cite{liu2024role}, to allow for quantum instruments. Nevertheless, for the sake of completeness, we include the explicit proof in \Cref{app:lemma_instrument_minimax_maximin_bounds}. 

\Cref{lem:instrument_minimax_maximin_bounds} essentially allows us to view the point-versus-mixture distinguishing task as a two-player game, where one player selects the channels $\Phi_1, \dots, \Phi_m$, and an adversary picks the mixture of alternatives $D \in \mc{D}_\epsilon(I/d)$. The former attempts to maximize the quantum $\chi^2$-divergence while the latter attempts to minimize this quantity. The key difference between the public and private coin settings is the order in which the players make their moves. The adversary's ability to go second in the private-coin setting allows them to pick a hard mixture \emph{depending} on the first player's chosen instruments, allowing us to show stronger lower bounds. 

From the above, it is clear that to lower bound the number of nodes $m$ necessary, it suffices to \textit{upper bound} the quantum $\chi^2$ divergences appearing in \Cref{lem:instrument_minimax_maximin_bounds}. To prove such upper bounds, we will exploit a version of the recently introduced quantum Ingster-Suslina method \cite[Lemma 4.3]{odonnell2025instanceoptimalquantumstatecertification}, generalized to allow for upper bounding arbitrary $\alpha$-$\chi^2$ divergences as opposed to simply the $0$-$\chi^2$ divergence:
\begin{lemma}[Quantum Ingster-Suslina Generalized]\label{lem:quantum_ingster-suslina_gen}
    Let $\bftheta$ be a random variable that parametrizes states $\rho_{i,\bftheta}\in \mathbb{C}^{d\times d}$ for all $i\in[m]$. Let $\sigma_i\in\mathbb{C}^{d\times d}$ be fixed quantum states. Let
    \begin{equation}
        \rho_{\bftheta}^{(m)}=\bigotimes_{i=1}^m\rho_{i,\bftheta},~\sigma^{(m)}=\bigotimes_{i=1}^m\sigma_i.
    \end{equation}
    Then,
    \begin{equation}
        1+\QChiAlpha{\alpha}{\mathbb{E}_{\bftheta}[\rho_{\bftheta}^{(m)}]}{\sigma^{(m)}}=\mathbb{E}_{\bftheta,\bftheta'}\left[\prod_{i=1}^m(1+Z_i(\bftheta,\bftheta'))\right]\leq\mathbb{E}_{\bftheta,\bftheta'}\left[\exp\left(\sum_{i=1}^m Z_i(\bftheta,\bftheta')\right)\right],
    \end{equation}
    where $\bftheta,\bftheta'$ are i.i.d., and
    \begin{equation}
        Z_i(\bftheta,\bftheta') \triangleq \Tr[\sigma_i^{-\alpha}(\rho_{i,\bftheta}-\sigma_i)\sigma_i^{\alpha-1}(\rho_{i,\bftheta'}-\sigma_i)].
    \end{equation}
\end{lemma}
The function $Z_i$ appearing above is different to that appearing in the less general quantum Ingster-Suslina method given in \cite[Lemma 4.3]{odonnell2025instanceoptimalquantumstatecertification}, however the proof follows the same arguments as \cite{odonnell2025instanceoptimalquantumstatecertification}. For completeness, we provide a proof of Lemma~\ref{lem:quantum_ingster-suslina_gen} in \Cref{app:symmetrized_quantum_ingster_suslina}.
By setting $\alpha=0.5$, we get the following as a corollary:
\begin{corollary}[Symmetrized Quantum Ingster Suslina]\label{lem:quantum_ingster-suslina_symm}
    Let $\theta$ be a random variable that parametrizes states $\rho_{i,\theta}\in \mathbb{C}^{d\times d}$ for all $i\in[m]$. Let $\sigma_i\in\mathbb{C}^{d\times d}$ be quantum states. Let
    \begin{equation}
        \rho_\theta^{(m)}=\bigotimes_{i=1}^m\rho_{i,\theta},~\sigma^{(m)}=\bigotimes_{i=1}^m\sigma_i.
    \end{equation}
    Then,
    \begin{equation}
        1+\QChi{\mathbb{E}_\theta[\rho_\theta^{(m)}]}{\sigma^{(m)}}=\mathbb{E}_{\theta,\theta'}\left[\prod_{i=1}^m(1+Z_i(\theta,\theta'))\right]\leq\mathbb{E}_{\theta,\theta'}\left[\exp\left(\sum_{i=1}^m Z_i(\theta,\theta')\right)\right],
    \end{equation}
    where $\theta,\theta'$ are i.i.d., and
    \begin{equation}
        Z_i(\theta,\theta') \triangleq \Tr[\sigma_i^{-1/2}(\rho_{i,\theta}-\sigma_i)\sigma_i^{-1/2}(\rho_{i,\theta'}-\sigma_i)].
    \end{equation}
\end{corollary} 
With this established, we can now state the specific almost-$\epsilon$ perturbation family we will consider to obtain our lower bounds. This instance was first considered in \cite{liu2024role} and has since been used to prove lower bounds in various other resource-constrained settings (see e.g. \cite{liu2024quantum,aliakbarpour2025adversarially}).

\begin{definition}[\cite{liu2024quantum} almost-$\epsilon$ perturbation]
\label{def:hard-instance-mic}
    Let $\{V_i\}_{i \in [d^2]}$ form an orthonormal basis for $\mathbb{C}^{d \times d}$ w.r.t the Hilbert--Schmidt inner product, with $V_{d^2} = I/\sqrt{d}$. For some integer $\frac{d^2}{2} \leq \ell \leq d^2-1$, and $z \in \{-1,+1\}^{\ell}$, define the parameterized perturbation
    \begin{equation}
        \Delta_z \triangleq \frac{c\epsilon}{\sqrt{d}} \cdot \frac{1}{\sqrt{\ell}} \sum_{i = 1}^\ell z_i V_i, \quad N_z\triangleq\min\left\{1,\frac{1}{d\|\Delta_z\|_\infty}\right\}, \quad \text{and } \quad \bar{\Delta}_z \triangleq N_z \cdot \Delta_z.
    \end{equation}
    Let $\rho_z \triangleq \frac{I}{d} + \bar{\Delta}_z$. It will be convenient to define the matrix $\mc{V} = [\mathrm{vec}(V_1), \dots, \mathrm{vec}(V_\ell)] \in \mathbb{C}^{d^2 \times \ell}$, containing only the vectorizations of the first $\ell$ operators.
\end{definition}

It was shown in \cite[Corollary 4.4]{liu2024role} that for any choice of $\{V_i\}$ and a uniformly random $\bfz \sim \{-1,+1\}^\ell$, the parameterized state $\rho_{\bfz}$ is $\epsilon$-far from $\mmstate$ with high probability, i.e., the above ensemble is a valid almost $\epsilon$-perturbation, whenever $\epsilon$ is smaller than an absolute constant $C$.

In summary, to prove lower bounds for distributed mixedness testing (and therefore distributed state certification) we take \Cref{lem:instrument_minimax_maximin_bounds} as a starting point. We then:
\begin{enumerate}
\item Consider the specific almost-$\epsilon$ perturbation family given in Definition~\ref{def:hard-instance-mic}.
\item Use the generalized quantum Ingster-Suslina method (\Cref{lem:quantum_ingster-suslina_symm}) to upper bound the divergences appearing in \Cref{lem:instrument_minimax_maximin_bounds} (with the almost-$\epsilon$ perturbation family from Definition~\ref{def:hard-instance-mic}); from this, we determine how large $m$ needs to be.
\end{enumerate}

Before carrying out the steps above, we require one last lemma about quantum instruments.
\begin{lemma}\label{lem:instrument_support}
    Let $\mathcal{I}\colon\mathbb{C}^{d\times d}\to\mathbb{C}^{d_q\times d_q}\otimes \mathbb{C}^{d_c\times d_c}$ be a quantum instrument, and $\beta=\mathcal{I}\left(\mmstate\right)$. Let $P=\Pi_{\supp(\beta)}$. Then
    \begin{equation}
        \mathcal{I}(X)=P\mathcal{I}(X)P.
    \end{equation}
\end{lemma}
\begin{proof}
     Let $F=\mathcal{I}(I)=d\beta$, and $P=\Pi_{\supp(F)}$. We will show that $\mathcal{I}(X)$ has no support on $\ker F$, meaning it is completely defined on the support of $F$. To do so, we write a Kraus operator decomposition of the instrument $\mathcal{I}$ given by $\{A_k\}$ (see \Cref{def:kraus_operator_decomposition}), which exists since a quantum instrument is a completely positive operator. Then
     \begin{equation}
         \mathcal{I}(X)=\sum_{k=1}^N A_k X A_k^\dagger,
     \end{equation}
     where $N\leq dd_cd_q$. Then, we have that
     \begin{equation}
         F=\mathcal{I}(I)=\sum_{k=1}^N A_kA_k^\dagger.
     \end{equation}
     For some $|v\rangle\in \ker(F)$, we have that
     \[
         0=\langle v|F|v\rangle =\sum_{k=1}^N\|A^\dagger_k v\|^2,
    \]
    and so
    $
         A_k^\dagger|v\rangle = 0,
     $
      which implies the lemma.
\end{proof}
This result implies that while general quantum instruments can map the maximally mixed state to some \textit{a priori} arbitrary state $\beta$, this initial mapping constrains what the quantum instrument can map every other input quantum state to, since it will never map outside the support of $\beta$.
\subsection{Main Lower Bounds}
In this section, we prove our lower bounds for the $[n_c,n_q,R]$ settings. We start with $R = \mathsf{public}$ -- i.e. for any public coin protocol communicating with both limited classical bits and quantum qubits. Namely, we show the following:
\begin{theorem}[\Cref{thm:intro-public} lower bound, restated]\label{thm:public_coin_lb}
    Let $d\geq 6$, $0<\epsilon< \eps_0$, for a sufficiently small absolute constant $\eps_0 >0$. Let $d_c=2^{n_c}$ and $d_q=2^{n_q}$. In the $(n_c, n_q, \mathsf{public})$ setting, the complexity of $\epsilon$-certification is at least $\Omega\left(\frac{d^2}{d_q\sqrt{d_c}\epsilon^2}\right)$. That is,
    \begin{equation}
        m=\Omega\left(\frac{d^2}{2^{n_q+n_c/2}\epsilon^2}\right)
    \end{equation}
distributed nodes are necessary for $\epsilon$-certification in the $(n_c, n_q, \mathsf{public})$ setting.
\end{theorem}

In the private-coin, i.e., $[n_c,n_q,\mathsf{private}]$-setting, we obtain the following stronger lower bound.

\begin{theorem}[\Cref{thm:intro-private} lower bound, restated]\label{thm:private_coin_lb}
With the same parameter ranges as~\Cref{thm:public_coin_lb}, except in the $(n_c, n_q, \mathsf{private})$ setting, the complexity of $\epsilon$-certification is at least $\Omega\left(\frac{d^2}{d_q^2 d_c\epsilon^2}\right)$. That is,
    \begin{equation}
        m\geq\Omega\left(\frac{d^3}{2^{2n_q + n_c}\epsilon^2}\right),
    \end{equation}
    distributed nodes are necessary for $\epsilon$-certification in the $(n_c, n_q, \mathsf{private})$ setting.
\end{theorem}

To prove these theorems, we first exhibit three key technical lemmas. We will prove \Cref{thm:public_coin_lb,thm:private_coin_lb} assuming these lemmas to be true, and defer their proofs to the subsequent section. First, we prove the following lemma relating the quantum $\chi^2$-divergence to norms of a certain map defined by the quantum instruments.

\begin{lemma}[$\chi^2$ divergence via channel norms]\label{lem:instrument_T_bound}
    Let $d\geq 6$, $0<\epsilon<C$ for some constant $C$, and $d_c,d_q\geq 1$. Additionally, let $\rho_z$ and $\mathcal{V}$ be as per \Cref{def:hard-instance-mic}, with some $\frac{d^2}{2} \leq \ell \leq d^2-1$. For all $i\in [m]$, let $\mathcal{I}_i:\mathbb{C}^{d\times d}\to \mathbb{C}^{d_q\times d_q}\otimes \mathbb{C}^{d_c\times d_c}$ be the quantum instrument associated with distributed node $i$. 
    For each quantum instrument, we define $\beta_i=\mathcal{I}_i\left(\mmstate\right)$ as the output of $\mathcal{I}_i$ with the maximally mixed state as input. We define the augmented map $\Lambda_i(X)=\beta_i^{-1/4}\mathcal{I}_i(X)\beta_i^{-1/4}$, and $T_\Lambda=\frac{1}{m}\sum_{i=1}^mM_{\Lambda_i}^\dagger M_{\Lambda_i}$, where $M_{\Lambda_i}$ is the Liouville matrix representation of $\Lambda_i$. Then,
    \begin{equation}
        \QChi{\underset{\bfz\sim\{-1,+1\}^\ell}{\mathbb{E}}\left[\bigotimes_{i=1}^m\mathcal{I}_i(\rho_{\bfz})\right]}{\bigotimes_{i=1}^m\mathcal{I}_i\left(\mmstate\right)} < \frac{1}{16},
    \end{equation}
    unless
    \begin{equation}
        m\geq \Omega\left(\frac{d\ell}{\epsilon^2}\cdot\frac{1}{\|\mathcal{V}^\dagger T_\Lambda\mathcal{V}\|_2}\right).
    \end{equation}
\end{lemma}

Additionally, we will require the following bounds on the norms of the Liouville matrices of such augmented maps.
\begin{lemma}\label{lem:instrument_infty_bound}
    For the parameters defined above and a quantum instrument $\mathcal{I}$, the augmented map $\Lambda$ satisfies:
    \begin{equation}
        \|M_{\Lambda}\|_\infty\leq \sqrt{d}.
    \end{equation}
\end{lemma}
\begin{lemma}\label{lem:instrument_2_bound}
    For the parameters defined above and a quantum instrument $\mathcal{I}$, the augmented map $\Lambda$ satisfies:
    \begin{equation}
        \|M_{\Lambda}\|_2\leq d_q\sqrt{dd_c}.
    \end{equation}
\end{lemma}

We defer the proof of these three lemmas to \Cref{s:three_lemmas}, and first prove \Cref{thm:public_coin_lb}.

\begin{proof}[Proof of \Cref{thm:public_coin_lb}]
    By \Cref{lem:instrument_minimax_maximin_bounds}, the number of distributed nodes $m$ must be large enough such that 
    \begin{equation}
        \min_{V_1, \dots, V_{d^2}}\max_{\mathcal{I}_1,...,\mathcal{I}_m}\QChi{\mathbb{E}_{\bfz}\left[\bigotimes_{i=1}^m\mathcal{I}_i(\rho_{\bfz})\right]}{\bigotimes_{i=1}^m\mathcal{I}_i\left(\mmstate\right)}\geq\frac{1}{16},
    \end{equation}
    where the expectation is over a random string $\bfz \sim \{-1,+1\}^{\ell}$, and $\rho_{\bfz}$ is the corresponding state from \Cref{def:hard-instance-mic} for the choice of $V_1, \dots, V_{d^2}$. By \Cref{lem:instrument_T_bound}, this means we must have
    \begin{equation}
        m\geq \Omega\left(\frac{d\ell}{\epsilon^2}\cdot\frac{1}{\|\mathcal{V}^\dagger T_\Lambda\mathcal{V}\|_2}\right),
    \end{equation}
    for all choices of $\mc{V}$. Notice that for every choice of basis $V_1,...,V_{d^2-1}$, $\mathcal{V}$ is an isometry; consequently, we have
    \begin{align}
        \|\mathcal{V}^\dagger T_\Lambda\mathcal{V}\|_2&\leq \|T_\Lambda\|_2\\
        &\leq \frac{1}{m}\sum_{i=1}^m\|M_{\Lambda_i}^\dagger M_{\Lambda_i}\|_2\\
        &\leq \frac{1}{m}\sum_{i=1}^m\|M_{\Lambda_i}^\dagger \|_2\|M_{\Lambda_i}\|_\infty\\
        &\leq d_q\sqrt{dd_c}\sqrt{d}=d_qd\sqrt{d_c},
    \end{align}
    where in the second inequality, we used the triangle inequality. In the third, we used Holder's inequality, noting that $\frac{1}{2}+\frac{1}{\infty}=\frac{1}{2}$ by the usual convention. Finally, we used \Cref{lem:instrument_2_bound} and \Cref{lem:instrument_infty_bound} in the last step. Substituting this into \Cref{lem:instrument_T_bound}, and setting $\ell=d^2-1$, we then get $m=\Omega\left(\frac{d^2}{d_q\sqrt{d_c}\epsilon^2}\right)$, as required. 
\end{proof}

Next, we prove our private-coin lower bound.

\begin{proof}[Proof of \Cref{thm:private_coin_lb}]
    Fix quantum instruments $\mathcal{I}_1,\dots,\mathcal{I}_m$. Here, we can choose $\mathcal{V}$ to minimize $\|\mathcal{V}^\dagger T_\Lambda \mathcal{V}\|_2$, for fixed channels. We start by restricting to the traceless subspace, $H\coloneqq \{X=X^\dagger\mid  X=0\}$, which has dimension $d^2-1$. Let $P$ be the projector onto $H$, and consider the positive semidefinite operator on this subspace, $B\coloneqq PT_\Lambda P$. Denote its eigenvalues by $0\leq \mu_1\leq...\leq \mu_N$. Then, to minimize this norm, we select $V_1,...,V_\ell$ to be the eigenvectors corresponding to the $\ell$ smallest eigenvalues, for $\ell$ to be chosen momentarily. Then, with this choice, we have that $\mathcal{V}^\dagger T_\Lambda\mathcal{V}$ is a diagonal matrix with entries $\mu_1,...,\mu_\ell$ along its diagonal, and hence
    \begin{equation}
        \|\mathcal{V}^\dagger T_\Lambda\mathcal{V}\|_2^2=\sum_{i=1}^\ell \mu_i^2.
    \end{equation}
    Each of the $\ell$ smallest eigenvalues are at most the average of the remaining eigenvalues, so
    \begin{align}
        \|\mathcal{V}^\dagger T_\Lambda\mathcal{V}\|_2&\leq \sqrt{\ell}\frac{\sum_{i=\ell+1}^{d^2-1}\mu_i}{d^2-\ell-1}\\
        &\leq \sqrt{\ell}\frac{\Tr(T_\Lambda)}{d^2-l-1}\\
        &=\frac{\sqrt{\ell}}{d^2-\ell-1}\frac{\sum_{i=1}^m\|M_{\Lambda_i}\|_2^2}{m}\\
        &\leq \frac{\sqrt{\ell}}{d^2-\ell-1}dd_q^2d_c.
    \end{align}
    We choose $\ell\coloneqq d^2/2$, which gives that $\|\mathcal{V}^\dagger T_\Lambda\mathcal{V}\|_2\leq \bigo(d_q^2 d_c)$. Plugging this in to \Cref{lem:instrument_T_bound}, we get that
    \begin{equation}
        m\geq \Omega\left(\frac{d^3}{d_cd_q^2\epsilon^2}\right),
    \end{equation}
    as required.
\end{proof}
\subsection{Proof of Lemmas~\ref{lem:instrument_T_bound},~\ref{lem:instrument_infty_bound} and~\ref{lem:instrument_2_bound}}\label{s:three_lemmas}
We now aim to prove Lemmas~\ref{lem:instrument_T_bound},~\ref{lem:instrument_infty_bound} and~\ref{lem:instrument_2_bound}, which were simply assumed in the above. To do this, we need a few more results. First is the Kadison-Schwarz inequality:
\begin{lemma}[Kadison--Schwarz]\label{lem:Kadison_Schwarz}
    Let $\Gamma$ be a unital completely positive map. Then, we have that
    \begin{equation}
        \Gamma(X^\dagger X)\succeq\Gamma(X)^\dagger\Gamma(X),
    \end{equation}
    where $A\preceq B$ if and only if $B-A$ is positive semidefinite. Here, unital means that $\Gamma(I)=I$.
\end{lemma}
We abuse notation slightly in the above definition by omitting the dimensions of the identity matrices.
Second is a standard moment-generating-function bound for random bitstrings. For completeness, we provide pointers to standard results in \cite{vershynin2018high} that allow its proof.
\begin{lemma}[{\cite[Claim IV.17; Prop 8.13]{acharya2020inferenceinformationconstraints, etingof2011introductionrepresentationtheory}}]\label{lem:exponential_expectation_bound}
    Let $z,z'\sim\{-1,1\}^\ell,~\lambda\in\mathbb{R},~A\in\mathbb{R}^{\ell\times\ell}$. Then,
    \begin{equation}
        \mathbb{E}_{z,z'}[\exp(\lambda z^\top Az')]\leq \exp(C\lambda^2\|A\|_2^2),
    \end{equation}
    whenever $\lambda\leq \frac{1}{2\|A\|_\infty}$, $C>0$ is some absolute constant.
\end{lemma}
\begin{proof}
    By Lemmas 6.2.3 and 6.2.4 of \cite{vershynin2018high}, it suffices to show that $z,z^\prime$ have constant subgaussian norms (see \cite[Definition 2.6.4]{vershynin2018high}). By \cite[Exercise 2.24(c)]{vershynin2018high}, a Rademacher random variable has constant subgaussian norm, and by \cite[Lemma 3.4.2]{vershynin2018high}, so do vectors of independent Rademacher variables, as desired. 
\end{proof}
With this in hand, we start by proving \Cref{lem:instrument_T_bound}.
 \begin{proof}[Proof of~\Cref{lem:instrument_T_bound}]
    For ease of notation, let
    \begin{equation}
        P=\mathbb{E}_{\bfz}\left[\bigotimes_{i=1}^m\mathcal{I}_i(\rho_{\bfz})\right],~Q=\bigotimes_{i=1}^m\mathcal{I}_i\left(\mmstate\right)=\bigotimes_{i=1}^m\beta_i.
    \end{equation}
    By \Cref{lem:quantum_ingster-suslina_symm}, we have that 
    \begin{equation}
        1+\QChi{P}{Q}\leq \mathbb{E}_{\bfz,\bfz'}\left[\exp\left(\sum_{i=1}^m Z_i(\bfz,\bfz')\right)\right],
    \end{equation}
    with
    \begin{align}
    Z_i(z,z') &= \Tr(\beta_i^{-1/2} (\mathcal{I}\left(\rho_z) - \beta_i\right) \beta_i^{-1/2}(\mathcal{I}\left(\rho_{z'}) - \beta_i\right)) \\ 
    &=\Tr(\beta_{i}^{-1/2}\mathcal{I}_i(\bar\Delta_z)\beta_{i}^{-1/2}\mathcal{I}_i(\bar\Delta_{z'}))\\
        &= \Tr(\beta_{i}^{-1/4}\mathcal{I}_i(\bar\Delta_z)\beta_{i}^{-1/4}\beta_{i}^{-1/4}\mathcal{I}_i(\bar\Delta_{z'})\beta_{i}^{-1/4})\\
        &=\Tr(\Lambda_i(\bar\Delta_z)\Lambda_i(\bar\Delta_{z'})),
    \end{align}
   where we have used the fact that $\rho_z = \frac{\mathds{1}}{d} + \bar\Delta_z$ to go from the first line to the second line. We then use the definition 
    \begin{equation}
        \bar\Delta_z=\frac{c\epsilon N_z}{\sqrt{d\ell}}\sum_{i=1}^\ell z_iV_i,
    \end{equation}
    and apply linearity of the quantum instruments to get that 
    \begin{align}
        \vecmap(\Lambda_i(\bar\Delta_z))&=\frac{c\epsilon N_z}{\sqrt{d\ell}}\vecmap\left(\Lambda_i\left(\sum_{i=1}^\ell z_iV_i\right)\right)\\
        &=\frac{c\epsilon N_z}{\sqrt{d\ell}}M_{\Lambda_i}\vecmap\left(\sum_{i=1}^\ell V_iz_i\right)\\
        &=\frac{c\epsilon N_z}{\sqrt{d\ell}}M_{\Lambda_i}\mathcal{V}z.
    \end{align}
    This implies that
    \begin{align}
        Z_i(z,z')&=\frac{c^2\epsilon^2N_zN_{z'}}{d\ell}z^\top\mathcal{V}^\dagger M_{\Lambda_i}^\dagger M_{\Lambda_i}\mathcal{V}z'\\
        \sum_{i=1}^mZ_i(z,z')&= \frac{mc^2\epsilon^2N_zN_{z'}}{d\ell}z^\top\mathcal{V}^\dagger T_\Lambda\mathcal{V}z'.
    \end{align}
    Hence, letting $A=\mathcal{V}^\dagger T_\Lambda\mathcal{V}$, we have:
    \begin{equation}
        1+\QChi{P}{Q}\leq \mathbb{E}_{z,z'}\left[\exp\left( \frac{mc^2\epsilon^2N_zN_{z'}}{d\ell}z^\top Az'\right)\right].
    \end{equation}
    Using an argument from the proof of Liu and Acharya~\cite[Lemma B.8]{liu2024role} on $N_z$, this normalization factor equals $1$ except with exponentially small probability. In particular, we can write
    \begin{equation}
        1+\QChi{P}{Q}\leq \mathbb{E}_{z,z'}\left[\exp\left( \frac{mc^2\epsilon^2}{d\ell}z^\top Az'\right)\right]+\frac{4}{e^d}.
    \end{equation}
    Then, setting $\lambda=\frac{mc^2\epsilon^2}{d\ell}>0$, we can analyze two cases. First, if $\lambda>\frac{1}{2\|A\|_\infty}$, we have that
    \begin{align}
        m&>\frac{d\ell}{2c^2\epsilon^2\|A\|_\infty}
        \geq\frac{d\ell}{2c^2\epsilon^2\|A\|_2}
        =\Omega\left(\frac{d\ell}{\epsilon^2\|A\|_2}\right),
    \end{align}
    which is the required bound. Otherwise, we have $\lambda\leq\frac{1}{2\|A\|_\infty}$, and then, by \Cref{lem:exponential_expectation_bound}, we get
    \begin{align}
        \mathbb{E}_{z,z'}\left[\exp\left(\lambda z^\top Az'\right)\right]&\leq \exp(C\lambda^2\|A\|_2^2),
    \end{align}
    so that
    \begin{align}
        1+\QChi{P}{Q}&\leq\exp\left(C\frac{m^2\epsilon^4}{d^2\ell^2}\|\mathcal{V}^\dagger T_\Lambda\mathcal{V}\|_2^2\right)+4e^{-d}.
    \end{align}
    Now, for $d\geq 6$, $4e^{-d} < \frac{1}{100}$, and so $\QChi{P}{S}\leq 1/16$, unless
    $
        \frac{m^2\epsilon^4}{d^2\ell^2}\|\mathcal{V}^\dagger T_\Lambda\mathcal{V}\|_2^2=\Omega(1)
    $, or, equivalently,
    \[
        m\geq\Omega\left(\frac{d\ell}{\epsilon^2}\cdot\frac{1}{\|\mathcal{V}^\dagger T_\Lambda\mathcal{V}\|_2}\right),
    \]
    which is the desired bound.
\end{proof}
Next, we prove \Cref{lem:instrument_infty_bound}.
\begin{proof}[Proof of~\Cref{lem:instrument_infty_bound}]
     Define $F \triangleq\mathcal{I}(I)=d\beta$, and $\Gamma(X)\triangleq F^{-1/2}\mathcal{I}(X)F^{-1/2}$. We first show $\|F^{-1/4}\mathcal{I}(X)F^{-1/4}\|_2\leq \|X\|_2$. Note that $\Gamma$ is completely positive, unital, and that $F$ is also Hermitian. Hence, by Kadison--Schwarz (\Cref{lem:Kadison_Schwarz}), we get that
    $\Gamma(X)^\dagger\Gamma(X)\preceq \Gamma(X^\dagger X)$, that is, 
    \begin{align}
       F^{-1/2}\mathcal{I}(X)^\dagger F^{-1}\mathcal{I}(X)F^{-1/2}&\preceq F^{-1/2}\mathcal{I}(X^\dagger X)F^{-1/2}\\
       \mathcal{I}(X)^\dagger F^{-1}\mathcal{I}(X)&\preceq \mathcal{I}(X^\dagger X).
   \end{align}
   By the same argument, we also get the other direction:
   \begin{equation}
       \mathcal{I}(X) F^{-1}\mathcal{I}(X)^\dagger\preceq \mathcal{I}(X X^\dagger).
   \end{equation}
   We can then take traces of both sides and use the fact that $\mathcal{I}$ is trace preserving to get that
    \begin{align}
        \Tr(\mathcal{I}(X)^\dagger F^{-1}\mathcal{I}(X))&\leq \Tr(\mathcal{I}(X^\dagger X))
        =\Tr(X^\dagger X)
        =\|X\|_2^2.
    \end{align}
    A similar argument yields that
    \begin{equation}
        \Tr(\mathcal{I}(X)F^{-1}\mathcal{I}(X)^\dagger)\leq \|X\|_2^2.
    \end{equation}
   We can decompose $F$ into its eigenbasis $\{|a\rangle\}$ with respective eigenvalues. Then, we have that, for every $i,j$,
    \begin{align}
    \langle i|F^{-1/4}\mathcal{I}(X)F^{-1/4}|j\rangle &=\sum_{k,\ell}\lambda_k^{-1/4}\lambda_\ell^{-1/4}\langle i|k\rangle \langle k|\mathcal{I}(X)|\ell\rangle\langle \ell|j\rangle\\
    &=\lambda_i^{-1/4}\lambda_j^{-1/4}\langle i|\mathcal{I}(X)|j\rangle\\
    &=\lambda_i^{-1/4}\lambda_j^{-1/4}\mathcal{I}(X)_{i,j}
    \end{align}
    which implies
    \begin{align}
    \|\langle i|F^{-1/4}\mathcal{I}(X)F^{-1/4}|j\rangle\|^2&=\lambda_i^{-1/2}\lambda_j^{-1/2}\|\mathcal{I}(X)_{i,j}\|^2
    \end{align}
    and so 
    \begin{align}
    \|F^{-1/4}\mathcal{I}(X)F^{-1/4}\|_2^2&=\sum_{i,j}\frac{\|\mathcal{I}(X)_{i,j}\|^2}{\sqrt{\lambda_i\lambda_j}}.
    \end{align}
    By the AM-GM inequality, we have that $\frac{1}{\sqrt{\lambda_i \lambda_j}}\leq \frac{1}{2}\left(\frac{1}{\lambda_i}+\frac{1}{\lambda_j}\right)$, so
    \begin{align}
        \|F^{-1/4}\mathcal{I}(X)F^{-1/4}\|_2^2&\leq \frac{1}{2}\sum_{i,j}\left(\frac{1}{\lambda_i}+\frac{1}{\lambda_j}\right)\|\mathcal{I}(X)_{i,j}\|^2\\
        &=\frac{1}{2}\sum_{i,j}\frac{1}{\lambda_i}\|\mathcal{I}(X)_{i,j}\|^2+\frac{1}{2}\sum_{i,j}\frac{1}{\lambda_j}\|\mathcal{I}(X)_{i,j}\|^2.
    \end{align}
    Considering the first term in the sum:
    \begin{align}
        \frac{1}{2}\sum_{i,j}\frac{1}{\lambda_i}\|\mathcal{I}(X)_{i,j}\|^2&=\frac{1}{2}\sum_i\frac{1}{\lambda_i}\sum_j\|\mathcal{I}(X)_{i,j}\|^2.
    \end{align}
    We next note that for any matrix $B$, we have that $(BB^\dagger)_{ac}=\sum_bB_{ab}B^\dagger_{bc}$. Letting $c=a$, we have that $(BB^\dagger)_{aa}=\sum_b|B_{ab}|^2$. Substituting this above, we get that
    \begin{align}
        \frac{1}{2}\sum_i\frac{1}{\lambda_i}\sum_j\|\mathcal{I}(X)_{i,j}\|^2&=\frac{1}{2}\sum_i\frac{1}{\lambda_i}(\mathcal{I}(X)\mathcal{I}(X)^\dagger)_{ii}
        =\frac{1}{2}\Tr(F^{-1}\mathcal{I}(X)\mathcal{I}(X)^\dagger)
        =\frac{1}{2}\Tr(\mathcal{I}(X)^\dagger F^{-1}\mathcal{I}(X))\,.
    \end{align}
    Similarly, for the second term, we have that:
    \begin{equation}
        \frac{1}{2}\sum_{i,j}\frac{1}{\lambda_j}\|\mathcal{I}(X)_{i,j}\|^2=\frac{1}{2}\Tr(\mathcal{I}(X)F^{-1}\mathcal{I}(X)^\dagger).
    \end{equation}
    Putting both together, we get that
    \begin{equation}
        \|F^{-1/4}\mathcal{I}(X)F^{-1/4}\|_2^2\leq \frac{1}{2}\Tr(\mathcal{I}(X)^\dagger F^{-1}\mathcal{I}(X))+\frac{1}{2}\Tr(\mathcal{I}(X)F^{-1}\mathcal{I}(X)^\dagger).
    \end{equation}
    But recall the right hand side now is upper bounded by $\frac{1}{2}\|X\|_2^2+\frac{1}{2}\|X\|_2^2$. Hence, this gives that $\|F^{-1/4}\mathcal{I}(X)F^{-1/4}\|_2^2\leq \frac{1}{2}2\|X\|_2^2=\|X\|^2_2$.
    Then, as a map, this implies that
    \begin{equation}
        \sup_{\|X\|_2\neq 0}\frac{\|F^{-1/4}\mathcal{I}(X)F^{-1/4}\|_2}{\|X\|_2}\leq 1.
    \end{equation}
    Further, using that $\beta=F/d$, we get that
    \begin{equation}
        \sup_{\|X\|_2\neq 0}\frac{\|\Lambda(X)\|_2}{\|X\|_2}\leq \sqrt d.
    \end{equation}
    Then, recalling the Liouville representation $M_\mathcal{\Lambda}$ of the map $\Lambda$ , we have that
    \begin{equation}
        \|M_{\Lambda}\|_\infty=\sup_{\|\vecmap(X)\|_2 \neq 0} \frac{\|M_{\Lambda}\vecmap(X)\|_2}{\vecmap(X)} =  \sup_{\|X\|_2\neq 0}\frac{\|\Lambda(X)\|_2}{\|X\|_2}\leq \sqrt d,
    \end{equation}
    as required.
\end{proof}
We are left with the last building block, \Cref{lem:instrument_2_bound}.
\begin{proof}[Proof of~\Cref{lem:instrument_2_bound}]
     Recall that we defined $\beta=\mathcal{I}\left(\mmstate\right)$. Then, by the definition of the instrument, we have that
    \begin{equation}
        \beta=\mathcal{I}\left(\mmstate\right)=\sum_{c=1}^{d_c}\beta_c\otimes|c\rangle\langle c|,
    \end{equation}
    where $\beta_c$ = $\mathcal{I}_c(I/d)$, the individual maps comprising the instrument. Also, we know that
    \begin{equation}
        \|M_\Lambda\|_2^2=\sum_j s_j(M_\Lambda)^2\leq \rank(M_\Lambda)\|M_\Lambda\|_\infty^2,
    \end{equation}
    where $s_j$ denotes the $j$-th singular value. Let $r_c=\rank(\beta_c)$. Recall that this register is $d_q$-dimensional, and so $r_c \leq d_q$. Then, this implies the rank of the operator $M_{\mathcal{I}_c}$ is at most $d_q^2$. Furthermore, there are $d_c$ basis vectors for the classical portion of the quantum channel, and hence $\rank(M_\Lambda)\leq d_cd_q^2$. Further, from \Cref{lem:instrument_infty_bound}, we know $\|M_\Lambda\|_\infty\leq \sqrt{d}$. Then, putting things together, we get that 
    \begin{align}
        \|M_\Lambda\|_2^2&\leq d_q^2dd_c.
    \end{align}
    That is, $\|M_\Lambda\|_2\leq d_q\sqrt{dd_c}$, as required.
\end{proof}
\printbibliography
\appendix

\section{Intermediate lemmas for compression bounds}
\label{app:compression-bounds}

We prove here the intermediate lemmas used for our compression bounds in \Cref{s:public_ub}. We will need to perform high-order integrals with respect to the unitary Haar measure, for which we introduce the Weingarten calculus. In general, $k$-order moments of random unitaries can be related to the elements of the symmetric group $\mc{S}_k$, as formalized in the following standard lemma (see e.g., \cite{mele2024introduction}).

\begin{lemma}
\label{lem:weingarten-calc}
Given a permutation $\pi \in \mc{S}_k$, let $P_\pi \in \mathbb{C}^{d^k \times d^k}$ be the associated permutation operator and $\wg(\pi,d)$ be its Weingarten coefficient. Let $M \in \mathbb{C}^{d^k \times d^k}$. Then,
    \begin{equation}
        \mathbb{E}_{\bfU \sim U(d)}[\bfU^{\otimes k} M \bfU^{\dag \otimes k}] = \sum_{\pi,\tau \in \mc{S}_k} \wg(\pi^{-1}\tau,d) \Tr(P_\tau^\dagger M) P_\pi.   
    \end{equation}
\end{lemma}

Note that we haven't defined the Weingarten coefficients mentioned above. These are real coefficients associated with each permutation and dependent on the dimension $d$. In this section, we will only be interested in computing fourth-order moments, and state bounds on these Weingarten coefficients specifically.

\begin{lemma}[{\cite[Tables II \& III]{brouwer1996diagrammatic}}]
\label{lem:weingarten-coeffs}
    Let $\mathrm{id}_4$ be the identity permutation in $S_4$. Then, 
    \begin{equation}
        |\wg(\mathrm{id}_4,d)| \leq d^{-4}, \quad \text{and } \forall \text{ }\mathrm{id}_4 \neq \pi \in \mc{S}_4, \quad \wg(\pi,d) \leq \bigo(d^{-6}).
    \end{equation}
\end{lemma}

\subsection{Second-moment bounds}

We first prove \Cref{lem:expected-distance-improved}.

\begin{proof}[Proof of \Cref{lem:expected-distance-improved}]
    Recall that we wish to bound $\mathbb{E}[\|\Phi_{\bfU}(\rho) - \Phi_{\bfU}(\sigma)\|_2^2] = \mathbb{E}[\Tr((\Phi_{\bfU}(\rho) - \Phi_{\bfU}(\sigma))^2)]$. Let $\Delta = \rho - \sigma$. By linearity of $\Phi_{\bfU}$, we can rewrite the inner trace as 
    \begin{align}
        \Tr(\Phi_{\bfU}(\Delta)^2) &= \Tr(\swap_{Q_1,Q_2} \cdot \Phi_{\bfU}(\Delta)^{\otimes 2})
        \\&= \Tr(\swap_{Q_1,Q_2} \otimes I_{B_1,B_2} \cdot ({\bfU}\Delta {\bfU}^\dag)^{\otimes 2}) \triangleq \Tr(M_{\bfU}) \label{eq:second-moment-1},
    \end{align}
    where the second equality used $\Phi_{\bfU}(M) = \Tr_B(UMU^\dag)$.
    By \Cref{cor:second-order}, and noting that $\Tr(\Delta) = 0$, we have
    \begin{align}
        \mathbb{E}[\Tr(M_{\bfU})] &= \frac{\Tr(\Delta^2)}{d^2-1}\Tr(\swap_{Q_1,Q_2} \otimes I_{B_1,B_2} \cdot \swap_{1,2}) -\frac{\Tr(\Delta^2)}{d(d^2-1)} \Tr(\swap_{Q_1,Q_2} \otimes I_{B_1,B_2} \cdot I_{1,2})
        \\&= \frac{\Tr(\Delta^2)}{d^2-1} \Tr(I_{Q_1,Q_2} \otimes \swap_{B_1,B_2}) - \frac{\Tr(\Delta^2)}{d(d^2-1)} \Tr(\swap_{Q_1,Q_2} \otimes I_{B_1,B_2})
        \\&= \frac{\Tr(\Delta^2)}{(d^2-1)} \left(d_q^2d_B - \frac{d_q d_B^2}{d}\right) = \frac{\Tr(\Delta^2)}{(d^2-1)} \left(d_q^2d_B - d_B\right)
        \\&\geq \frac{\Tr(\Delta^2)}{d^2} \cdot \frac{d_q^2 d_B}{2} 
        \\&= \frac{d_q}{2d} \cdot \|\rho - \sigma\|_2^2,
    \end{align}
    where the second equality used $\swap_{1,2} = \swap_{Q_1,Q_2} \otimes \swap_{B_1,B_2}$, the inequality holds for $d_q \geq 2$, and the final step used $d = d_qd_B$ and $\Tr(\Delta^2) = \Tr((\rho-\sigma)^2) = \|\rho-\sigma\|_2^2$.
\end{proof}
We next prove the generalization \Cref{lemma:first_moment_upsilon}. To do this, we will first establish some intermediate notation and lemmas. For each branch $\mathcal{I}_{U,c}(X)$, let us define
\begin{equation}
    \mathcal{I}_{U,c}(X)\triangleq \Tr_A[(\langle c|_C\otimes I_{QA})UXU^\dagger(|c\rangle_C\otimes I_{QA}]=\Tr_A(B_c(U)),
\end{equation}
where $B_c(U)\triangleq (\langle c|_C\otimes I_{QA})UXU^\dagger(|c\rangle_C\otimes I_{QA})$ for every classical label $c$.
\begin{lemma}\label{lemma:swap_branch_representation}
     $\|\mathcal{I}_{U,c}(X)\|_2^2=\Tr[(B_c(U)\otimes B_c(U))(\swap_{Q_1,Q_2}\otimes I_{A_1,A_2})]$ for every classical label $c$.
\end{lemma}
\begin{proof}
    Using that $\Tr[(C\otimes D)\swap]=\Tr[CD]$, we have that
    \begin{equation}
    \|\Tr_A[B_c(U)]\|_2^2=\Tr[(\Tr_A[B_c(U)\otimes\Tr_A(B_c(U))\swap_{Q_1,Q_2}])]=\Tr[(B_c(U)\otimes B_c(U))(\swap_{Q_1,Q_2}\otimes I_{A_1,A_2})].
    \end{equation}
    The second equality we get by recalling that we can pull back partial traces, in the sense that $\Tr[(\Tr_{C,D}Z)Y)]=\Tr[Z(Y\otimes I_{D_1,D_2})]$, where $I_{D_1,D_2}$ is shorthand for $I_{D_1}\otimes I_{D_2}$.
\end{proof}
Using this, we can get another way to express $\|\mathcal{I}_U(X)\|_2^2$:
\begin{lemma}\label{lemma:upsilon_expression}
    Define an operator acting on $(C_1Q_1A_1)\otimes(C_2Q_2A_2)$ by
    \begin{equation}
        R=\sum_{c=1}^{d_c}|c\rangle\langle c|_{C_1}\otimes |c\rangle\langle c|_{C_2}\otimes \swap_{Q_1,Q_2}\otimes I_{A_1,A_2}.
    \end{equation}
    Then, for every traceless Hermitian $X$, $\|\mathcal{I}_U(X)\|_2^2=\Tr[R(UXU^\dagger)^{\otimes 2}]$.
\end{lemma}
\begin{proof}
    Recall \Cref{lemma:swap_branch_representation} gives the contribution for each branch of a classical label $c$, over the space of registers $QA$. We can embed this back into the full $CQA$ space by inserting the projectors corresponding to the classical registers, for each copy. Then, we can take the sum over all $c$, which gives that 
    \begin{equation}
        \|\mathcal{I}_U(X)\|_2^2=\Tr[R(UXU^\dagger)^{\otimes 2}]
    \end{equation}
\end{proof}
This latter quantity we already analyzed in \Cref{lem:single_string_second_moment}, and we mirror the analysis here.
\begin{lemma}\label{lemma:traceless_second_moment}
    Let $X$ be a traceless Hermitian, and let $s=\|X\|_2^2$. Then, 
    \begin{equation}
        \mathbb{E}_{\bfU}[(\bfU X\bfU^\dagger)^{\otimes 2}]=-\frac{s}{d(d^2-1)}I+\frac{s}{d^2-1}\swap
    \end{equation}
\end{lemma}
\begin{proof}
    The lemma is immediate from \Cref{cor:second-order} and the fact that $X$ is traceless.
\end{proof}

We can also find the trace of the operator $R$:
\begin{lemma}\label{lemma:R_trace}
    We have that $\Tr[R]=d_cd_qd_A^2$ and $\Tr[R \swap]=d_cd_q^2d_A$, where $\swap$ denotes the full swap operator on $C_1Q_1A_1$ and $C_2Q_2A_2$.
\end{lemma}
\begin{proof}
    For a fixed value of $c$, the term inside the sum of $R$ is given by $|c\rangle\langle c|_{C_1}\otimes |c\rangle\langle c|_{C_2}\otimes \swap_{Q_1,Q_2}\otimes I_{A_1,A_2}$. This has trace equal to $\Tr[\swap_{Q_1,Q_2}]\cdot\Tr[I_{A_1,A_2}]=d_qd_A^2$. Summing over all $d_c$ labels gives $\Tr[R]=d_cd_qa^2$. For the second, we can view the overall swap operator as a tensor of swaps $\swap= \swap_C\otimes \swap_Q\otimes \swap_A$. For a fixed label $c$, we have that $\Tr[|c\rangle\langle c|\otimes |c\rangle\langle c|\swap_C]=1$. Next, $\Tr[\swap_{Q_1,Q_2}\swap_Q]=\Tr[I_{Q_1,Q_2}]=d_q^2$. Finally, $\Tr[I_{A_1,A_2}\swap_A]=\Tr[\swap_A]=d_A$. Then, summing over all $d_c$ labels gives $\Tr[R\swap]=d_cd_q^2d_A$.
\end{proof}
We now prove \Cref{lemma:first_moment_upsilon}.
\begin{proof}[Proof of \Cref{lemma:first_moment_upsilon}]
    Starting from \Cref{lemma:upsilon_expression}, we get that $\|\mathcal{I}_U(X)\|_2^2=\Tr[R(UXU^\dagger)^{\otimes 2}]$. Then, we can take expectations, and apply \Cref{lemma:traceless_second_moment}, to get that
    \begin{equation}
        \mathbb{E}_{\bfU}[\|\mathcal{I}_{\bfU}(X)\|_2^2]=\Tr\left[R\left(-\frac{s}{d(d^2-1)}I+\frac{s}{d^2-1}\swap\right)\right].
    \end{equation}
    We can bound the operators with \Cref{lemma:R_trace} to get that this quantity equals
    \begin{align}
        \mathbb{E}_{\bfU}[\|\mathcal{I}_{\bfU}(X)\|_2^2]&=-\frac{s}{d(d^2-1)}d_cd_qd_A^2+\frac{s}{d^2-1}d_cd_q^2d_A\\
        &=-\frac{s}{(d^2-1)}d_A+\frac{s}{d^2-1}dd_q\\
        &=\frac{d_qd-d_A}{d^2-1}\|X\|_2^2.
    \end{align}
    To see the final inequality, note $d_A=\frac{d}{d_cd_q}$. Then, we can rewrite $d_qd-d_A=d(d_q-1/(d_cd_q))$. Assuming $d_cd_q^2\geq 2$, we have that $d_q-1/d_cd_q\geq d_q/2$, and hence $dd_q-d_A\geq dd_q/2$. Then, using that $d^2-1\leq d^2$, we get the stated inequality.
\end{proof}

\subsection{Fourth-moment bounds}

We now prove \Cref{lem:expected-distance-fourth-improved}.

\begin{proof}[Proof of \Cref{lem:expected-distance-fourth-improved}]
    Recalling the definition of $M_U$ from the proof of \Cref{lem:expected-distance-improved}, we have
    \begin{align}
        \mathbb{E}[\|\Phi_{\bfU}(\rho) - \Phi_{\bfU}(\sigma)\|_2^4] &= \mathbb{E}[\Tr(M_{\bfU})^2] = \mathbb{E}[\Tr(M_{\bfU} \otimes M_{\bfU})]
        \\&= \mathbb{E}[\Tr(\swap_{Q_1,Q_2} \otimes \swap_{Q_3,Q_4} \otimes I_{B_1, \dots, B_4} \cdot {\bfU}^{\otimes 4} \Delta^{\otimes 4} {\bfU}^{\dag \otimes 4})].
    \end{align}
    Going forward, let $\gamma = \swap_{1,2} \otimes \swap_{3,4} \in \mc{S}_4$. We will now use \Cref{lem:weingarten-calc} with $k = 4$ to average over ${\bfU}$. As $\Delta$ is traceless, $\Tr(P_\tau^\dag \Delta^{\otimes 4})$ in \Cref{lem:weingarten-calc} is only non-zero when $\tau$ has cycle type $(4)$ or $(2,2)$. Thus, we have
    \begin{align}
         \mathbb{E}[\Tr(\gamma_Q \otimes I_B \cdot {\bfU}^{\otimes 4} \Delta^{\otimes 4} {\bfU}^{\dag \otimes 4})] &= \sum_{\pi \in \mc{S}_4} \Tr(\pi_{Q,B} \cdot \gamma_Q \otimes I_B) \sum_{\tau \in \mathrm{cyc}(2,2) \cup \mathrm{cyc}(4)} \wg(\pi^{-1} \tau) \Tr(\tau^{-1} \Delta^{\otimes 4})
         \\&= \sum_{\pi \in \mc{S}_4} \Tr(\pi_{Q} \cdot \gamma_Q) \Tr(\pi_B) \sum_{\tau \in \mathrm{cyc}(2,2) \cup \mathrm{cyc}(4)} \wg(\pi^{-1} \tau) \Tr(\tau^{-1} \Delta^{\otimes 4}).
    \end{align}
    We will also use the fact that $\Tr(\Delta^4) = \|\Delta\|_4^4 \leq \|\Delta\|_2^4 = \Tr(\Delta^2)^2$, where the inequality follows from monotonicity of Schatten norms. This implies that $\Tr(\tau^{-1} \Delta^{\otimes 4}) \leq \Tr(\Delta^2)^2$ for any choice of $\tau$. We now split the above summation into cases depending on $\pi$'s cycle type. Note that for any permutation $\pi$, we have $\Tr(\pi) = d^{\# \pi}$, where $\# \pi$ is the number of cycles of $\pi$.
    
    First consider the case when $\pi \in \mathrm{cyc}(4)$. Note that $\Tr(\pi_B) = d_B$, and that, by \Cref{lem:weingarten-coeffs}, $\wg(\pi^{-1}\tau) \leq d^{-4}$ for any choice of $\tau$. Further, as $\pi$ and $\gamma$ have different cycle types, $\pi \cdot \gamma \neq \mathrm{id}$, implying that $\Tr(\pi_Q \cdot \gamma_Q) \leq d_q^3$. As there is only a constant number of choices for $\pi, \tau$, we have
    \begin{equation}
        \sum_{\pi \in \mathrm{cyc}(4)} \Tr(\pi_{Q} \cdot \gamma_Q) \Tr(\pi_B) \sum_{\tau \in \mathrm{cyc}(2,2) \cup \mathrm{cyc}(4)} \wg(\pi^{-1} \tau) \Tr(\tau^{-1} \Delta^{\otimes 4}) \leq \bigo \left(\frac{d_q^3 d_B}{d^4}\right) \cdot \Tr(\Delta^2)^2 = \bigo \left(\frac{d_q^2}{d^3}\right) \cdot \Tr(\Delta^2)^2. 
    \end{equation}

   Next, assume $\pi \in \mathrm{cyc}(3,1)$. As $\pi$ has two cycles, $\Tr(\pi_B) = d_B^2$. Since $\pi,\tau$ have different cycle types, by \Cref{lem:weingarten-coeffs}, $\wg(\pi^{-1}\tau) \leq \bigo(d^{-6})$ for any choice of $\tau$. Again, $\pi,\gamma$ have different cycle types, and so $\Tr(\pi_Q \cdot \gamma_Q) \leq d_q^3$. Altogether,

   \begin{equation}
        \sum_{\pi \in \mathrm{cyc}(3,1)} \Tr(\pi_{Q} \cdot \gamma_Q) \Tr(\pi_B) \sum_{\tau \in \mathrm{cyc}(2,2) \cup \mathrm{cyc}(4)} \wg(\pi^{-1} \tau) \Tr(\tau^{-1} \Delta^{\otimes 4}) \leq \bigo \left(\frac{d_q^3 d_B^2}{d^6}\right) \cdot \Tr(\Delta^2)^2 = \bigo \left(\frac{d_q}{d^4}\right) \cdot \Tr(\Delta^2)^2. 
    \end{equation}

    In the next case, let $\pi \in \mathrm{cyc}(2,2)$. As $\pi$ has two cycles, $\Tr(\pi_B) = d_B^2$, and, by \Cref{lem:weingarten-coeffs}, $\wg(\pi^{-1}\tau) \leq d^{-4}$ for any choice of $\tau$. Finally, as both $\pi$ and $\gamma$ have the same cycle type, $\Tr(\pi_Q \cdot \gamma_Q)$ is maximized by $\pi = \gamma^{-1}$, implying $\Tr(\pi_Q \cdot \gamma_Q) \leq d_q^4$. Then, 
    \begin{equation}
         \sum_{\pi \in \mathrm{cyc}(2,2)} \Tr(\pi_{Q} \cdot \gamma_Q) \Tr(\pi_B) \sum_{\tau \in \mathrm{cyc}(2,2) \cup \mathrm{cyc}(4)} \wg(\pi^{-1} \tau) \Tr(\tau^{-1} \Delta^{\otimes 4}) \leq \bigo \left(\frac{d_q^4 d_B^2}{d^4}\right) \cdot \Tr(\Delta^2)^2 = \bigo \left(\frac{d_q^2}{d^2}\right) \cdot \Tr(\Delta^2)^2. 
    \end{equation}

   Next, consider the case when $\pi \in \mathrm{cyc}(2,1,1)$. As $\pi$ has 3 cycles, $\Tr(\pi_B) = d_B^3$. Using \Cref{lem:weingarten-coeffs}, and as we sum over $\pi,\tau$ of different cycle types, we have $\wg(\pi^{-1}\tau) \leq \bigo(d^{-6})$. Again, as $\pi$ and $\gamma$ have different types, $\Tr(\pi_Q \gamma_Q) \leq d_q^3$, implying

   \begin{equation}
         \sum_{\pi \in \mathrm{cyc}(2,1,1)} \Tr(\pi_{Q} \cdot \gamma_Q) \Tr(\pi_B) \sum_{\tau \in \mathrm{cyc}(2,2) \cup \mathrm{cyc}(4)} \wg(\pi^{-1} \tau) \Tr(\tau^{-1} \Delta^{\otimes 4}) \leq \bigo \left(\frac{d_q^3 d_B^3}{d^6}\right) \cdot \Tr(\Delta^2)^2 = \bigo \left(\frac{1}{d^3}\right) \cdot \Tr(\Delta^2)^2. 
    \end{equation}

    In the final case, $\pi \in \mathrm{cyc}(1,1,1,1)$, i.e., $\pi = \mathrm{id}$. Note that $\Tr(\pi_B) = d_B^4$, and again, as $\pi \neq \tau$, $\wg(\pi^{-1}\tau) \leq \bigo(d^{-6})$ for any choice of $\tau$. Finally, $\Tr(\pi_Q \gamma_Q) = \Tr(\gamma_Q) = d_q^2$. This yields the bound
    \begin{equation}
        \sum_{\pi \in \mathrm{cyc}(1,1,1,1)} \Tr(\pi_{Q} \cdot \gamma_Q) \Tr(\pi_B) \sum_{\tau \in \mathrm{cyc}(2,2) \cup \mathrm{cyc}(4)} \wg(\pi^{-1} \tau) \Tr(\tau^{-1} \Delta^{\otimes 4}) \leq \bigo \left(\frac{d_q^2 d_B^4}{d^6}\right) \cdot \Tr(\Delta^2)^2 = \bigo \left(\frac{d_B^2}{d^4}\right) \cdot \Tr(\Delta^2)^2.
    \end{equation}

    It is easy to verify that the largest contribution is due to the third case, where $\pi$ has cycle type $(2,2)$. Thus, we obtain
    \begin{equation}
        \mathbb{E}[\|\Phi_{\bfU}(\rho) - \Phi_{\bfU}(\sigma)\|_2^4] \leq \bigo\left(\frac{d_q^2}{d^2}\right) \cdot \|\rho-\sigma\|_2^4,
    \end{equation}
    concluding the proof.
\end{proof}

To prove \Cref{lemma:upsilon_second_moment}, we require further notation regarding permutations as well as some more properties of Weingarten coefficients, which we present now. First, we assume all our permutations are in the symmetric group $S_4$, and introduce some notation.
\begin{definition}[Permutation notation]
    For a permutation $\pi$, we let $|\cyc(\pi)|$ be the number of cycles in the permutation (including singleton cycles). Further, let $|\pi|=4-|\cyc(\pi)|$ be the \emph{transposition length}.
\end{definition}
For example, $\pi=(12)(3)(4)$ has $|\cyc(\pi)|=3$ and $|\pi|=1$. Given four copies of a state $\rho$, with each copy defined on the space $H=C\otimes Q\otimes A$, and for a given permutation $\pi\in S_4$, we can define a permutation operator that acts on the copies according to the permutation $\pi$:
\begin{equation}
    P_\pi^H=P_\pi^C\otimes P_\pi^Q\otimes P^A_\pi,
\end{equation}
Note when expanded to four copies of a state, we can denote $R$ from \Cref{lemma:upsilon_expression} as $R_{1,2}$ and $R_{3,4}$ to operate on the pairs. Then, we let $\mathbf{R}=R_{1,2}\otimes R_{3,4}$. Then, we have that our second moment is given by
\begin{equation}
    \|\mathcal{I}_U(X)\|_2^4=\Tr[\mathbf{R}(UXU^\dagger)^{\otimes 4}].
\end{equation}

We will use the following standard result for inner products of permutation operators with tensor-copies of an operator; this can be verified via, say, tensor network diagrams. For a more general statement, we refer the reader to \cite[Section 2]{Collins_2010}.

\begin{lemma}\label{lemma:permutation_trace}
    Let $\pi\in S_4$ be a permutation, and $\gamma$ a cycle in $\cyc(\pi)$. Let $\ell(\gamma)$ denote the length of the cycle. Let $K$ be an operator on a finite-dimensional Hilbert space $\mc{H}$ and let $P_\pi$ be the representation of the permutation $\pi$ on $\mathcal{H}^{\otimes 4}$. Then, we have
    \begin{equation}
        \Tr[P_\pi K^{\otimes 4}]=\prod_{\gamma\in\cyc(\pi)}\Tr[K^{\ell(\gamma)}].
    \end{equation}
\end{lemma}

We also make use of another result by Collins and Nechita to bound the Weingarten coefficients \cite[Section 2]{Collins_2010}:
\begin{lemma}\label{lemma:Weingarten_bound}
    For every permutation $\pi\in S_4$, there is an absolute constant $C_W > 0$ such that $|\wg(\pi;d)|\leq C_wd^{-4-|\pi|}$.
\end{lemma}
We can prove one more small result about the composition of permutations:
\begin{lemma}\label{lemma:cycle_permutation}
    Let $\pi_1,\pi_2\in S_4$ such that $\pi_1$ has no fixed points. Then, letting, $\pi_3=\pi_1^{-1}\pi_2$, we have that
    \begin{equation}
        |\cyc(\pi_2)|\leq 2+|\pi_3|.
    \end{equation}
\end{lemma}
\begin{proof}
    As $|\pi|$ satisfies the triangle inequality, we have that
    \begin{equation}
        |\pi_2|=|\pi_1\pi_3|\geq |\pi_1|-|\pi_3|.
    \end{equation}
    Since $\pi_1$ has no fixed points it either is two swap cycles, or a cycle between all four elements. Hence, $|\cyc(\pi_1)|\leq 2$. This implies $|\pi_1|\geq 2$. We get the result by rearranging.
\end{proof}
We now need a specific result about the action of the permutation operator with our operator $\mathbf{R}$.
\begin{lemma}\label{lemma:R_permutation_bound}
    For every $\pi\in S_4$, we have that $|\Tr[\mathbf{R}P_\pi]|\leq d_c^2d_q^4d_A^{|\cyc(\pi)|}$.
\end{lemma}
\begin{proof}
    To write out the operators explicitly, let us denote
    \begin{equation}
        \mathbf{R}_C=\sum_{c,d}|c,c,d,d\rangle\langle c,c,d,d|_{C_1,C_2,C_3,C_4},
    \end{equation}
    which is a projector of rank $d_c^2$. Let us denote $\alpha = (12)(34)$. Then, recall $\mathbf{R}=R_{1,2}\otimes R_{3,4}$ which factors across the registers as $\mathbf{R}=\mathbf{R}_C\otimes P^Q_\alpha\otimes I_{A_1,A_2,A_3,A_4}$. Then, we can factor this among the traces of $P_\pi$ to get
    \begin{equation}
        \Tr[\mathbf{R}P_\pi]=\Tr[\mathbf{R}_CP_\pi^C]\Tr[P_\alpha^QP_\pi^Q]\Tr[P_\pi^A].
    \end{equation}
    We now aim to bound each individual term. Since each permutation operator is also a unitary, we have that
    \begin{equation}
        |\Tr[\mathbf{R}_CP_\pi^C]|\leq |\Tr[\mathbf{R}_C]|=d_c^2.
    \end{equation}
    For the second term, we apply the cycle trace identity to get that
    \begin{equation}
        \Tr[P_\alpha^QP_\pi^Q]=\Tr[P_{\alpha\pi}^Q]=d_q^{|\cyc(\alpha\pi)|}\leq d_q^4,
    \end{equation}
    where implicitly, the permutations are working on the operator $I_Q^{\otimes 4}$. With the same argument, we get that
    \begin{equation}
        \Tr[P_\pi^A]=d_A^{|\cyc(\pi)|}.
    \end{equation}
    Putting this all together, we get that $|\Tr[\mathbf{R}P_\pi]|\leq d_c^2d_q^4d_A^{|\cyc(\pi)|}$.
\end{proof}
With all this setup out of the way, we can finally prove the fourth-moment bound, \Cref{lemma:upsilon_second_moment}.
\begin{proof}[Proof of \Cref{lemma:upsilon_second_moment}]
    Recall by definition, $\|\mathcal{I}_U(X)\|_2^4=\Tr[\mathbf{R}(UXU^\dagger)^{\otimes 4}]$, and thus taking expectations, we have that
    \begin{equation}
        \mathbb{E}_{\bfU}[\|\mathcal{I}_{\bfU}(X)\|_2^4]=\Tr[\mathbf{R}\mathbb{E}_{\bfU}(\bfU X\bfU^\dagger)^{\otimes 4}].
    \end{equation}
    Using \Cref{lem:weingarten-calc}, we get an expression for this expectation
    \begin{equation}
        \mathbb{E}_{\bfU}[(\bfU X\bfU^\dagger)^{\otimes 4}]=\sum_{\pi,\tau\in S_k}\wg(\pi^{-1}\tau,d)\Tr[P_\pi^\dagger X^{\otimes 4}]P_\tau.
    \end{equation}
    Then, applying the operator $\mathbf{R}$ and taking traces, we can recover $\|\mathcal{I}_U(X)\|_2^4$:
    \begin{equation}
        \mathbb{E}_{\bfU}[\|\mathcal{I}_{\bfU}(X)\|_2^4]=\sum_{\pi,\tau\in S_k}\wg(\pi^{-1}\tau,d)\Tr[P_\pi^\dagger X^{\otimes 4}]\Tr[\mathbf{R}P_\tau].
    \end{equation}
    We can use \Cref{lemma:permutation_trace} to bound the middle term first. This gives that
    \begin{equation}
        \Tr[P_\pi^\dagger X^{\otimes 4}]=\prod_{\gamma\in\cyc(\pi)}\Tr[X^{l(\gamma)}].
    \end{equation}
    Since $X$ is traceless, this will only be nonzero for permutations with no singletons, which is either two cycles of length two, or a cycle of length four. In the first case, we get $\Tr[X^2]^2=\|X\|_2^4$. In the latter, we get $\Tr[X^4]$, which is at most $\|X\|_2^4$ for Hermitian $X$. Hence in either case, we get
    \begin{equation}
        \Tr[P_\pi^\dagger X^{\otimes 4}]\leq \|X\|_2^4.
    \end{equation}
    Using \Cref{lemma:permutation_trace}, we can also bound the latter term, which is at most $d_c^2d_q^4d_A^{|\cyc(\tau)|}$. For some $\pi$ that contributes, and for some arbitrary $\tau$, we can set $\upsilon=\pi^{-1}\tau$. Then, by \Cref{lemma:cycle_permutation}, we get that $|\Tr[\mathbf{R}P_\tau]|\leq d_c^2d_q^4d_A^{2+|\upsilon|}\leq d_c^2d_q^{4+|\upsilon|}d_A^{2+|\upsilon|}$, since $d_q\geq 1$. We can use the Weingarten bound in \Cref{lemma:Weingarten_bound} to get that $|\wg(\upsilon;d)|\leq C_Wd^{-4-|\upsilon|}$. Hence using this and that $d=d_cd_qd_A$, we get that
    \begin{align}
        |\wg(\upsilon;d)| \cdot |\Tr[\mathbf{R}P_\tau]|&\leq C_W\frac{d_c^2d_q^{4+|\upsilon|}d_A^{2+|\upsilon|}}{d^{4+|\upsilon|}}\\
        &=C_W\frac{1}{d_c^{2+|\upsilon|}d_A^2}\\
        &\leq C_W\frac{1}{d_c^2d_A^2}=C_W\frac{d_q^2}{d^2}.
    \end{align} 
    We can then sum over all possible choices of $\pi$ and $\tau$. Only $9$ such choices of $\pi$ are nonzero, and there are $24$ choices of $\tau$, so summing over all of these, which is a constant amount, we get that
    \begin{equation}
        \mathbb{E}_{\bfU}[\|\mathcal{I}_{\bfU}(X)\|_2^4]\leq O\left(\frac{d_q^2}{d^2}\|X\|_2^4\right).
    \end{equation}
    The explicit constant in the statement of the lemma can be taken as $C=9\cdot24\cdot C_W$.
\end{proof}
\section{BOW Estimator Analysis}\label{app:BOW_analysis}
We start by restating the theorem we want to prove:
\begin{theorem}
    Given $m$ copies of a CQ state $\rho$, and full knowledge of a CQ state $\sigma$ with marginal distributions $p,q$ respectively, there is an estimator $\hat{\Gamma}$ such that
    \begin{equation}
        \mathbb{E}[\hat{\Gamma}]=\Gamma=\|\rho - \sigma\|_2^2,
    \end{equation}
    and 
    \begin{equation}
        \mathbb{V}(\hat\Gamma)\leq O\left(\frac{\|p\|_\infty\Gamma}{m}+\frac{\|p\|_2^2+\|q\|_2^2}{m^2}\right).
    \end{equation}
\end{theorem}
Before going into the proof, it is worth overviewing the approach of B\u{a}descu, O'Donnell, and Wright \cite{buadescu2019quantum}, since we will essentially mimic their arguments. To do this, we will need a few definitions and results on quantum probability, which we give now. These are all taken from \cite{buadescu2019quantum}.

\begin{definition}[Observable]
    For a vector space $V$, an operator $X$, which is an endomorphism on $V$ (i.e., a homomorphism from $V$ to $V$), is called an observable.
\end{definition}

\begin{definition}[Moments]
    Let $\rho$ be a quantum state on Hilbert space $V$. Then an observable $X$ has expectation
    \begin{equation}
        \mathbb{E}_\rho[X]=\Tr(\rho X).
    \end{equation}
    For two such observables $X$ and $Y$, the covariance is defined as
    \begin{equation}
        \Cov_\rho[X,Y]=\mathbb{E}_\rho[(X-\mu_XI)^\dagger(Y-\mu_YI)],
    \end{equation}
    where $\mu_X=\mathbb{E}_\rho[X]$ and $\mu_Y=\mathbb{E}_\rho[Y]$. This also gives that $\Cov_\rho[X,Y]=\mathbb{E}_\rho[X^\dagger Y]-\mu_X^*\mu_Y$. The variance is defined by
    \begin{equation}
        \mathbb{V}_\rho[X]=\Cov_\rho[X,X].
    \end{equation}
    In particular, when $X$ is Hermitian, this recovers the recognizable result $\mathbb{V}_\rho[X]=\mathbb{E}_\rho[X^2] - \mathbb{E}_\rho[X]^2$. When clear, we will also omit the subscripts. We will also use operator and observable interchangeably.
\end{definition}
We can then show the following properties.
\begin{lemma}\label{lemma:operator_expectation}
    Let $\rho$ and $\sigma$ be quantum states, with an operator $A$ acting on the first system, and $B$ on the second system. Then, $\mathbb{E}_{\rho\otimes\sigma}[A\otimes B]=\mathbb{E}_\rho[A]\mathbb{E}_\sigma[B]$.
\end{lemma}
\begin{proof}
    Just from the definition, we have
    \begin{equation}
        \Tr[(\rho\otimes\sigma)(A\otimes B)]=\Tr[\rho A]\Tr[\sigma B],
    \end{equation}
    as required.
\end{proof}
This also gives the following covariance result:
\begin{corollary}\label{cor:operator_covariance}
    For observables $X,Y$ and two states $\rho,\sigma$, we have $\Cov_{\rho\otimes\sigma}[X \otimes I,I \otimes Y] = 0$.
\end{corollary}

Let us recall some additional notation from \cite{buadescu2019quantum}. Given $m$ copies of $\rho$ and $m$ copies of $\sigma$, associated naturally with registers $\{1, \dots, 2m\}$, they define observables
\begin{equation}
    \bigo_{(\rho \rho)} \triangleq \frac{2}{m(m-1)} \sum_{\substack{i < j, \\ 1 \leq i,j \leq m}} \swap_{i,j}, \quad \bigo_{(\rho \sigma)} \triangleq \frac{1}{m^2} \sum_{\substack{1 \leq i \leq m, \\ m+1 \leq j \leq 2m}} \swap_{i,j}, \quad \bigo_{(\sigma \sigma)} \triangleq \frac{2}{m(m-1)} \sum_{\substack{i < j, \\ m+1 \leq i,j \leq 2m}} \swap_{i,j};
\end{equation}

these observables are natural unbiased estimators for the quantities $\Tr(\rho^2), \Tr( \rho \sigma), \Tr(\sigma^2)$ respectively. Here, by unbiased estimators, we mean that their means with respect to $\rho^{\otimes m} \otimes \sigma^{\otimes m}$ yield the desired quantities.

Moreover, these observables are \emph{efficient}, in the sense that they have the lowest variance among all such observables~\cite[Corollary 4.24]{buadescu2019quantum}. 

As such, one can naturally estimate $\|\rho - \sigma\|_2^2 = \Tr(\rho^2) - 2\Tr(\rho \sigma) + \Tr(\sigma^2)$ via the observable $\bigo_{(\rho \rho)} - 2\bigo_{(\rho \sigma)} + \bigo_{(\sigma \sigma)}$, which is an unbiased and efficient estimator for this quantity~\cite[Corollary 4.25]{buadescu2019quantum}. Then, \cite{buadescu2019quantum} achieve \Cref{lem:hscertify} by bounding the variance of this observable. 

To produce our estimator for the squared Hilbert--Schmidt distance between two CQ states, we will consider a modification of the estimator above, along with a mild simplification: as we possess the description of one of the CQ states $\sigma$, we need only estimate $\Tr(\rho^2) - 2 \Tr(\rho \sigma)$ from $m$ copies of $\rho$, and we will thus only analyse the complexity of this task. Again, one can easily follow the logic above to construct an observable that acts as an efficient unbiased estimator for this quantity. Our task will then be to bound the variance of such an observable with respect to the state $\rho^{\otimes m}$, where we will exploit the fact that $\rho,\sigma$ are classical-quantum states.

First, let us introduce some additional notation and preliminary results. Throughout, for $i \in [m]$ and an operator $M$, $M_i$ will be used to denote the operator $M$ on the $i$th register; we may drop the subscript when the register is clear from context. In the context of multiple registers, $M_i$ alone indicates that we have the identity operator on the remaining registers. Similarly, we let $\Tr_i$ denote the partial trace over the $i$th register. We will require the following definition of a \emph{partial} expectation.

\begin{definition}[Partial Expectation]
    Assume we have two copies of a quantum state $\rho$, and we have an operator $X_{1,2}$ that acts on $\rho_1 \otimes \rho_2$. Then, define the partial trace over the second copy as follows:
    \begin{equation}
        \mathbb{E}_{2}[X_{1,2}]=\Tr_2[(I_1\otimes\rho_2)X_{1,2}],
    \end{equation}
    As one would expect, $\mathbb{E}_1[\mathbb{E}_2[X_{1,2}]] = \mathbb{E}_{\rho_1\otimes\rho_2}[X_{1,2}]$. 
\end{definition}

One can easily generalize this to $m$-register operators. However, we will only ever need to take the partial expectation of operators acting on two copies of a state, so we do not consider this extension. We will need to analyze the case where two identical observables operate on two different pairs of registers but share one common index. In such cases, the moment of their product is as follows.

\begin{lemma}\label{lem:partial_operator_variance}
    Assume we have $m$ copies of a state $\rho$. Let $H_{1,2}$ and $H_{1,3}$ acting on registers $(1,2)$ and $(1,3)$ respectively. Denote their mean $\mu=\mathbb{E}_{\rho \otimes \rho}[H_{a,b}]$, which is the same for all possible $a \neq b$. Denote $\widetilde{H}_{a,b} = H_{a,b}-\mu I$. Let $G = \mathbb{E}_i[H_{1,i}]$, and $\widetilde{G}=G-\mu I$. Then, 
    \begin{equation}
        \mathbb{E}_{\rho^{\otimes 3}}[\widetilde{H}_{1,2}\widetilde{H}_{1,3}]=\mathbb{E}_\rho[\widetilde{G}^2]  =\mathbb{V}_\rho[G].  
    \end{equation}
\end{lemma}

\begin{proof}
    First, note we choose the indices $(1,2)$ and $(1,3)$ without loss of generality, since all copies of $\rho$ are the same. The result holds broadly for any pair of indices with one common index. Notice that taking the partial expectation of $\widetilde{H}_{1,i}$ with respect to index $i$, we obtain $\widetilde{G}_1$. Thus, we can compute the expectation in the lemma as
    \begin{align}
        \mathbb{E}_{\rho^{\otimes 3}}[\widetilde{H}_{1,2}\widetilde{H}_{1,3}]&=\mathbb{E}_1[\mathbb{E}_2[\mathbb{E}_3[\widetilde{H}_{1,2}\widetilde{H}_{1,3}]]]\\
        &=\mathbb{E}_1[\mathbb{E}_2[\widetilde{H}_{1,2}]\mathbb{E}_3[\widetilde{H}_{1,3}]]\\
        &=\mathbb{E}_1[\widetilde{G}_1^2]\\
        &=\mathbb{E}_\rho[\widetilde{G}_1^2],
    \end{align}
    where the second equality used that $\widetilde{H}_{1,2}$ is identity on the third register. To get the second part of the lemma, notice that $\mathbb{E}_\rho[G_1]=\mu$. Then,
    \begin{align}
    \mathbb{E}_\rho[\widetilde{G}_1^2]&=\mathbb{E}_\rho[(G_1-\mu I)^2]\\
        &=\mathbb{E}_\rho[G_1^2]-2\mu \cdot \mathbb{E}_\rho[G_1]+\mu^2 \mathbb{E}_\rho[I]\\
        &=\mathbb{E}_\rho[G_1^2]-\mathbb{E}_\rho[G_1]^2\\
        &=\mathbb{V}_\rho[G_1].
    \end{align}
\end{proof}

Let us now explicitly state the estimator we will design. First, we aim to estimate $\Tr(\rho^2)$, where recall we can write the CQ-state $\rho = \sum_{c \in [d_c]} \rho_c \otimes \ketbra{c}{c}$. Note that 
\begin{equation}
\Tr[\rho^2] = \sum_c\Tr[\rho^2_c].
\end{equation}

Let us first discuss how to estimate a single term $\Tr(\rho_c^2)$ in this latter sum, given two copies of $\rho$. Then, observe that we can write
\begin{equation}
    \Tr(\rho_c^2) = \Tr(\swap_{Q_1,Q_2} \otimes \ketbra{c}{c}_{C_1} \otimes \ketbra{c}{c}_{C_2} \cdot \rho^{\otimes 2}) = \mathbb{E}[\swap_{Q_1,Q_2} \otimes \ketbra{c}{c}_{C_1} \otimes \ketbra{c}{c}_{C_2}].
\end{equation}

For any pair $i \neq j \in [m]$, we denote $M_{i,j}(c) \triangleq \swap_{Q_i,Q_j} \otimes \ketbra{c}{c}_{C_i} \otimes \ketbra{c}{c}_{C_j}$, such that
\begin{equation}
    \mathbb{E}_{\rho^{\otimes m}} [M_{i,j}(c)] = \Tr(\rho_c^2).
\end{equation}
We then define the observable
\begin{equation}
    S_{ij} \triangleq \sum_{c=1}^{d_c} M_{i,j}(c),
\end{equation}
which satisfies 
\begin{equation}
    \mathbb{E}_{\rho^{\otimes m}} = \sum_c \Tr(\rho_c^2) = \Tr(\rho^2),
\end{equation}
yielding an unbiased estimator. Our current observable only acts on two copies of $\rho$; to make full use of the $m$ copies at our disposal and further reduce the variance, we will average this observable over all such pairs. Then, the observable we consider for estimation $\Tr(\rho^2)$ is
\begin{equation}
    \frac{1}{{\binom{m}{2}}}\sum_{i<j}S_{ij}.
\end{equation}

It remains now to provide an estimator for $\Tr(\rho \sigma)$. For a single copy, a natural observable here is just $\sigma$ itself. Given that we have $m$ copies, we again compute the average, and use the observable
\begin{equation}
    \frac1m \sum_{i = 1}^m \sigma_i.
\end{equation}

Lastly, to write the observable as an unbiased estimator for $\|\rho-\sigma\|_2^2$, we also need an observable acting on $\rho^{\otimes m}$ whose mean is $\Tr(\sigma^2)$; however, this is a $\rho$-independent scalar. Consequently, the correct observable is just $\Tr(\sigma^2) \cdot I^{\otimes m}$. 

Hence, to estimate $\|\rho - \sigma\|_2^2 = \Tr(\rho^2) - 2 \Tr(\rho \sigma)$, we define the observable
\begin{equation}
    \hat\Gamma=\frac{1}{\binom{m}{2}}\sum_{i<j}S_{ij}-\frac{2}{m}\sum_{i=1}^m \sigma_i+\Tr[\sigma^2]I^{\otimes m}
\end{equation}
We claim this is an unbiased estimator for $\|\rho-\sigma\|_2^2$ with variance as claimed in \Cref{thm:cq_hs_estimator}. The unbiasedness is by design, but we prove it in the lemma below for completeness. 
\begin{lemma}\label{lemma:cq_hs_estimator_expectation}
    The above estimator satisfies $\mathbb{E}_{\rho^{\otimes m}}[\hat\Gamma]= \|\rho - \sigma\|_2^2 = \Gamma$.
\end{lemma}

\begin{proof}
    We can individually calculate the expectations of each operator. First, let us compute the expectation of $S_{1,2}$, which is equivalent to the expectation of every other pair. This gives us that
    \begin{equation}
        \mathbb{E}[S_{1,2}]=\sum_c\Tr[(\rho_c\otimes\rho_c)\swap_{Q_1,Q_2}] = \sum_c \Tr(\rho_c^2) = \Tr(\rho^2),
    \end{equation}
    as $\Tr[(A\otimes B)\swap]=\Tr[AB]$ for any operators $A,B$ of appropriate dimensions. 
    
    Next, for each $\sigma_i$, the expectation is simply given by
    \begin{equation}
        \mathbb{E}[\sigma_i]=\Tr[\sigma_i\rho^{\otimes m}]  = \Tr(\rho\sigma).
    \end{equation}
    Finally, the expectation of the last operator $\Tr(\sigma^2) \cdot I $ is a constant, i.e.,  $\Tr[\sigma^2]$. Then, putting all three together, the expectation of the CQ estimator is
    \begin{align}
        \mathbb{E}[\hat\Gamma]&=\frac{1}{\binom{m}{2}}\sum_{i<j}\Tr[\rho^2]+\frac{2}{m}\sum_{i = 1}^m\Tr(\rho\sigma)+\Tr[\sigma^2]\\
        &= \|\rho - \sigma\|_2^2,
    \end{align}
    as claimed.
\end{proof}

We next prove that this estimator has the required variance. To do this, we will rewrite our operator notation. Let
\begin{equation}
    H_{i,j}=S_{i,j}-\sigma_i-\sigma_j+\Tr[\sigma^2]I.
\end{equation}
It is then easy to check that $\hat\Gamma=\frac{1}{\binom{m}{2}}\sum_{i<j}H_{i,j}$. We now prove the following result on the variance.

\begin{lemma}\label{lemma:cq_hs_estimator_variance}
    Let $\mu=\mathbb{E}[H_{i,j}]$, and $\widetilde{H}_{i,j}=H_{i,j}-\mu I$. Let $G_i=\mathbb{E}_j[H_{i,j}]$, the expectation over just the second term in $H_{i,j}$. Define $A_1=\mathbb{V}_\rho(G_i)$ and $A_2=\mathbb{E}_{\rho\otimes\rho}[H_{i,j}^2]$ and recall that these quantities are identical for all choices of $(i,j)$. Then,
    \begin{equation}
        \mathbb{V}(\hat\Gamma)\leq O\left(\frac{A_1}{m}+\frac{A_2}{m^2}\right).
    \end{equation}
\end{lemma}
\begin{proof}
    Since shifts by the mean do not change variance, we have that
    \begin{equation}
        \mathbb{V}(\hat\Gamma)=\mathbb{V}\left(\frac{1}{\binom{m}{2}}\sum_{i<j}H_{i,j}\right)=\mathbb{V}\left(\frac{1}{\binom{m}{2}}\sum_{i<j}\widetilde{H}_{i,j}\right).
    \end{equation}
    Since $\mathbb{E}[\widetilde{H}_{i,j}]=0$, this gives that
    \begin{equation}
        \mathbb{V}(\hat\Gamma)=\mathbb{E}\left[\left(\frac{1}{\binom{m}{2}}\sum_{i<j}\widetilde{H}_{i,j}\right)^2\right]=\frac{1}{\binom{m}{2}^2}\sum_{i<j}\sum_{k<l}\mathbb{E}[\widetilde{H}_{i,j}\widetilde{H}_{k,l}].
    \end{equation}
    We now consider three cases, depending on the relationship between the pairs of tuples $(i,j)$ and $(k,l)$.

    \paragraph{Case 1: Disjoint tuples.} It the tuples are completely disjoint, then by \Cref{cor:operator_covariance}, we have that $\mathbb{E}[\widetilde{H}_{i,j}\widetilde{H}_{k,l}]=\mathbb{E}[\widetilde{H}_{i,j}]\mathbb{E}[\widetilde{H}_{k,l}]=0$.

    \paragraph{Case 2: Identical tuples.} If $(i,j)=(k,l)$, then the contribution is given by $\mathbb{E}[\widetilde{H}_{i,j}^2]=\mathbb{V}(H_{i,j})\leq \mathbb{E}[H_{i,j}^2]=A_2$. Since we have $\binom{m}{2}$ such terms, the total contribution to the variance is at most
    \begin{equation}
        \frac{A_2}{\binom{m}{2}}\leq \bigo\left(\frac{A_2}{m^2}\right).
    \end{equation}

    \paragraph{Case 3: Pairs that overlap once.} Assume now the two pairs overlap on one index. By symmetry, it is sufficient to consider the contribution to the variance by pairs $H_{1,2}$ and $H_{1,3}$, and then sum over all pairs. Then, by \Cref{lem:partial_operator_variance}, we have that $\mathbb{E}[\widetilde{H}_{1,2}\widetilde{H}_{1,3}]= \mathbb{V}[G_1] = A_1$. For each fixed tuple $(i,j)$, exactly $2(m-2)$ other tuples share exactly one index with it, with $m-2$ coinciding on index $i$, and the other $m-2$ coinciding on index $j$. Hence, by summing over all pairs, the total contribution to the variance in this case is at most
    \begin{align}
        \frac{2(m-2)}{\binom{m}{2}}A_1&\leq O\left(\frac{A_1}{m}\right).
    \end{align}
    Hence, summing all three cases, we have that the variance of the estimator is at most
    \begin{equation}
        \mathbb{V}(\hat\Gamma)\leq O\left(\frac{A_1}{m}+\frac{A_2}{m^2}\right).
    \end{equation}
\end{proof}

To prove \Cref{thm:cq_hs_estimator}, it remains to bound the quantities $A_1$ and $A_2$. We begin by bounding $A_1$.
\begin{lemma}\label{lemma:A_1_bound}
    $A_1\leq \|p\|_\infty\Gamma$.
\end{lemma}
\begin{proof}
    We can look at $H_{1,2}$ wlog. Recall that $G_1$ is its partial expectation over the second index, and hence we can compute the partial expectation over each individual term in $H_{1,2}$. We start with $\mathbb{E}_2[S_{1,2}]$:
    \begin{align}
        \mathbb{E}_2[S_{1,2}]&=\Tr_2[(I\otimes\rho)S_{1,2}]\\
        &=\Tr_2\left[\sum_c(I\otimes\rho_c)\swap_{Q_1,Q_2}\otimes|c\rangle\langle c|_{C_1}\otimes|c\rangle\langle c|_{C_2}\right]\\
        &= \sum_c\rho_C\otimes\otimes|c\rangle\langle c|_{C_1}\\
        &=\rho,
    \end{align}
    where we used the identity $\Tr_2((I \otimes A) \swap) = A$ for any operator $A$.
    
    Second, it is clear that $\mathbb{E}_2[\sigma_1]=\sigma_1$, since it is independent of the second register.
    On the other hand, we have that
    \begin{align}
        \mathbb{E}_2[\sigma_2]&=\Tr_2[(I\otimes\rho)(I\otimes \sigma)]\\
        &= I_1 \cdot \Tr(\rho\sigma),
    \end{align}
    which is a scalar multiple of the identity. 
    Finally, we note $\Tr[\sigma^2]I$ also becomes a scalar multiple of the identity after tracing out the second copy. Hence, we have
    \begin{equation}
        G_1 = \Tr_2(S_{1,2} - \sigma_1 - \sigma_2 + Tr(\sigma^2) \cdot I^{\otimes 2}) = \rho - \sigma + (\Tr(\sigma^2)-\Tr(\rho\sigma)) \cdot I.
    \end{equation}
    When computing the variance, we can ignore the scalar multiple of identity above, as such scalar shifts do not impact the variance.

    Let us now rewrite $\rho - \sigma$ in the classical-quantum basis. For each $c \in [d_c],$ let $\Delta_c \triangleq \rho_c - \sigma_c$, and let the projector onto a classical string $c$ be $\Pi_c=\ketbra{c}{c}$.
    Then, we have that
    \begin{align}
        \rho-\sigma &= \sum_c(\rho_c-\sigma_c)\otimes\Pi_c=\sum_c\Delta_c\otimes\Pi_c,\\
        (\rho-\sigma)^2&=\sum_c\Delta_c^2\otimes\Pi_c,\\
        \mathbb{V}_\rho(\rho-\sigma)&\leq \mathbb{E}_\rho[(\rho-\sigma)^2]\\
        &=\sum_c \Tr(\rho_c\Delta_c^2)\\
        &\leq \sum_c p_c\Tr[\Delta_c^2]\\
        &\leq \|p\|_\infty \sum_{c} \Tr(\Delta_c^2) = \|p\|_\infty\Gamma,
    \end{align}
    where we use that by definition $0\preceq \rho_c \preceq p_cI$. This then yields
    \begin{equation}
        A_1 = \mathbb{V}_\rho[G] = \mathbb{V}_\rho(\rho - \sigma) \leq \|p\|_\infty \Gamma,
    \end{equation}
    as claimed.
\end{proof}

We now bound $A_2$:

\begin{lemma}\label{lemma:A_2_bound}
    $A_2\leq \bigo(\|p\|_2^2+\|q\|_2^2)$.
\end{lemma}

\begin{proof}
    We start with an operator inequality. By the operator Jensen's inequality \cite[Theorem 2.1]{HANSEN_2003}, for Hermitian operators $B_i$, we have that
\begin{equation}
    (B_1+B_2+B_3+B_4)^2\preceq C \cdot \left(B_1^2+B_2^2+B_3^2+B_4^2\right),
\end{equation}
for some  $C > 0$. Recalling that $H_{1,2}$ is a sum of four operators, it suffices to bound the second moment of each of the operators. First, consider $S_{1,2}^2$, which simply applies the swap twice. Hence, 
\begin{equation}
    S_{1,2}^2=\sum_cI_{Q_1,Q_2}\otimes \ketbra{c}{c}_{C_1} \otimes \ketbra{c}{c}_{C_2}.
\end{equation}

Hence, we have that $\mathbb{E}[S_{1,2}^2]=\sum_c \Tr(\rho_c)^2=\|p\|_2^2$. 

Next, consider $\sigma^2=\sum_c\sigma_c^2\otimes\Pi_c$. Note that since $\sigma_c$ is positive semidefinite, and $\Tr[\sigma_c]=q_c$, every eigenvalue of $\sigma_c$ is at most $q_c$. Hence, we have that $\sigma_c^2\preceq q_c^2I$. Then, $\Tr(\rho_c\sigma_c^2)\leq p_cq_c^2\leq q_c^2$.
Summing over all possible $c$, this gives $\mathbb{E}[\sigma^2]\leq\sum_c q_c^2=\|q\|_2^2$, for $i = 1,2$. 

Lastly, by definition, we have that $\Tr[\sigma^2]I \preceq \sum_c q_c^2 \cdot I=\|q\|_2^2 \cdot I$. Since $\|q\|_2^2\leq 1$, the second moment of this operator is also less than $\|q\|_2^2$.

Combining our four estimates along with the operator Jensen's inequality gives $A_2\leq \bigo(\|p\|_2^2+\|q\|_2^2)$.
\end{proof}

We conclude this section by noting that \Cref{thm:cq_hs_estimator} is immediate from \Cref{lemma:cq_hs_estimator_expectation,lemma:cq_hs_estimator_variance,lemma:A_1_bound,lemma:A_2_bound}.

\section{Proof of \Cref{lem:instrument_minimax_maximin_bounds}}\label{app:lemma_instrument_minimax_maximin_bounds}
Here, we prove \Cref{lem:instrument_minimax_maximin_bounds}.
\begin{proof}
    We will treat the action of distributed nodes as a set of quantum instruments $\{\mathcal{I}_i\}_{i\in[m]}$. In the public-coin setting, each $\mathcal{I}_i$ can depend on some shared random string $\mathbf{r}$. Let $\mathcal{A}(\rho,\mathbf{r})$ be the output of the referee node. Let $m$ be large enough to succeed with probability at least $3/4$. Then, we have that
    \begin{align}
        \frac{1}{2}\Pr_{\rho\sim D}[\mathcal{A}(\rho,\mathbf{r})=\reject]+\frac{1}{2}\Pr_{\rho=I/d}[\mathcal{A}(\rho,\mathbf{r})=\accept]&\geq \frac{1}{2}\cdot\frac{3}{4}\cdot\Pr_{\rho\sim D}\left[\|\rho-I/d\|_1\geq\epsilon\right]+\frac{1}{2}\cdot\frac{3}{4}\\
        &\geq \frac{9}{16}.
    \end{align}
    Then, for any $D$, there exists some random string $\mathbf{r}=r$ such that
    \begin{equation}
        \frac{1}{2}\Pr_{\rho\sim D}[\mathcal{A}(\rho,\mathbf{r})=\reject|\mathbf{r}=r]+\frac{1}{2}\Pr_{\rho=I/d}[\mathcal{A}(\rho,\mathbf{r})=\accept|\mathbf{r}=r]\geq \frac{9}{16}.
    \end{equation}
    Thus, for each $D$, there is a deterministic choice of quantum instruments $\mathcal{I}_1^{(r)},...,\mathcal{I}_m^{(r)}$ such that the referee node can distinguish between the equally likely cases of $\rho\sim D$ and $\rho=I/d$ with probability at least 9/16. Now, defining
    \begin{equation}
        P_r=\mathbb{E}_{\rho\sim D}\left[\bigotimes_{i=1}^m\mathcal{I}_i^{(r)}(\rho)\right],~Q_r=\bigotimes_{i=1}^m\mathcal{I}_i^{(r)}\left(\mmstate\right),
    \end{equation}
    we can apply the Helstrom bound; i.e., the optimal success probability for distinguishing equally likely $P_r$ and $Q_r$ is given by 
    \begin{equation}
        \frac{1}{2}+\frac{1}{4}\|P_r-Q_r\|_1\geq \frac{9}{16}.
    \end{equation}
    Using \Cref{lem:1_norm_qchi_inequality}, we finally see that for each $D$, there is a deterministic choice of instruments that yields the desired public-coin bound. 
    
    Moving to the private-coin model, each distributed node has its own internal randomness, $r_1,...,r_m$. Then, fixing the random strings each distributed node gets, this determines fixed quantum instruments $\mathcal{I}_1^{r_1},...,\mathcal{I}_m^{r_m}$. Then, we can take an expectation over all random strings, and model the action of each distributed node as a deterministic instrument $\mathcal{I}_i=\mathbb{E}_{r_i}[\mathcal{I}_i^{r_i}]$. This is a valid quantum instrument, since it is an average of quantum instruments, and hence is completely positive and trace preserving while preserving the classical-quantum decomposition. Then, we can define for these instruments
    \begin{equation}
        P=\mathbb{E}_{\rho\sim D}\left[\bigotimes_{i=1}^m\mathcal{I}_i(\rho)\right],~Q=\bigotimes_{i=1}^m\mathcal{I}_i\left(\mmstate\right).
    \end{equation}
    Again, by the Helstrom bound, we must have that for any $D$,
    \begin{equation}
        \|P-Q\|_1\geq \frac{1}{4}.
    \end{equation}
    Since this holds true for any $D$, it holds true for the minimum:
    \begin{equation}
        \min_{D\in\mathcal{D}_\epsilon\left(\mmstate\right)}\|P-Q\|_1\geq \frac{1}{4}.
    \end{equation}
    Finally, we maximize over all possible instruments, giving that
    \begin{equation}
        \max_{\mathcal{I}_1,...,\mathcal{I}_m}\min_{D\in\mathcal{D}_\epsilon\left(\mmstate\right)}\QChi{\mathbb{E}_{\rho\sim D}\left[\bigotimes_{i=1}^m\mathcal{I}_i(\rho)\right]}{\bigotimes_{i=1}^m\mathcal{I}_i\left(\mmstate\right)}\geq\frac{1}{16}.
    \end{equation}
    Using \Cref{lem:1_norm_qchi_inequality} again gives the result for private coins.
\end{proof}
\section{General Quantum Ingster-Suslina Lemma}\label{app:symmetrized_quantum_ingster_suslina}
We prove in this section \Cref{lem:quantum_ingster-suslina_gen} by following the work of O'Donnell and Wadhwa~\cite{odonnell2025instanceoptimalquantumstatecertification}, but extending the analysis to generic $\alpha$-$\chi^2$-divergences.
\begin{proof}
    Let $R_\theta=\rho_\theta^{(m)}$, $S=\sigma^{(m)}$, $P=\mathbb{E}_{\bftheta}[R_{\bftheta}]$ for convenience. Then, notice by the definition of the $\chi^2$ divergence, we have that
    \begin{align}
        \QChiAlpha{\alpha}{\mathbb{E}_{\bftheta}[\rho^{( m)}_{\bftheta}]}{\sigma^{(m)}} + 1 &=\QChiAlpha{\alpha}{P}{S} + 1 = \Tr(S^{-\alpha}P S^{\alpha-1}P)\\
        &= \Tr(S^{-\alpha}E_{\bftheta
        }[R_{\bftheta}] S^{\alpha-1}E_{\bftheta'}[R_{\bftheta'}])\\
        &=\mathbb{E}_{\bftheta,\bftheta'} \Tr(S^{-\alpha}R_{\bftheta} S^{1-\alpha}R_{\bftheta'})\\
        &=\mathbb{E}_{\bftheta,\bftheta'}\left[ \Tr\left(\bigotimes_{i=1}^m\sigma_i^{-\alpha}\rho_{i,\bftheta}\sigma_{i}^{\alpha-1}\rho_{i,\bftheta'}\right)\right]\\
        &=\mathbb{E}_{\bftheta,\bftheta'}\left[\prod_{i=1}^m \Tr(\sigma_i^{-\alpha}\rho_{i,\bftheta}\sigma_{i}^{\alpha-1}\rho_{i,\bftheta'})\right],
    \end{align}
    Now, let $\Delta_{i,\theta} \triangleq \rho_{i,\theta}-\sigma_i$ and notice that $\Tr(\Delta_{i,\theta}) = \Tr(\rho_{i,\theta}) - \Tr(\sigma_i) = 0$. We can now rewrite the trace above in terms of $\Delta$s and $\sigma$s:
    \begin{align}
        \Tr(\sigma_i^{-\alpha}\rho_{i,\theta}\sigma_{i}^{\alpha-1}\rho_{i,\theta'}) &= \Tr(\sigma_i^{-\alpha}(\sigma_i+\Delta_{i,\theta})\sigma_i^{\alpha-1}(\sigma_i+\Delta_{i,\theta'}))
        \\&=\Tr(\sigma_i^{-\alpha}\sigma_i\sigma_i^{\alpha-1}\sigma_i)+\Tr(\sigma_i^{-\alpha}\sigma_i\sigma_i^{\alpha-1}\Delta_{i,\theta'})\nonumber
        \\&+\Tr(\sigma_i^{-\alpha}\Delta_{i,\theta}\sigma_i^{\alpha-1}\sigma_i)+\Tr(\sigma_i^{-\alpha}\Delta_{i,\theta}\sigma_i^{\alpha-1}\Delta_{i,\theta'}),
    \end{align}
    where in the last line we expand and use linearity of traces. Notice now that we have, considering each of the four terms,
    \begin{align}
        &\Tr(\sigma_i^{-\alpha}\sigma_i\sigma_i^{\alpha-1}\sigma_i)=\Tr(\sigma_i)=1,\\
        &\Tr(\sigma_i^{-\alpha}\sigma_i\sigma_i^{\alpha-1}\Delta_{i,\theta'})=\Tr(\Delta_{i,\theta'})=0,\\
        &\Tr(\sigma_i^{-\alpha}\Delta_{i,\theta}\sigma_i^{\alpha-1}\sigma_i)=\Tr(\Delta_{i,\theta}\sigma_i^{-\alpha}\sigma_i\sigma_i^{\alpha-1})=\Tr(\Delta_{i,\theta})=0,\\
        &\Tr(\sigma_i^{-\alpha}\Delta_{i,\theta}\sigma_i^{\alpha-1}\Delta_{i,\theta'})=Z_i(\theta,\theta').
    \end{align}
    Hence, we have that
    \begin{align}
        \Tr(\sigma_i^{-\alpha}\rho_{i,\theta}\sigma_{i}^{\alpha-1}\rho_{i,\theta'})&=1+Z_i(\theta,\theta'),
    \end{align}
    and thus
    \begin{align}
        1+\QChiAlpha{\alpha}{P}{S}&=\mathbb{E}_{\bftheta,\bftheta'}\left[\prod_{i=1}^m(1+Z_i(\bftheta,\bftheta'))\right].
    \end{align}
    We can get the final inequality using that $1+x\leq e^x$ when $x\geq 0$.
\end{proof}

\section{Centralized mixedness testing lower bound with the \cite{liu2024quantum} instance}
\label{s:centralized-mixedness-testing-bound}

In this section, we use \Cref{lem:quantum_ingster-suslina_gen} for $\alpha = 0$ and the hard instance from \Cref{def:hard-instance-mic} to provide an alternate proof of the $\Omega(d/\eps^2)$ copy complexity lower bound for mixedness testing \cite{o2015quantum,odonnell2025instanceoptimalquantumstatecertification} in the fully unrestricted setting.

\begin{theorem}
    Let $d \geq 6, 0 < \eps < C$. For the ensemble given in \Cref{def:hard-instance-mic},
    \begin{equation}
        \qdchi\left(\mathbb{E}_{\bfz}[\rho_{\bfz}^{\otimes n}] \big\| \left(\mathbbm{1}/d\right)^{\otimes n}\right) \leq \exp\left(\frac{n^2c^4\eps^4}{2\ell^2}\right) - 1 + \frac{4}{e^d}.
    \end{equation}
    Thus, the quantum $\chi^2$-divergence is $< \frac{1}{16}$ unless $n = \Omega(\sqrt{\ell}/\eps^2)$. Setting $\ell = d^2-1$ recovers the $\Omega(d/\eps^2)$ lower bound for mixedness testing.
\end{theorem}

\begin{proof}
    Keeping \Cref{lem:quantum_ingster-suslina_gen} in mind, with $\alpha = 0$, we wish to bound the moment generating function of $Z(\bfz,\bfz^\prime) \triangleq d\cdot\tr(\bar{\Delta}_{\bfz} \bar{\Delta}_{\bfz^\prime})$. Recall that $N_z \triangleq \min\{1, \frac{1}{d\|\Delta_z\|_\infty}\}$. Then,
    \begin{align}
        Z(z,z^\prime) &= d\cdot\tr(\bar{\Delta}_z \bar{\Delta}_{z^\prime}) 
        \\&= d \cdot \frac{c^2\eps^2}{d\ell} \sum_{i,j \in [\ell]} \tr(z_iz^\prime_j V_i V_j) N_z N_{z^\prime}
        \\&= \frac{c^2\eps^2}{\ell}N_zN_{z^\prime} \sum_{i \in [\ell]} z_i z_i^\prime 
        \\&= \frac{c^2\eps^2}{\ell}N_zN_{z^\prime} z^\top z^\prime.
    \end{align}
    Now, by \Cref{lem:quantum_ingster-suslina_gen}, we have
    \begin{align}
       \qdchi\left(\mathbb{E}_{\bfz}[\rho_{\bfz}^{\otimes n}] \big\| \left(\mathbbm{1}/d\right)^{\otimes n}\right) + 1 &\leq \mathbb{E}_{\bfz,\bfz^\prime}\left[\exp(n Z(\bfz,\bfz^\prime))\right]
       \\&= \mathbb{E}_{\bfz,\bfz^\prime} \left[\exp\left(\frac{nc^2\eps^2}{\ell} N_{\bfz} N_{{\bfz}^\prime} \cdot \bfz^\top \bfz^\prime\right)\right]
       \\&\leq \mathbb{E}_{\bfz,\bfz^\prime} \left[\exp\left(\frac{nc^2\eps^2}{\ell} \cdot \bfz^\top \bfz^\prime\right)\right] + \frac{4}{e^d}
       \\&= \prod_{i = 1}^{\ell}\left(\mathbb{E}_{\bfz_i,\bfz_i^\prime} \left[
            \exp\left(\frac{nc^2\eps^2}{\ell} \cdot \bfz_i \bfz_i^\prime\right)
       \right]\right) + \frac{4}{e^d}
       \\&= \left(\mathbb{E}_{\boldsymbol{b} \sim \{-1,+1\}} \left[
            \exp\left(\frac{nc^2\eps^2}{\ell} \cdot \boldsymbol{b}\right)
       \right]\right)^{\ell} + \frac{4}{e^d}
       \\&\leq \left(\exp\left(\frac{n^2c^4\eps^4}{2\ell^2}\right)\right)^{\ell} + \frac{4}{e^d}
       \\&= \exp\left(\frac{n^2c^4\eps^4}{2\ell}\right) + \frac{4}{e^d}.
    \end{align}
    The second inequality uses the fact that $N_{\bfz} = 1$ except with exponentially small probability, as shown in \cite[Lemma B.8]{liu2024quantum}. The next step uses the fact that the entries of $\bfz$ are drawn independently, and the subsequent step treats each product $\bfz_i\bfz_i^\prime$ as a uniformly random sign. In the last inequality, we used $\frac{e^x + e^{-x}}{2} \leq e^{x^2/2}$.
\end{proof}

\section{Exponential separations from shared entanglement}\label{s:entanglement}

In this section, we demonstrate that for certain testing problems for $n$-qubit states, two of our distributed inference settings, namely the settings of unbounded classical communication with shared randomness and bounded classical communication with shared \emph{entanglement}, are exponentially separated. Following \Cref{def:distributed-model}, these are the $(n_c > 0, n_q = 0, R = \mathsf{public}, E = 0)$ and $(n_c = 2n, n_q = 0, R = \mathsf{private}, E = n)$ settings respectively. First, note that for unbounded $n_c$, the former setting is equivalent to the centralized setting where testers can perform non-adaptive single-copy measurements. In this centralized setting, problems such as purity testing and unsigned Pauli shadow tomography are known to require exponentially many copies \cite{chen2022exponential}. 

However, given access to just $2$-copy measurements, one can solve the above problems with constant complexity by performing \emph{Bell sampling}. The key insight that leads to our exponential separation is the observation that one can also perform Bell sampling in the $(n_c = 2n, n_q = 0, R = \mathsf{private}, E = n)$ setting.

Let us first define Bell sampling. For any string $x = (a,b) \in \{0,1\}^{2n}$, we define the Weyl operators
\begin{equation}
    W_x \triangleq i^{a.b} (X^{a_1}Z^{b_1}) \otimes \dots \otimes (X^{a_n}Z^{b_n}),
\end{equation}
where $X,Z$ are the single-qubit Pauli matrices. It is not hard to check that the Weyl operators are self-adjoint, i.e. $W_x = W_x^\dagger$.
We now define the Bell state associated with each $x \in \{0,1\}^{2n}$.
\begin{equation}
    \ket{\psi_x} \triangleq (W_x \otimes I) \ket{\epr_n} = \frac{i^{a.b}}{\sqrt{2^n}} \sum_{k \in \{0,1\}^n} (-1)^{k.b} \ket{k\oplus a, k}.
\end{equation}
We note that the $4^n$ Bell states form an orthonormal basis for $\mathbb{C}^{4^n}$, called the \emph{Bell basis}. Measuring two copies of a quantum state $\rho$ in the Bell basis results in a \emph{Bell sample} $x \in \{0,1\}^{2n}$, with probability 
\begin{equation}
    p_\rho(x) \defeq \bra{\psi_x} \rho \otimes \rho \ket{\psi_x}.
\end{equation}

Bell sampling is a standard procedure in quantum learning and testing, and indeed leads to efficient algorithms for unsigned Pauli shadow tomography and purity testing. Let us now formally state our main observation:

\begin{proposition}
\label{obs:distributed-bell-sampling}
    Given two distributed nodes $N_1, N_2$ in the $(n_c = 2n, n_q = 0, R = \mathsf{private}, E = n)$ setting, where each distributed node receives a copy of an $n$-qubit mixed state $\rho$, there exists an algorithm allowing the central node $N_c$ to exactly sample from the distribution $p_\rho$ with probability $1$.
\end{proposition}

Before proving the above proposition, note that this immediately implies efficient algorithms for purity testing and unsigned shadow tomography in our distributed setting, by pairing up adjacent distributed nodes to generate multiple Bell samples. Consequently, we obtain exponential separations between the two settings mentioned at the start of this section. For simplicity, we only state such a separation formally for purity testing, but note that it easily extends to any task which can be solved efficiently via Bell sampling and is hard in the setting of non-adaptive single-copy measurements. By the observation above and \cite[Theorem 5.11]{chen2022exponential}, we have the following separation:

\begin{corollary}
    For any $n$-qubit quantum state $\rho$, in the $(n_c > 0, n_q = 0, R = \mathsf{public}, E = 0)$-setting, the number of distributed nodes needed to test whether $\rho$ is a pure state or the maximally mixed state is at least $\Omega(2^{n/2})$. However, in the $(n_c = 2n, n_q = 0, R = \mathsf{private}, E = n)$-setting, only $\bigo(1)$ distributed nodes suffice to perform such a test.
\end{corollary}

It remains now to prove \Cref{obs:distributed-bell-sampling}.

\begin{proof}[Proof of \Cref{obs:distributed-bell-sampling}]
    Our distributed Bell sampling routine is straightforward and inspired by quantum teleportation. However, despite sharing sufficiently many EPR pairs, the distributed nodes cannot quite perform teleportation; as we do not allow any communication (even classical) between $N_1$ and $N_2$, they cannot communicate the post-measurement corrections necessary for teleportation. Instead, they will send these corrections to the central node $N_c$.

    Let us assume that each node places its copy of $\rho$ in register $A_i$ and its half of the EPR pairs in register $B_i$. Then, the joint state of the two nodes is given by:
    \begin{equation}
        \sigma = \rho_{A_1} \otimes \rho_{A_2} \otimes \ket{\epr_n}\bra{\epr_n}_{B_1,B_2},
    \end{equation}
    where $\ket{\epr_n}$ denotes $n$ copies of the maximally entangled state.

    Now, for $i \in \{1,2\}$, each node $N_i$ measures its registers $A_i,B_i$ in the Bell basis, and communicates its output string $z_i \in \{0,1\}^{2n}$ to $N_c$. $N_c$ will then simply output $z = z_1 \oplus z_2$, and we will show that this is distributed according to $p_\rho(z)$.

    Without loss of generality, we can assume $N_1$ performs its Bell basis measurements before $N_2$. When discussing $N_1$'s measurements, we will only be concerned with the registers $A_1, B_1, B_2$. The measurement performed by $N_1$ is then described by the POVM consisting of orthogonal projectors $\{\Pi_x = \ketbra{\psi_x}{\psi_x}_{A_1,B_1} \otimes I_{B_2} \}_{x \in \{0,1\}^{2n}}$. Let the outcome of this measurement on $\sigma_{A_1,B_1,B_2}$ be some string $z_1 = (a,b) \in \{0,1\}^{2n}$.

    Let us compute the post-measurement state conditioned on obtaining outcome $z_1$. For ease of exposition, we first assume $\rho$ is a pure state, i.e., we write $\rho = \ket{\phi}\bra{\phi}$ for some $\ket{\phi} \triangleq \sum_{i \in \{0,1\}^n} \alpha_i \ket{i} \in \mathbb{C}^{2^n}$ and amplitudes $\alpha_i \in \mathbb{C}$.
    Now, to compute the post-measurement states, let us compute $\Pi_{z_1} (\ket{\phi}_{A_1} \otimes \ket{\epr_n}_{B_1,B_2})$.
    \begin{align}
         \Pi_{z_1} (\ket{\phi}_{A_1} \otimes \ket{\epr_n}_{B_1,B_2}) &= (\ketbra{\psi_{z_1}}{\psi_{z_1}}_{A_1,B_1} \otimes I_{B_2}) (\ket{\phi}_{A_1} \otimes \ket{\epr_n}_{B_1,B_2})
         \\&= \frac{(-i)^{a.b}}{2^n} \sum_j (-1)^{j.b} (\ketbra{\psi_{z_1}}{j \oplus a,j}_{A_1,B_1} \otimes I_{B_2}) \sum_{k, l} \alpha_k \ket{k,l,l}_{A_1,B_1,B_2}
         \\&= \frac{(-i)^{a.b}}{2^n} \ket{\psi_{z_1}}_{A_1,B_1} \sum_{j,k,l} (-1)^{j.b} \alpha_k \delta_{j \oplus a, k} \delta_{j,l} \ket{l}_{B_2}
         \\&= \frac{1}{2^n} \ket{\psi_{z_1}}_{A_1,B_1} (-i)^{a.b} \sum_{k} (-1)^{k.b} (-1)^{a.b} \alpha_{k} \ket{k \oplus a}_{B_2}
         \\&= \frac{1}{2^n} \ket{\psi_{z_1}}_{A_1,B_1} \otimes W_{z_1} \ket{\phi}_{B_2}.
     \end{align}
     Thus, for a generic mixed state, by linearity, we can write
     \begin{equation}
         \Pi_{z_1} (\rho \otimes \ket{\epr_n}\bra{\epr_n}_{B_1,B_2}) \Pi_{z_1} = \frac{1}{4^n} \ket{\psi_{z_1}}\bra{\psi_{z_1}}_{A_1,B_1} \otimes W_{z_1} \rho W_{z_1}^\dag.
     \end{equation}
     Now, we clearly have $\Tr(\Pi_{z_1} (\rho \otimes \ket{\epr_n}\bra{\epr_n}_{B_1,B_2})) = \frac{1}{4^n}$, i.e., each outcome $z_1$ occurs with probability $\frac{1}{4^n}$, and the post-measurement state conditioned on measuring $z_1$ is $\ket{\psi_{z_1}}\bra{\psi_{z_1}}_{A_1,B_1} \otimes W_{z_1} \rho W_{z_1}^\dag$.

    Thus, conditioned on $z_1$, $N_2$'s mixed state is given by $W_{z_1} \rho W_{z_1}^\dagger \otimes \rho$. Measuring this state in the Bell basis, $N_2$ obtains outcome $z_2$ with probability
    \begin{align}
        \Pr(z_2 | z_1) &= \bra{\psi_{z_2}} W_{z_1} \rho  W_{z_1}^\dagger \otimes \rho \ket{\psi_{z_2}}
        \\&=  \bra{\psi_{z_1 \oplus z_2}} \rho \otimes \rho \ket{\psi_{z_1 \oplus z_2}}
        \\&= p_\rho(z_1 \oplus z_2).
    \end{align}
    Finally, the probability of $N_c$ outputting a string $z$ is given by
    \begin{align}
        \Pr(z) &= \sum_{z_1,z_2 : z_1 \oplus z_2 = z}  \Pr(z_1) \Pr(z_2| z_1) 
        \\&= \sum_{z_1,z_2 : z_1 \oplus z_2 = z} \frac{1}{4^n} p_\rho(z_1 \oplus z_2)
        \\&= p_\rho(z).
    \end{align}
    Thus, the described algorithm produces a single sample from the distribution $p_\rho$, as claimed.     
\end{proof}
\end{document}